%% file: main.tex
\documentclass[acmtocl,acmnow]{acmtrans2m}
\newtheorem{theorem}{Theorem}[section]

\newtheorem{proposition}[theorem]{Proposition}
\newtheorem{lemma}[theorem]{Lemma}
\newdef{definition}[theorem]{Definition}
\newdef{remark}[theorem]{Remark}
\newdef{example}[theorem]{Example}

\input{header}

\title{Least and Greatest Fixed Points in Linear Logic}
\author{David Baelde \\ University of Minnesota}

\begin{abstract}
The first-order theory of MALL (multiplicative, additive linear logic)
over only equalities is a well-structured but weak logic since it cannot
capture unbounded (infinite) behavior.  
Instead of accounting for
unbounded behavior via the addition of the exponentials ($\oc$ and
$\wn$), we add least and greatest fixed point operators. 
The resulting logic, which we call \mumall, satisfies two fundamental
proof theoretic properties: we establish weak normalization for it,
and we design a focused proof system that we prove complete
with respect to the initial system.
That second result provides a strong normal form
for cut-free proof structures that can be used, for example, to help
automate proof search.
We show how these foundations can be applied to intuitionistic logic.
\end{abstract}
\category{F.4.1}{Mathematical Logic and Formal Languages}{%
  Mathematical Logic}[Proof theory]
\category{F.3.1}{Logics and Meanings of Programs}{Specifying and
  Verifying and Reasoning about Programs}[Specification techniques]
\category{F.3.3}{Logics and Meanings of Programs}{Studies of
  Program Constructs}[Program and recursion schemes]
\terms{Design, Theory, Verification}
\keywords{fixed points, linear logic, (co)induction, recursive definitions,
  cut elimination, normalization, focusing, proof search}

\begin{document}
\begin{bottomstuff} 
This work has been supported in part by
the National Science Foundation grant CCF-0917140
and by
INRIA through the ``Equipes Associ{\'e}es'' Slimmer.
Opinions, findings, and conclusions or recommendations
expressed in this paper are those of the author
and do not necessarily reflect the views of the National Science Foundation.
\end{bottomstuff}
\maketitle

\section{Introduction}
\input{intro}

\section{\mumall} \label{sec:mumall}
\input{mumall}

\section{Normalization} \label{sec:norm}
\input{norm}

\section{Focusing} \label{sec:focusing}
\input{focus}

\section{Conclusion}
\input{conclu}

\begin{acks}
This paper owes a lot to the anonymous reviewers of an earlier version,
and I thank them for that.
I also wish to thank Dale Miller with whom I started this work,
Olivier Laurent and Alexis Saurin for their insights on focusing,
and Pierre Clairambault, St\'ephane Gimenez, Colin Riba
and especially Alwen Tiu for
helpful discussions on normalization proofs.
\end{acks}

\bibliographystyle{acmtrans}
%\bibliography{../../references/master.bib}

\input{main.bbl}
\begin{received}
Received October 2009;
revised July 2010;
accepted September 2010
\end{received}

\end{document}

%% file: header.tex
\newcommand{\resp}{resp.~}
\newcommand{\vs}{vs.~}
\newcommand{\WN}{{\cal WN}}

\usepackage{proof}

\usepackage{url}

\usepackage{rotating}
\usepackage{multirow}
\usepackage[all]{xy}

\usepackage{amsfonts,amsmath,amssymb}
\usepackage{cmll}
\newcommand{\lltens}{\mathrel{\otimes}}
\newcommand{\llone}{\mathbf{1}}
\newcommand{\llzero}{\mathbf{0}}
\newcommand{\llplus}{\mathrel{\oplus}}
\newcommand{\llpar}{\parr}
\newcommand{\llwith}{\with}
\newcommand{\llimp}{\multimap}
\newcommand{\lleq}{\multimapboth}

\newcommand{\freeze}[1]{(#1)^*}

\renewcommand{\H}{{\cal H}}
\newcommand{\G}{{\cal G}}

\newcommand{\x}{\vec{x}}
\newcommand{\y}{\vec{y}}

\renewcommand{\t}{\vec{t}}
\newcommand{\B}{\overline{B}}

\renewcommand{\|}{\; | \;}

\newcommand{\enc}[1]{[#1]}

\newcommand{\mmumall}{\mu\mbox{MALL}}
\newcommand{\mumall}{$\mmumall$}
\newcommand{\mmuLJ}{\mu\mbox{LJ}}
\newcommand{\mmuLJL}{\mu\mbox{LJL}}

\newcommand{\mmuLL}{\mu\mbox{LL}}

\newcommand{\muLL}{$\mmuLL$}
\newcommand{\muLJ}{$\mmuLJ$}
\newcommand{\muLJL}{$\mmuLJL$}

\newcommand{\eqdef}{\stackrel{def}{=}}
\newcommand{\unif}{\stackrel{.}{=}}

\newcommand{\ra}{\rightarrow}
\newcommand{\ie}{\emph{i.e.,}}
\newcommand{\eg}{\emph{e.g.,}}

\newcommand{\half}{\hbox{\em half}}

\newcommand{\set}[2]{\{\; #1 \;:\; #2 \;\}}

\newcommand{\Set}[2]{%
  \left\{\;\begin{array}{ccc} #1 &\; : \;& #2 \end{array}\;\right\}}
\newcommand{\All}[1]{\left\{\;\begin{array}{c} {#1} \end{array}\;\right\}}

\usepackage{disp}
\newcommand{\disp}[2]{\dispop[\downarrow]{\longrightarrow}{#1}{#2}}

\newcommand{\interp}[1]{[#1]}
\newcommand{\E}{{\cal E}}

%% file: intro.tex
% It should not sound "only focusing"
% it should not sound "too much proof search"

%% Motivation

Inductive and coinductive definitions are ubiquitous in mathematics and
computer science, from arithmetic to operational semantics and concurrency
theory. These recursive definitions provide natural and very expressive
ways to write specifications.
The primary means of reasoning on inductive specifications is by induction,
which involves the generalization of the tentative theorem
in a way that makes it \emph{invariant}
under the considered inductive construction.
Although the invariant might sometimes be the goal itself,
it can be very different in general,
sometimes involving concepts that are absent from the theorem statement.
When proving theorems, most of the ingenuity actually goes into
discovering invariants.
Symmetrically,
proving coinductive specifications is done by coinduction,
involving \emph{coinvariants} which again can have little to do with
the initial specification.
A proof theoretical framework supporting (co)inductive definitions
can be used as a foundation for prototyping, model checking
and reasoning about many useful computational systems.
But that great expressive power comes with several difficulties such
as undecidability, and even \emph{non-analyticity}:
because of (co)induction rules and their arbitrary (co)invariants,
proofs do not enjoy any reasonable form of subformula property.
Nevertheless, we shall see that modern proof theory provides useful tools
for understanding least and greatest fixed points
and controlling the structure of proofs involving those concepts.

%% Focusing

% The two fundamental properties
Arguably, the most important property of a logic is its consistency.
In sequent calculus, consistency is obtained from cut elimination,
which requires a symmetry between one connective and its dual,
or in other words between construction and elimination,
conclusion and hypothesis.
The notions of polarity and focusing are more recent in proof theory
but their growing importance puts them on par with cut elimination.
Focusing organizes proofs in stripes of \emph{asynchronous}
and \emph{synchronous} rules,
removing irrelevant interleavings and inducing a reading of the logic
based on macro-connectives aggregating stripes of usual connectives.
Focusing is useful to justify game theoretic
semantics~\cite{miller06mfps,delande08lics,delande10apal}
and has been central to the design of Ludics \cite{girard01mscs}.
% Their impact on proof search
From the viewpoint of proof search, % cut elimination usually implies
% that there is no need to invent lemmas.
focusing plays the essential role of reducing the space of the search
for a cut-free proof, by identifying situations when backtracking
is unnecessary.
\index{Logic programming}
In logic programming, it plays the more demanding role of correlating
the declarative meaning of a program with its operational meaning, given
by proof search.
Various computational systems have employed different focusing
theorems: much of Prolog's design and implementations can be justified
by the completeness of SLD resolution \cite{apt82jacm}; uniform proofs
(goal-directed proofs) in intuitionistic and intuitionistic linear
logics have been used to justify $\lambda$Prolog \cite{miller91apal}
and Lolli \cite{hodas94ic}; the classical linear logic programming
languages LO \cite{andreoli91ngc}, Forum \cite{miller96tcs} 
and the inverse method \cite{chaudhuri05csl} have
used directly Andreoli's general focusing result \cite{andreoli92jlc}
for linear logic.
% The problem with fixed points
In the presence of fixed points,
proof search becomes particularly problematic
since cut-free derivations are not analytic anymore.
Many systems use various heuristics to restrict the search space,
but these solutions lack a proof theoretical justification.
In that setting, focusing becomes especially interesting, as
it yields a restriction of the search space while preserving completeness.
Although it does not provide a way to decide the undecidable,
focusing brings an appreciable leap forward,
pushing further the limit where proof theory and completeness leave place to
heuristics.

%% Linear logic

In this paper, we propose a fundamental proof theoretic study of the notions
of least and greatest fixed point. By considering fixed points as primitive
notions rather than, for example, encodings in second-order logic,
we shall obtain strong results about the structure of their proofs.
We introduce the logic \mumall\ which extends the multiplicative and additive
fragments of linear logic (MALL) with least and greatest fixed points
and establish its two fundamental properties, \ie\ cut elimination
and focusing.
There are several reasons to consider linear logic.
First, its classical presentation allows us to internalize
the duality between least and greatest fixed point operators,
obtaining a simple, symmetric system.
Linear logic also allows the independent study of fixed points
and exponentials, two different approaches to infinity.
Adding fixed points to linear logic without exponentials
yields a system where they are the only source of infinity;
we shall see that it is already very expressive.
Finally, linear logic is simply a decomposition of intuitionistic
and classical logics~\cite{girard87tcs}.
Through this decomposition, the study of linear logic has brought a lot
of insight to the structure of those more common systems.
In that spirit, we provide in this paper some foundations that have
already been used in more applied settings.

%% Previous and related work

The logic \mumall\ was initially designed as an elementary system
for studying the focusing of logics supporting 
(co)inductive definitions~\cite{momigliano03types};
leaving aside the simpler underlying propositional layer (MALL instead of LJ),
fixed points are actually more expressive than this notion of definition
since they can express mutually recursive definitions.
But \mumall\ is also relatively close to type theoretical systems involving
fixed points~\cite{mendler91apal,matthes98csl}.
The main difference is that our logic is a first-order one,
although the extension to second-order would be straightforward
and the two fundamental results would extend smoothly.
Inductive and coinductive definitions have also been approached
by means of cyclic proof
systems~\cite{santocanale01brics,brotherston05tableaux}.
These systems are conceptually appealing, but generally weaker in a
cut-free setting;
some of our earlier work~\cite{baelde09tableaux} addresses
this issue in more details.

There is a dense cloud of work related to \mumall.
Our logic and its focusing have been used to revisit the foundations
of the system Bedwyr~\cite{baelde07cade},
a proof search approach to model checking.
A related work \cite{baelde09tableaux} carried out in \mumall{}
establishes a completeness result for inclusions of finite automata
leading to an extension of cyclic proofs.
The treatment of fixed points in \mumall, as presented in this paper,
can be used in full linear logic (\muLL) and intuitionistic logic (\muLJ).
\muLL\ has been used to encode and reason about
various sequent calculi~\cite{miller.ep}.
\muLJ\ has been given a game semantics~\cite{clairambault09fossacs}, and
has been used in the interactive
theorem prover \texttt{Tac} where focusing provides
a foundation for automated (co)inductive theorem proving~\cite{baelde10ijcar},
and in~\cite{nigam09phd} to extend a logical approach
to tabling~\cite{miller07csla} where focusing is used to avoid
redundancies in proofs.
Finally, those logics have also been extended with (minimal) generic 
quantification~\cite{miller05tocl,baelde08lfmtp},
which fully enables reasoning in presence of variable binding,
\eg\ about operational semantics, logics or type systems.

The rest of this paper is organized as follows.
In Section~\ref{sec:mumall}, we introduce the logic,
provide a few examples and study its basic proof theory.
Section~\ref{sec:norm} establishes cut elimination for \mumall,
by adapting the candidates of reducibility argument
to obtain a proof of weak normalization.
Finally, we investigate the focusing of \mumall\ in Section~\ref{sec:focusing},
and present a simple application to intuitionistic logic.

%% file: mumall.tex
We assume some basic knowledge of simply-typed 
$\lambda$-calculus~\cite{barendregt97handbook}
which we leverage as a representation framework,
following Church's approach to syntax.
This allows us to
consider syntax at a high-level, modulo $\alpha\beta\eta$-conversion.
In this style,
we write $P x$ to denote a formula from which $x$ has been totally abstracted
out ($x$ does not occur free in $P$), so that $P t$ corresponds
to the substitution of $x$ by $t$,
and we write $\lambda x. P$ to denote a vacuous abstraction.
\emph{Formulas} are objects of type $o$,
and the syntactic variable $\gamma$ shall represent a term type, \ie\
any simple type that does not contain $o$.
A \emph{predicate} of arity $n$ is an object of type
$\gamma_1\ra\ldots\ra\gamma_n\ra o$,
and a \emph{predicate operator} (or simply \emph{operator})
of first-order arity $n$ and second-order arity $m$ is an object of type
$\tau_1\ra\ldots\ra\tau_m\ra\gamma_1\ra\ldots\ra\gamma_n\ra o$
where the $\tau_i$ are predicate types of arbitrary arity.
We shall see that the term language can in fact be chosen quite freely:
for example terms might be first-order, higher-order,
or even dependently typed, as long as equality and substitution are defined.

We shall denote terms by $s,t$; 
formulas by $P,Q$; operators by $A,B$;
term variables by $x,y$;
predicate variables by $p, q$;
and atoms (predicate constants) by $a,b$.
The syntax of \mumall\ formulas is as follows:
\begin{eqnarray*}
P &::=&
P \lltens P \| P \llplus P \| P \llpar P \| P \llwith P
\| \llone \| \llzero \| \bot \| \top \| a\t \| a^\perp\t \\
& \| & \exists_\gamma x.~ P x \| \forall_\gamma x.~ P x
  \| s =_\gamma t
  \| s \neq_\gamma t \\
& \| & \mu_{\gamma_1\ldots\gamma_n}(\lambda p \lambda \x.~ P p \x) \t
  \| \nu_{\gamma_1\ldots\gamma_n}(\lambda p \lambda \x.~ P p \x)\t
  \| p\t \| p^\perp\t
\end{eqnarray*}
The quantifiers have type $(\gamma\rightarrow o)\rightarrow o$ and the
equality and \emph{disequality} (\ie\ $\neq$)
have type $\gamma\rightarrow\gamma\rightarrow o$.
The connectives $\mu$ and $\nu$ have type
$(\tau \ra \tau) \ra \tau$ where $\tau$ is $\gamma_1\ra\cdots\ra\gamma_n\ra o$
for some arity $n\ge 0$.
We shall almost always elide the references to $\gamma$,
assuming that they can be determined from the context when it
is important to know their value.
Formulas with top-level connective $\mu$ or $\nu$
are called \emph{fixed point expressions} and can be arbitrarily \emph{nested}
(such as in $\nu (\lambda p.~ p \lltens
   \mu (\lambda q.~ \llone \llplus a \lltens q))$, written
$\nu p.~ p \lltens (\mu q.~ \llone \llplus a\lltens q)$ for short)
and \emph{interleaved}
(\eg\ $\mu p.~ \llone \llplus \mu q.~ \llone \llplus p \lltens q$).
Nested fixed points correspond to iterated (co)inductive definitions
while interleaved fixed points correspond to mutually (co)inductive
definitions, with the possibility of simultaneously defining
an inductive and a coinductive.

Note that negation is not part of the syntax of our formulas,
except for atoms and predicate variables.
This is usual in classical frameworks, where negation is instead
defined as an operation on formulas.

\begin{definition}[Negation ($P^\perp$,$\B$)]
\emph{Negation} is the involutive operation on formulas satisfying the 
following equations:
\[ \begin{array}{rclp{1cm}rcll}
  (P \llpar Q)^\bot &\equiv& P^\bot \lltens Q^\bot
  & ~ &
  (P \llwith Q)^\bot &\equiv& P^\bot \llplus Q^\bot
  & \\
  \bot^\bot &\equiv& \llone
  & ~ &
  \top^\bot &\equiv& \llzero
  & \\
  (s = t)^\bot &\equiv& s \neq t
  & ~ &
  (\forall x.~ P x)^\bot &\equiv& \exists x.~ (P x)^\bot
  & \\
  (\nu{}B\t)^\bot &\equiv& \mu\B\t
  & ~ &
  (a^\perp\t)^\perp &\equiv& a\t & (p^\perp\t)^\perp \equiv p\t
  \\
  \B &\equiv&
  \multicolumn{5}{l}{
    \lambda p_1 \ldots \lambda p_m
    \lambda x_1 \ldots \lambda x_n.~
       (B p_1^\perp\ldots p_m^\perp x_1 \ldots x_n)^\bot
  }
  & \mbox{for operators}
  \\
  P^\perp &\equiv&
  \multicolumn{5}{l}{
    \lambda x_1\ldots\lambda x_n.~ (P x_1\ldots x_n)^\perp
  }
  & \mbox{for predicates}
  \end{array}
\]
An operator $B$ is said to be \emph{monotonic} when
it does not contain any occurrence of a negated predicate variable.
We shall write
$P\llimp Q$ for $P^\bot \llpar Q$, and
$P\lleq Q$ for $(P\llimp Q)\llwith (Q\llimp P)$.
\end{definition}

We shall assume that \emph{all predicate operators are monotonic},
and do not have any free term variable.
By doing so, we effectively exclude negated predicate variables $p^\perp$
from the logical syntax; they are only useful as intermediate devices when
computing negations.

\begin{example}
We assume a type $n$ and two constants $0$ and $s$
of respective types $n$ and $n\ra n$.
The operator $(\lambda{}p\lambda{}x.~ x=0 \llplus \exists{}y.~ x=s~(s~y)
\lltens p~y)$ whose least fixed point describes even numbers is
monotonic, but
$(\lambda p\lambda x.~ x=0 \llplus \exists y.~ x=s~y \lltens (p~y
  \llimp\llzero))$ is non-monotonic because of the occurrence of $p^\perp y$
that remains once the definition of $\llimp$ has been expanded
and negations have been computed.
\end{example}

A \emph{signature}, denoted by $\Sigma$, is a list of distinct typed variables.
We write $\Sigma\vdash t:\gamma$ when $t$ is a well-formed term of type
$\gamma$ under the signature $\Sigma$; we shall not detail how
this standard judgment is derived.
A \emph{substitution} $\theta$ consists of a domain signature $\Sigma$,
an image signature $\Sigma'$, and a mapping from each $x:\gamma$ in $\Sigma$
to some term $t$ of type $\gamma$ under $\Sigma'$.
We shall denote the image signature $\Sigma'$ by $\Sigma\theta$.
Note that we do not require each variable from $\Sigma\theta$ to be
used in the image of $\Sigma$: for example, we do consider the substitution
from $\Sigma$ to $(\Sigma,x)$ mapping each variable in $\Sigma$ to its
counterpart in the extended signature.
If $\Sigma\vdash t:\gamma$, then $t\theta$ denotes the result of substituting
free variables in $t$ by their image in $\theta$,
and we have $\Sigma\theta\vdash t\theta:\gamma$.

Our \emph{sequents} have the form $\Sigma;\vdash\Gamma$ where the signature
$\Sigma$ denotes universally quantified terms\footnote{
  Term constants and atoms are viewed as being introduced,
  together with their types, in an external, toplevel signature that is
  never explicitly dealt with.
  Predicate variables are not found in either of those signatures;
  they cannot occur free in sequents.
},
and $\Gamma$ is a multiset of formulas, \ie\
expressions of type $o$ under $\Sigma$.
Here, we shall make an exception to the higher-order abstract syntax notational
convention: when we write $\Sigma;\vdash\Gamma$ using the metavariable
$\Sigma$ (\ie\ without detailing the contents of the signature)
we allow variables from $\Sigma$ to occur in $\Gamma$.
It is often important to distinguish different occurrences of a formula
in a proof, or track a particular formula throughout a proof;
such distinctions are required
for a meaningful computational interpretation of cut elimination,
and they also play an important role in our focusing mechanisms.
In order to achieve this, we shall use the notion of \emph{location}.
From now on, we shall consider a formula not only as the structure
that it denotes, namely an abstract syntax tree,
but also as an address where this structure is located.
Similarly, subformulas have their own locations,
yielding a \emph{tree of locations} and sublocations.
We say that two locations are \emph{disjoint} when they do not share
any sublocation.
Locations are independent of the term structure of formulas:
all instantiations of a formula have the same location,
which amounts to say that locations are attached to formulas
abstracted over all terms.
We shall not provide a formal definition of locations,
which would be rather heavy,
but a high-level description should give a sufficient understanding
of the concept.
A formal treatment of locations can be found in \cite{girard01mscs},
and locations can also be thought of as denoting nodes in proof nets
or variable names in proof terms.
% NOTE this is because we want balancing to be preserved by instantiation,
%   when commuting some \neq rules
Locations allow us to make a distinction between \emph{identical} formulas,
which have the same location,
and \emph{isomorphic} formulas which only denote the same structure.
When we talk of the \emph{occurrences} of a formula in a proof,
we refer to identical formulas occurring at different places in that 
derivation.
We shall assume that formulas appearing in a sequent have
\emph{pairwise disjoint locations}. In other words,
sequents are actually sets of formulas-with-location,
which does not exclude that a sequent can contain several isomorphic
formulas.

We present the inference rules for \mumall{} in Figure~\ref{fig:mumall}.
Rules which are not in the identity group are called \emph{logical rules},
and the only formula whose toplevel connective is required for the application
of a logical rule is said to be \emph{principal} in that rule application.
In the $\neq$ rule, $\theta$ is a substitution of domain $\Sigma$
ranging over universal variables,
$\Gamma\theta$ is the result of applying that substitution to
every term of every formula of $\Gamma$.
In the $\nu$ rule, which provides both induction and coinduction,
$S$ is called the (co)invariant,
and is a closed formula of the same type as $\nu{}B$,
of the form $\gamma_1\ra\cdots\ra\gamma_n\ra o$.
We shall adopt a proof search reading of derivations: for instance,
we call the $\mu$ rule ``unfolding'' rather than ``folding'',
and we view the rule whose conclusion is the conclusion of a
derivation as the first rule of that derivation.

Inference rules should be read from the locative viewpoint,
which we illustrate with a couple of key examples.
In the $\forall$ and $\exists$ rules, the premise and conclusion sequents
only differ in one location: the principal location is replaced by its 
only sublocation.
The premise sequents of the $\neq$ rule are locatively identical
to the conclusion sequent,
except for the location of the principal $\neq$ formula that has been removed.
Similarly in the $\llwith$ rule,
the formulas of the context $\Gamma$ are copied in the two premises,
each of them occurring (identically) three times in the rule.
% \[ \infer{\vdash \Gamma, P\llwith Q}{\vdash \Gamma, P & \vdash \Gamma, Q} \]
In the axiom rule, the two formulas are locatively distinct but
have dual structures.
% \[ \infer{\vdash (\mu B \t)^\bot, \mu B \t}{} \]
%
In the $\nu$ rule, the formulas from the co-invariance proofs have \emph{new}
locations, as well as the co-invariant in the main premise.
This means that these locations can be changed at will,
much like a renaming of bound variables.
A greatest fixed point has infinitely many sublocations,
regarding the coinvariants as its subformulas.
In the $\mu$ rule, the formula $B(\mu B)\t$ is the only sublocation
of the least fixed point.
Distinct occurrences of $\mu B$ in $B (\mu B)$ (\resp $\nu B$ in $B(\nu B)$)
have distinct locations, so that the graph of locations remains a tree.
It is easy to check that inference rules preserve the fact that sequents
are made of disjoint locations.

\begin{figure}[t]
\begin{center}
\[
  \begin{array}{c}
  \mbox{Identity group}
  \\[6pt]
  \infer[cut]{\Sigma; \vdash \Gamma, \Delta}{
     \Sigma; \vdash \Gamma, P^\perp &
     \Sigma; \vdash P, \Delta}
  \qquad
  \infer[init]{\Sigma; \vdash P, P^\bot}{}
  \end{array}
\]
\[ \addtolength{\arraycolsep}{18pt}
 \begin{array}{c}
  \mbox{MALL rules} \\[6pt]
  \infer[\bot]{\Sigma; \vdash \Gamma, \bot}{
         \Sigma; \vdash \Gamma}
  \quad
  \infer[\llpar]{\Sigma; \vdash \Gamma, P\llpar Q}{
         \Sigma; \vdash \Gamma, P, Q}
  \quad
  \infer[\lltens]{\Sigma; \vdash \Gamma,\Delta, P\lltens Q}{
         \Sigma; \vdash \Gamma, P &
         \Sigma; \vdash \Delta, Q}
  \quad
  \infer[\llone]{\Sigma; \vdash \llone}{}
  \\[6pt]
  \infer[\top]{\Sigma; \vdash \Delta, \top}{}
  \quad
  \infer[\llwith]{\Sigma; \vdash \Gamma, P \llwith Q}{
         \Sigma; \vdash \Gamma, P & \vdash \Gamma, Q}
  \quad
  \infer[\llplus]{\Sigma;\vdash \Gamma, P_0\llplus P_1}{
         \Sigma; \vdash \Gamma, P_i}
 \end{array}
\]
\[ \addtolength{\arraycolsep}{18pt}
\begin{array}{c}
\mbox{First-order structure} 
  \\[6pt]
  \infer[\forall]{\Sigma ; \vdash \Gamma, \forall_\gamma x. P x}{
         \Sigma, x:\gamma ; \vdash \Gamma, P x}
  \quad
  \infer[\exists]{\Sigma; \vdash \Gamma, \exists_\gamma x. Px}{
         \Sigma\vdash t:\gamma & \Sigma; \vdash \Gamma, Pt}
  \\[6pt]
  \infer[\neq]{\Sigma; \vdash \Gamma, s\neq t}{
    \set{\Sigma\theta; \vdash \Gamma\theta}{\theta\in csu(s\unif t)}}
  \quad
  \infer[=]{\Sigma; \vdash t=t}{}
\end{array} \]
\[
  \begin{array}{c}
  \mbox{Fixed points}
  \\[6pt]
  \infer[\nu]{\Sigma ; \vdash \Gamma, \nu{}B\t}{
    \Sigma ; \vdash \Gamma, S\t &
    \x ; \vdash BS\x, (S\x)^\bot}
  \qquad
  \infer[\mu{}]{\Sigma; \vdash \Gamma, \mu{}B\t}{
                \Sigma; \vdash \Gamma, B(\mu{}B)\t}
  \end{array}
\]
\end{center}
\caption{Inference rules for first-order \mumall}
\label{fig:mumall}
\end{figure}

Note that \mumall\ is a \emph{conservative} extension of MALL, meaning that
a MALL formula is derivable in MALL if and only if it is derivable in \mumall:
it is easy to check that MALL and \mumall\ have the same
cut-free derivations of MALL formulas,
and cut elimination will allow us to conclude.

In the following, we use a couple of notational shortcuts.
For conciseness, and when it does not create any ambiguity,
we may use ${\bullet}$ to denote implicitly abstracted variables,
\eg\ $P{\bullet}x$ denotes $\lambda y. P y x$.
Similarly, we may omit abstractions, \eg\ $\perp$ used as a coinvariant
stands for $\lambda\x.\perp$ and, when $S_1$ and $S_2$ are predicates
of the same type, $S_1\llpar S_2$ stands for
$\lambda \x.~ S_1\x\llpar S_2\x$.
Finally, we shall omit the signature of sequents whenever unambiguous,
simply writing $\;\vdash \Gamma$.

\subsection{Equality} %%%%%%%%%%%%%%%%%%%%%%%%%%%%%%%%%%%%%%%%%%%%%%%%%%%%%%%

The treatment of equality dates back
to~\cite{girard92mail,schroeder-Heister93lics}, originating from
logic programming.
In the disequality rule, which is a case analysis on all unifiers,
$csu$ stands for \emph{complete set of unifiers}, that is a set ${\cal S}$
of unifiers of $u\unif v$ such that any unifier $\sigma$ can be written
as $\theta\sigma'$ for $\theta\in{\cal S}$.
For determinacy reasons,
we assume a fixed mapping from unification problems to complete
sets of unifiers, always taking $\{id\}$ for $csu(u\unif u)$.
Similarly, we shall need a fixed mapping from each unifier
$\sigma' \in csu(u\theta\unif v\theta)$ to a $\sigma \in csu(u\unif v)$
such that $\theta\sigma' = \sigma\theta'$ for some $\theta'$
| existence is guaranteed since
  $\theta\sigma'$ is a unifier of $u\unif v$.
%%
% NOTE minimality was useless except for id (and it was a loss of generality)
% Without loss of generality,
% we shall only consider minimal sets, \ie\ $csu(u\unif v)$ should never
% contain $\theta$ and $\sigma$ such that $\sigma = \theta\sigma'$,
% and we assume that the identity substitution $id$ always belongs to
% $csu(u\unif u)$.
%
In the first-order case, and in general when most general unifiers exist,
the $csu$ can be restricted to having at most one element.
But we do not rule out higher-order terms,
for which unification is undecidable and complete sets of unifiers
can be infinite~\cite{huet75tcs}
| in implementations, we restrict to well-behaved fragments
such as higher-order patterns~\cite{miller92jsc}.
Hence, the left equality rule might be infinitely branching. But
derivations remain inductive structures (they don't have infinite
branches) and are handled naturally in our proofs by means
of (transfinite) structural induction.
Again, the use of higher-order terms, and even the presence of the equality
connectives are not essential to this work. All the results presented below 
hold in the logic without equality, and do not make much assumptions on 
the language of terms.

It should be noted that our ``free'' equality is more powerful than
the more usual Leibniz equality. Indeed, it implies the injectivity
of constants: one can prove for example that $\forall x.~ 0 = s~x \llimp 
\llzero$ since there is no unifier for $0\unif s~x$.
This example also highlights that constants and universal variables
are two different things,
since only universal variables are subject to unification |
which is why we avoid calling them eigenvariables.
It is also important to stress that the disequality rule does not
and must not embody any assumption about the signature,
just like the universal quantifier.
That rule enumerates substitutions over open terms,
not instantiations by closed terms.
Otherwise, with an empty domain we would prove
$\forall x.~ x=x \llimp \llzero$ (no possible instantiation for $x$)
and $\forall x.~ x=x$, but not (without cut) $\forall x.~ \llzero$.
Similarly,
by considering a signature with a single constant $c:\tau_2$,
so that $\tau_1$ is empty while $\tau_1\ra\tau_2$ contains only $\lambda x.~ 
c$, we would indeed be able to prove $\forall x.~ x=x$ and
$\forall x.~ x=x \llimp \exists y.~ x=\lambda a.~ y$
but not (without cut) $\forall x \exists y.~ x=\lambda a.~ y$.

\begin{example}
Units can be represented by means of $=$ and $\neq$.
Assuming that $2$ and $3$ are two distinct constants, then we have
$2=2 \lleq \llone$ and $2=3 \lleq \llzero$
(and hence $2\neq 2 \lleq \bot$ and $2\neq 3 \lleq \top$).
\end{example}

\subsection{Fixed points} %%%%%%%%%%%%%%%%%%%%%%%%%%%%%%%%%%%%%%%%%%%%%%%%

Our treatment of fixed points follows from a line of work on definitions
\cite{girard92mail,schroeder-Heister93lics,mcdowell00tcs,momigliano03types}.
In order to make that lineage explicit and help the understanding of our rules,
let us consider for a moment an intuitionistic framework (linear or not).
In such a framework, the rules associated with least fixed points can be derived
from Knaster-Tarski's characterization of an operator's least fixed point in
complete lattices: it is the least of its pre-fixed points\footnote{
  Pre-fixed points of $\phi$ are those $x$ such that 
    $\phi(x) \leq x$.
}.
\[ \infer{\Sigma; \mu B \t \vdash S \t}{ \x; BS\x\vdash S\x} \quad\quad
     \infer{\Sigma; \Gamma\vdash \mu B \t}{\Sigma; \Gamma \vdash B(\mu B)\t} \]
As we shall see,
the computational interpretation of the left rule is recursion.
Obviously, that computation cannot be performed without knowing
the inductive structure on which it iterates. In other words,
a cut on $S\t$ cannot be reduced until a cut on $\mu B\t$ is performed.
As a result, a more complex left introduction rule is usually considered
(\eg\ in \cite{momigliano03types})
which can be seen as embedding this suspended cut:
\[ \infer{\Sigma; \Gamma, \mu B\t \vdash P}{
     \Sigma; \Gamma, S\t \vdash P & \x; BS\x\vdash S\x} \quad\quad
   \infer{\Sigma; \Gamma\vdash \mu B \t}{\Sigma; \Gamma \vdash B(\mu B)\t} \]
Notice, by the way, how the problem of suspended cuts (in the first set of
rules) and the loss of subformula property (in the second one) relate
to the arbitrariness of $S$, or in other words the difficulty of
finding an invariant for proving $\Gamma, \mu B\t \vdash P$.

Greatest fixed points can be described similarly as the greatest
of the post-fixed points:
\[ \infer{\Sigma;\Gamma, \nu B\t\vdash P}{\Sigma;\Gamma, B(\nu B)\t\vdash P}
 \quad\quad
   \infer{\Sigma;\Gamma\vdash\nu B\t}{\Sigma;\Gamma\vdash S\t &
      \x;S\x\vdash BS\x}
\]

\begin{example} \label{example:nat}
Let $B_{nat}$ be the operator
  $(\lambda{}N\lambda{}x.~ x=0
          \llplus \exists{}y.~ x=s~y \lltens N~y)$
and $nat$ be its least fixed point $\mu B_{nat}$.
Then the following inferences can be derived from the above rules:
\[ \infer{\Sigma; \Gamma, nat~t \vdash P}{
      \Sigma; \Gamma, S~t \vdash P &
      \vdash S~0 &
      y; S~y \vdash S~(s~y)}
\quad\quad
   \infer{\Sigma; \Gamma \vdash nat~0}{}
\quad\quad
   \infer{\Sigma; \Gamma \vdash nat~(s~t)}{\Sigma; \Gamma\vdash nat~t}
\]
\end{example}

Let us now consider the translation of those rules to classical linear logic,
using the usual reading of $\Gamma\vdash P$ as
$\vdash \Gamma^\perp, P$ where $(P_1,\ldots,P_n)^\perp$ is
$(P_1^\perp,\ldots,P_n^\perp)$.
It is easy to see that the above right introduction rule for $\mu$
(\resp\ $\nu$) becomes the $\mu$ (\resp\ $\nu$) rule of 
Figure~\ref{fig:mumall}, by taking $\Gamma^\perp$ for $\Gamma$.
Because of the duality between least and greatest fixed points
(\ie\ $(\mu B)^\perp \equiv \nu \B$) the other rules collapse.
The translation of the above left introduction rule for $\nu$
corresponds to an application of the $\mu$ rule of \mumall\ on
$(\nu B\t)^\perp \equiv \mu \B\t$.
The translation of the left introduction rule for $\mu$ is as follows:
\[ \infer{\vdash \Gamma^\perp, (\mu B \t)^\bot, P}{
     \vdash \Gamma^\perp, S^\perp\t, P &
     \vdash (B S \x)^\perp, S \x
   } \]
Without loss of generality, we can write $S$ as $S'^\perp$.
Then $(B S\x)^\perp$ is simply $\B S'\x$
and we obtain exactly the $\nu$ rule of \mumall\ on $\nu\B$:
\[ \infer[\nu]{\vdash \Gamma^\perp, \nu \B \t, P}{
     \vdash \Gamma^\perp, S' \t, P &
     \vdash \B S' \x, S'^\perp \x
   } \]
In other words, by internalizing syntactically the
duality between least and greatest fixed points that exists in 
complemented lattices, we have also obtained the identification
of induction and coinduction principles.

\begin{example}
As expected from the intended meaning of $\mu$ and $\nu$,
$\nu{}(\lambda{}p.p)$ is provable
(take any provable formula as the coinvariant)
and its dual $\mu{}(\lambda{}p.p)$ is not provable.
More precisely, $\mu{}(\lambda{}p.p) \lleq \llzero$
and $\nu{}(\lambda{}p.p) \lleq \top$.
\end{example}

\subsection{Comparison with other extensions of MALL}

The logic \mumall\ extends MALL
with first-order structure ($\forall$, $\exists$, $=$ and $\neq$)
and fixed points ($\mu$ and $\nu$).
A natural question is whether fixed points can be compared
with other features that bring infinite behavior,
namely exponentials and second-order quantification.
% We do not discuss first-order connectives in this section, since
% they can be added independently of fixed points to MALL
% and its extensions.

In~\cite{baelde07lpar}, we showed that \mumall\ can be encoded
into full second-order linear logic (LL2), \ie\ MALL with exponentials
and second-order quantifiers,
by using the well-known second-order encoding:
\[ [\mu B\t] \equiv \forall S.~ !(\forall x.~ [B]S\x\llimp S\x) \llimp S\t \]
This translation highlights the fact that fixed points combine
second-order aspects (the introduction of an arbitrary (co)invariant)
and exponentials (the iterative behavior of the $\nu$ rule in
cut elimination).
The corresponding translation of \mumall\ derivations into LL2
is very natural |
anticipating the presentation of cut elimination
for \mumall, cut reductions in the original and encoded derivations
should even correspond quite closely.
We also provided a translation from LL2 proofs of encodings
to \mumall\ proofs, under natural constraints on second-order instantiations;
interestingly, focusing is used to ease this reverse translation.
% For least fixed points, once the universal quantification
% and the implication have been introduced, $\mu B$ is represented by the
% second-order universal variable $S$ and the $\mu$ rule corresponds to
% applying the pre-fixed point hypothesis $[B]S\x\llimp S\x$.
% For greatest fixed points, the dualization of the above encoding
% trivially allows to simulate the $\nu$ rule:
% $[\nu B\t] \equiv \exists S.~ !(\forall x.~ [\B]S\x\llimp S\x) \lltens S\t$.
% Those observations yield a soundness result: the encoding of provable
% \mumall\ statements is provable in LL2.
% Completeness is less simple.
% First, the application of a single $\mu$ or $\nu$ rule in \mumall{}
% corresponds to several rules in LL2,
% but there may be LL2 proofs which perform those sequences
% in a less organized way, for example interleaving several of them.
% More importantly, second-order order existential quantifiers in $[\nu B\t]$
% could be instantiated with formulas that are not encodings of
% \mumall\ coinvariants,
% which roughly corresponds to performing coinductions in \mumall\ with
% arbitrary LL2 coinvariants.
% Under the assumption that second-order instantiations are restricted
% to encodings of \mumall\ invariants, we proved completeness in
% \cite{baelde07lpar,baelde08phd}, by first focusing LL2 derivations to
% force that sequences of rules corresponding to $\mu$ and $\nu$ rules
% are applied consecutively.
% It is currently unknown if completeness holds
% without restricting second-order instantiations.

It is also possible to encode exponentials using fixed points,
as follows:
\[ [?P] \equiv \mu (\lambda p.~ \perp \llplus (p\llpar p) \llplus [P])
\quad\quad [!P] \equiv [?P^\perp]^\perp \]
This translation trivially allows to simulate the rules of weakening ($W$),
contraction ($C$) and dereliction ($D$) for $[?P]$ in \mumall:
each one is obtained by applying the $\mu$ rule and choosing the corresponding
additive disjunct.
Then, the promotion rule can be obtained for the dual of the encoding.
Let $\Gamma$ be a sequent containing only formulas of the form $[?Q]$,
and $\Gamma^\perp$ denote the tensor of the duals of those formulas,
we derive $\vdash\Gamma,[!P]$ from $\vdash\Gamma,[P]$
using $\Gamma^\perp$ as a coinvariant for $[!P]$:
\[ \infer[\nu]{\vdash \Gamma,
        \nu (\lambda p.~ \llone \llwith (p \lltens p) \llwith [P])}{
     \infer=[\lltens,init]{\vdash \Gamma,\Gamma^\perp}{} &
     \infer{\vdash \Gamma, \llone \llwith (\Gamma^\perp\lltens\Gamma^\perp)
                              \llwith [P]}{
        \infer=[W]{\vdash \Gamma, \llone}{\infer{\vdash\llone}{}} &
        \infer=[C]{\vdash \Gamma, \Gamma^\perp\lltens\Gamma^\perp}{
          \infer{\vdash \Gamma, \Gamma, \Gamma^\perp\lltens\Gamma^\perp}{
            \infer=[\lltens,init]{\vdash \Gamma, \Gamma^\perp}{} &
            \infer=[\lltens,init]{\vdash \Gamma, \Gamma^\perp}{}}} &
        \vdash \Gamma, [P]}} \]
Those constructions imply that the encoding of provable statements
involving exponentials is also provable in \mumall.
But the converse is more problematic:
not all derivations of the encoding can be translated
into a derivation using exponentials. Indeed, the encoding of
$[!P]$ is an infinite tree of $[P]$, and there is nothing
that prevents it from containing different proofs of $[P]$,
while $!P$ must be uniform, always providing the same proof of $P$.
Finally, accordingly with these different meanings, cut reductions
are different in the two systems.

It seems unlikely that second-order quantification can be encoded in
\mumall, or that fixed points could be encoded using only second-order
quantifiers or only exponentials. In any case, if such encodings existed
they would certainly be as shallow as the encoding of exponentials,
\ie\ at the level of provability,
and not reveal a connection at the level of proofs and cut elimination
like the encoding of fixed points in LL2.

\subsection{Basic meta-theory}

\begin{definition} \label{def:inst}
If $\theta$ is a term substitution, and $\Pi$ a derivation of
$\Sigma;\vdash\Gamma$, then we define $\Pi\theta$, a derivation of
$\Sigma\theta;\vdash \Gamma\theta$:
%\begin{longitem}
%\item
  $\Pi\theta$ always starts with the same rule as $\Pi$,
  its premises being obtained naturally by applying
  $\theta$ to the premises of $\Pi$.
%\item
  The only non-trivial case is the $\neq$ rule.
  Assuming that we have a derivation $\Pi$ where $u\neq v$ is principal,
  with a subderivation $\Pi_\sigma$ for each $\sigma\in csu(u\unif v)$,
  we build a subderivation of $\Pi\theta$ for each
  $\sigma'\in csu(u\theta\unif v\theta)$.
  Since $\theta\sigma'$ is a unifier for $u\unif v$,
  it can be written as $\sigma\theta'$ for some $\sigma\in csu(u\unif v)$.
  Hence, $\Pi_\sigma\theta'$ is a suitable derivation for $\sigma'$.
  Note that some $\Pi_\sigma$ might be unused in that process,
  if $\sigma$ is incompatible with $\theta$,
  while others might be used infinitely many times\footnote{
    Starting with a $\neq$ rule on $x\neq y~z$, which admits the most
    general unifier $[(y~z) / x]$, and applying the substitution
    $\theta = [u~v / x]$, we obtain $u~v \neq y~z$ which has no
    finite $csu$. In such a case, the infinitely many subderivations
    of $\Pi\theta$ would be instances of the only subderivation of $\Pi$.
  }.
%\end{longitem}
\end{definition}

Note that the previous definition encompasses common signature manipulations
such as permutation and extension, since it is possible for a substitution
to only perform a renaming, or to translate a signature to an extended one.

We now define functoriality, a proof construction that is used to
derive the following rule:
\[ \infer[B]{\Sigma ; \vdash B P, \B Q}{\x; \vdash P\x, Q\x} \]
In functional programming terms, it corresponds to a $map$ function:
its type is $(Q\llimp P) \llimp (BQ\llimp BP)$
(taking $Q^\perp$ as $Q$ in the above inference).
Functoriality is particularly useful for dealing with fixed points:
it is how we propagate reasoning/computation underneath $B$
\cite{matthes98csl}.

\begin{definition}[Functoriality, $F_B(\Pi)$]
Let $\Pi$ be a proof of $\x;\vdash P\x, Q\x$
and $B$ be a monotonic operator such
that $\Sigma\vdash B : (\vec{\gamma}\ra o)\ra o$.
We define $F_B(\Pi)$, a derivation of $\Sigma;\vdash BP, \B Q$,
by induction on the maximum depth of occurrences of $p$ in $B p$:
\begin{longitem}
\item When $B = \lambda p.~ P'$, $F_B(\Pi)$ is an instance of $init$ on $P'$.
  % \[ \infer{\Sigma; \vdash P', P'^\perp}{} \]
%  \textbf{We have a choice here}: either adopt the general rule,
%  or use the explicit expansion of the identities on fixed points |
%  the later would require to show later that those explicit identities
%  always belong to candidates, while it is trivial to see it for the
%  generalized rule.
\item When $B = \lambda p.~ p\t$, $F_B(\Pi)$ is $\Pi[\t/\x]$.
 % \[ \infer{\Sigma; \vdash P\t, Q\t}{\Pi[\t/\x]} \]
\item Otherwise,
  we perform an $\eta$-expansion based on the toplevel connective of $B$
  and conclude by induction hypothesis.
  We only show half of the connectives, because dual connectives
  are treated symmetrically.
  There is no case for units, equality and disequality
  since they are treated as part of the vacuous abstraction case.

  % The transformations are well-known for MALL and simple
  % for the first-order connectives.
  When $B = \lambda p.~ B_1 p \lltens B_2 p$:
  \[ \infer[\llpar]{\Sigma; \vdash B_1 P \lltens B_2 P, \B_1 Q\llpar \B_2 Q}{
     \infer[\lltens]{\Sigma; \vdash B_1 P \lltens B_2 P, \B_1 Q, \B_2 Q}{
            \infer{\Sigma; \vdash B_1 P, \B_1 Q}{F_{B_1}(\Pi)} &
            \infer{\Sigma; \vdash B_2 P, \B_2 Q}{F_{B_2}(\Pi)}}}
  \]

  When $B = \lambda p.~ B_1 p \llplus B_2 p$:
  \[ \infer[\llwith]{\Sigma; \vdash B_1 P \llplus B_2 P, \B_1 Q\llwith \B_2 Q}{
       \infer[\llplus]{\Sigma; \vdash B_1 P \llplus B_2 P, \B_1 Q}{
            \infer{\Sigma; \vdash B_1 P, \B_1 Q}{F_{B_1}(\Pi)}} &
       \infer[\llplus]{\Sigma; \vdash B_1 P \llplus B_2 P, \B_2 Q}{
            \infer{\Sigma; \vdash B_2 P, \B_2 Q}{F_{B_2}(\Pi)}}}
  \]

  When $B = \lambda p.~ \exists x.~ B' p x$:
  \[ \infer[\forall]{\Sigma; \vdash \exists x.~ B' P x, \forall x.~ \B' Q x}{
      \infer[\exists]{\Sigma,x; \vdash \exists x.~ B' P x, \B' Q x}{
       \infer{\Sigma,x; \vdash B' P x, \B' Q x}{F_{B'{\bullet}x}(\Pi)}
      }} \]

  When $B = \lambda p.~ \mu (B' p) \t$,
  we show that $\nu (\overline{B'} P^\perp)$ is a coinvariant of
  $\nu (\overline{B'} Q)$:
  \[ \infer[\nu]{\Sigma; \vdash \mu (B' P)\t, \nu (\overline{B'} Q)\t}{
     \infer[init]{
         \Sigma; \vdash \mu (B' P)\t, \nu (\overline{B'} P^\perp)\t}{} &
     \infer[\mu]{\x; \vdash \mu (B' P)\x,
                   (\overline{B'} Q) (\nu (\overline{B'} P^\perp))\x}{
     \infer{\x; \vdash (B' P) (\mu (B' P))\x,
                   (\overline{B'} Q) (\nu (\overline{B'} P^\perp))\x}{
          F_{(\lambda p. B' p (\mu (B' P))\x)}(\Pi)}}}
  \]
\end{longitem}
\end{definition}

\begin{proposition}[Atomic initial rule] \label{def:atomicinit}
We call \emph{atomic} the $init$ rules acting on atoms or fixed points.
The general rule $init$ is derivable from atomic initial rules.
\end{proposition}

\begin{proof}
By induction on $P$, we build a derivation of $\;\vdash P^\perp, P$
using only atomic axioms.
If $P$ is not an atom or a fixed point expression,
we perform an $\eta$-expansion as in the previous definition
and conclude by induction hypothesis.
Note that although the identity on fixed points can be expanded,
it can never be eliminated: repeated expansions do not terminate
in general.
\end{proof}

The constructions used above can be used to establish the \emph{canonicity} of
all our logical connectives: if a connective is duplicated into, say, red and
blue variants equipped with the same logical rules, then those two versions
are equivalent.
Intuitively, it means that our connectives define a unique
logical concept.
This is a known property of the connectives of first-order
MALL, we show it for $\mu$ and its copy $\hat{\mu}$ by using our color-blind
expansion:
\[ \infer[\nu]{\vdash \nu \B \t, \hat{\mu} B \t}{
          \infer[init]{\vdash \hat{\nu} \B \t, \hat{\mu} B \t}{} &
          \infer[\hat{\mu}]{\vdash \B (\hat{\nu} \B) \x, \hat{\mu} B \x}{
            \infer[init]{\vdash \B (\hat{\nu} \B) \x, B (\hat{\mu} B) \x}{}
          }
   } \]

\begin{proposition} \label{prop:unfold}
The following inference rule is derivable:
\[
  \infer[\nu{}R]{\vdash \Gamma, \nu{}B\t}{\vdash \Gamma, B(\nu{}B)\t}
\]
\end{proposition}

\begin{proof}
The unfolding $\nu{}R$ is derivable from $\nu$,
using $B(\nu{}B)$ as the coinvariant $S$.
The proof of coinvariance $\,\vdash B(B(\nu B))\x, \B(\mu \B)\x$ is obtained
by functoriality on $\,\vdash B(\nu B)\x, \mu\B\x$, itself obtained from
$\mu$ and $init$.
\end{proof}

% In an intuitionistic framework, the unfolding of least fixed points
% on the left-hand side would be obtained in the same way.

\begin{example}
In general the least fixed point entails the greatest.  The
following is a proof of $\mu{}B\t \llimp \nu{}B\t$,
showing that $\mu B$ is a coinvariant of $\nu B$:
\[  \infer[\mbox{$\nu$ on $\nu{}B\t$ with $S:=\mu{}B$}]{
       \vdash \nu{}\B\t, \nu{}B\t}{
     \infer[init]{\vdash \nu\B\t, \mu{}B\t}{}
     &
     \infer[\nu{}R]{\vdash B(\mu{}B)\x, \nu\B\x}{
       \infer[init]{\vdash B(\mu{}B)\x, \B(\nu\B)\x}{}
     }
   } \]
The greatest fixed point entails the least fixed point when the fixed
points are \emph{noetherian}, \emph{i.e.}, predicate operators have
vacuous second-order abstractions.
Finally, the $\nu{}R$ rule allows to derive $\mu B \t \lleq B (\mu B)\t$,
or equivalently $\nu B\t \lleq B (\nu B)\t$.
\end{example}

\subsection{Polarities of connectives} %%%%%%%%%%%%%%%%%%%%%%%%%%%%%%%%%%
\label{sec:mumall_positivity}

It is common to classify inference rules between invertible and non-invertible 
ones. In linear logic, we can use the refined notions of \emph{positivity}
and \emph{negativity}.
A formula $P$ is said to be positive (\resp $Q$ is said to be 
negative) when $P\lleq \oc P$ (\resp $Q\lleq \wn Q$).
A logical connective is said to be positive (\resp negative) when
it preserves positivity (\resp negativity). For example, $\lltens$ is
positive since $P\lltens P'$ is positive whenever $P$ and $P'$ are.
This notion is more semantical than invertibility, and has the advantage
of actually saying something about non-invertible connectives/rules.
Although it does not seem at first sight to be related to proof-search,
positivity turns out to play an important role in the understanding
and design of focused
systems~\cite{liang07csl,laurent02phd,laurent05apal,danos93kgc,danos93wll}.

Since \mumall\ does not have exponentials, it is not possible
to talk about positivity as defined above.
Instead, we are going to take a backwards approach: we shall first define
which connectives are negative, and then check that the obtained negative
formulas have a property close to the original negativity.
This does not trivialize the question at all: it turns out that only
one classification allows to derive the expected property.
We refer the interested reader to \cite{baelde08phd} for the extension
of that proof to \muLL, \ie\ \mumall\ with exponentials,
where we follow the traditional approach.

\begin{definition} \label{def:connectives}
We classify as \emph{negative} the following connectives:
$\llpar$, $\bot$, $\llwith$, $\top$, $\forall$, $\neq$, $\nu$.
Their duals are called \emph{positive}.
A formula is said to be negative (\resp positive) when
all of its connectives are negative (\resp positive).
Finally, an operator $\lambda{}p\lambda\x.Bp\x$ is said to be negative
(\resp positive) when the formula $Bp\x$ is negative (\resp positive).
\end{definition}

Notice, for example, that $\lambda{}p\lambda\x.p\x$
is both positive and negative. But $\mu p. p$ is only positive
while $\nu p. p$ is only negative.
Atoms (and formulas containing atoms)
are neither negative nor positive: indeed, they offer no structure\footnote{
  This essential aspect of atoms makes them often less interesting or
  even undesirable. For example, in our work on minimal generic 
  quantification~\cite{baelde08lfmtp} we show and exploit the fact
  that this third quantifier can be defined in \muLJ\ \emph{without atoms}.
} from which the following fundamental property could be derived.

\begin{proposition} \label{prop:struct}
The following structural rules are admissible for any negative formula $P$:
\[
\infer[C]{\Sigma; \vdash \Gamma, P}{\Sigma; \vdash \Gamma, P, P}
\quad
\infer[W]{\Sigma; \vdash \Gamma, P}{\Sigma; \vdash \Gamma}
\]
\end{proposition}

We can already note that this proposition could not hold if $\mu$ was
negative, since $\mu (\lambda p. p)$ cannot be weakened (there is obviously
no cut-free proof of $\,\vdash \mu (\lambda p. p), \llone$).

\begin{proof}
We first prove the admissibility of $W$.
This rule can be obtained by cutting a derivation of $\Sigma; \vdash P, \llone$.
We show more generally that
for any collection of negative formulas $(P_i)_i$,
there is a derivation of $\;\vdash (P_i)_i, \llone$.
This is done by induction on the total size of $(P_i)_i$,
counting one for each connective, unit, atom or predicate variable but
ignoring terms.
The proof is trivial if the collection is empty.
Otherwise,
if $P_0$ is a disequality we conclude by induction with one less formula,
and the size of the others unaffected by the first-order instantiation;
if it is $\top$ our proof is done;
if it is $\bot$ then $P_0$ disappears and we conclude by induction hypothesis.
The $\llpar$ case is done by induction hypothesis,
the resulting collection has one more formula but is smaller;
the $\llwith$ makes use of two instances of the induction hypothesis;
the $\forall$ case makes use of the induction hypothesis with an extended
signature but a smaller formula.
Finally, the $\nu$ case is done by applying the $\nu$ rule with $\bot$ as the 
invariant:
\[ \infer{\vdash \nu B \t, (P_i)_i, \llone}{
      \infer{\vdash \bot, (P_i)_i, \llone}{
      \vdash (P_i)_i, \llone
      } &
      \vdash B (\lambda \x. \bot) \x, \llone
    }
\]
The two subderivations are obtained by induction hypothesis.
For the second one there is only one formula, namely
$B (\lambda \x. \bot) \x$, which is indeed negative (by monotonicity of $B$)
and smaller than $\nu B$.

We also derive contraction ($C$) using a cut, this time against
a derivation of $\vdash (P\llpar P)^\perp, P$.
A generalization is needed for the greatest fixed point case,
and we derive the following for any negative $n$-ary operator $A$:
\[
  \vdash
  (A(\nu{}B_1)\ldots(\nu{}B_n) \llpar A(\nu{}B_1)\ldots(\nu{}B_n))^\perp,
  A(\nu{}B_1\llpar\nu{}B_1)\ldots(\nu{}B_n\llpar\nu{}B_n)
\]
We prove this by induction on $A$:
\begin{longitem}
\item It is trivial if $A$ is a disequality, $\top$ or $\bot$.
\item If $A$ is a projection $\lambda\vec{p}.~ p_i\t$,
  we have to derive
  $\vdash (\nu{}B_i\t\llpar\nu{}B_i\t)^\perp, \nu{}B_i\t\llpar\nu{}B_i\t$,
  which is an instance of $init$.
\item If $A$ is $\lambda \vec{p}.~ A_1\vec{p}\llpar A_2\vec{p}$,
  we can combine our two induction hypotheses to derive the following:
  \[ \vdash ((A_1(\nu{}B_i)_i\llpar A_1(\nu{}B_i)_i)\llpar
             (A_2(\nu{}B_i)_i\llpar A_2(\nu{}B_i)_i))^\perp,
            A_1(\nu{}B_i)_i\llpar A_2(\nu{}B_i)_i \]
  We conclude by associativity-commutativity of the tensor,
  which amounts to use cut against an easily obtained derivation of
  $\vdash ((P_1\llpar P_2)\llpar(P_1\llpar P_2)),
     ((P_1\llpar P_1)\llpar(P_2\llpar P_2))^\perp$
  for $P_j := A_j(\nu{}B_i)_i$.
\item If $A$ is $\lambda\vec{p}.~ A_1\vec{p}\llwith A_2\vec{p}$
  we introduce the additive conjunction
  and have to derive two similar premises:
  \[ \vdash
       ((A_1\llwith A_2)(\nu B_i)_i \llpar (A_1\llwith A_2)(\nu B_i)_i)^\perp,
       A_j (\nu B_i \llpar \nu B_i)_i \mbox{~ for $j\in\{1,2\}$} \]
  To conclude by induction hypothesis, we have to choose the correct
  projections for the negated $\llwith$.
  Since the $\llwith$ is under the $\llpar$, we have to use a cut |
  one can derive in general
  $\;\vdash ((P_1\llwith P_2)\llpar(P_1\llwith P_2))^\perp, P_j\llpar P_j$
  for $j\in\{1,2\}$.
\item When $A$ is $\lambda \vec{p}.~ \forall{}x.~ A'\vec{p}x$,
  the same scheme applies:
  we introduce the universal variable and instantiate the two existential
  quantifiers under the $\llpar$ thanks to
  a cut.
\item Finally, we treat the greatest fixed point case:
  $A$ is $\lambda\vec{p}.~ \nu{}(A'\vec{p})\t$.
  Let $B_{n+1}$ be $A'(\nu{}B_i)_{i\leq n}$.
  We have to build a derivation of
\[ \vdash (\nu B_{n+1} \t\llpar \nu B_{n+1}\t)^\perp,
     \nu (A'(\nu B_i\llpar \nu B_i)_i)\t \]
  We use the $\nu$ rule, showing
  that $\nu B_{n+1}\llpar \nu B_{n+1}$ is a coinvariant of
  $\nu (A'(\nu B_i\llpar \nu B_i)_i)$.
  The left subderivation of the $\nu$ rule is thus an instance of $init$,
  and the coinvariance derivation is as follows:
\[ \small
  \infer[cut]{
     \vdash
     (\nu{}B_{n+1}\x\llpar\nu{}B_{n+1}\x)^\perp,
     A'(\nu{}B_i\llpar\nu{}B_i)_i(\nu{}B_{n+1}\llpar\nu{}B_{n+1})\x
  }{
    \vdash
    (A'(\nu{}B_i)_i(\nu{}B_{n+1})\x \llpar 
     A'(\nu{}B_i)_i(\nu{}B_{n+1})\x)^\perp,
    A'(\nu{}B_i\llpar\nu{}B_i)_i(\nu{}B_{n+1}\llpar\nu{}B_{n+1})\x
    & \Pi' }
\]
  Here, $\Pi'$ derives
  $\vdash (\nu B_{n+1}\x \llpar \nu B_{n+1} \x)^\perp,
   A'(\nu B_i)_i(\nu B_{n+1})\x\llpar A'(\nu B_i)_i(\nu B_{n+1})\x$,
  unfolding $\nu B_{n+1}$ under the tensor.
  We complete our derivation by induction hypothesis,
  with the smaller operator expression $A'$
  and $B_{n+1}$ added to the $(B_i)_i$.
\end{longitem}
\vspace{-0.5cm}\end{proof}

The previous property yields some interesting remarks about
the expressiveness of \mumall.
It is easy to see that provability is undecidable in \mumall,
by encoding (terminating) executions of a Turing machine as a least fixed
point.
But this kind of observation does not say anything about what
theorems can be derived, \ie\ the complexity of reasoning/computation
allowed in \mumall.
Here, the negative structural rules derived in Proposition~\ref{prop:struct}
come into play.
Although our logic is linear, it enjoys those derived structural rules
for a rich class of formulas: for example, $nat$ is positive, hence reasoning
about natural numbers allows contraction and weakening, just
like in an intuitionistic setting.
Although the precise complexity of the normalization of \mumall{}
is unknown, we have adapted some remarks 
from~\cite{burroni86,girard87tcs,alves06csl} to build an encoding
of primitive recursive functions in \mumall{}~\cite{baelde08phd} |
in other words, all primitive recursive functions can be proved
total in \mumall.

\subsection{Examples} \label{sec:mumall_examples} %%%%%%%%%%%%%%%%%%%%%%%%%%%%

We shall now give a few theorems derivable in \mumall. Although we do not
provide their derivations here but only brief descriptions of how to obtain
them, we stress that all of these examples are proved naturally.
The reader will note that although \mumall\ is linear,
these derivations are intuitive
and their structure resembles that of proofs in intuitionistic logic.
We also invite the reader to check that the $\mu$-focusing system
presented in Section~\ref{sec:foc_mumall} is a useful guide
when deriving these examples, leaving only the important choices.
It should be noted that atoms are not used in this section;
in fact, atoms are rarely useful in \mumall, as its main application
is to reason about (fully defined) fixed points.

Following the definition of $nat$ from Example~\ref{example:nat},
we define a few least fixed points expressing basic properties of natural 
numbers. Note that all these definitions are positive.
\[ \begin{array}{lcl}
 even &\stackrel{def}{=}&
  \mu(\lambda{}E\lambda{}x.~ x=0
           \llplus \exists{}y.~ x=s~(s~y) \lltens E~y)
\\
   plus &\stackrel{def}{=}&
     \mu(\lambda{}P\lambda{}a\lambda{}b\lambda{}c.~ a=0 \lltens b=c \\
   & & \phantom{\mu\lambda{}P\lambda{}a\lambda{}b\lambda{}c.}
     \llplus \exists{}a'\exists{}c'. a=s~a' \lltens c=s~c' \lltens P~a'~b~c')
\\
  leq &\stackrel{def}{=}&
   \mu(\lambda{}L\lambda{}x\lambda{}y.~ x=y
                     \llplus \exists{}y'.~ y=s~y' \lltens L~x~y')
\\
  \half &\stackrel{def}{=}&
    \mu(\lambda{}H\lambda{}x\lambda{}h.~ (x=0 \llplus x=s~0) \lltens h=0 \\
  & & \phantom{\mu\lambda{}H\lambda{}x\lambda{}h.}
      \llplus \exists{}x'\exists{}h'.~ x=s~(s~x') \lltens h=s~h' 
                                    \lltens H~x'~h')
\\
  ack & \eqdef &
    \mu (\lambda A \lambda m \lambda n \lambda a.~
      m = 0 \lltens a = s~n \\
  & & \phantom{\mu\lambda \lambda m \lambda n \lambda a.~}
        \llplus (\exists p.~ m = s~p \lltens n = 0 \lltens A~p~(s~0)~a) \\
  & & \phantom{\mu\lambda \lambda m \lambda n \lambda a.~}
        \llplus (\exists p \exists q \exists b.~ m = s~p \lltens n = s~q 
           \lltens A~m~q~b \lltens A~p~b~a))
\end{array} \]

The following statements are theorems in \mumall.
The main insights required for proving these theorems involve
deciding which fixed point expression should be introduced by
induction: the proper invariant is not the difficult choice here since
the context itself is adequate in these cases.
\[ \begin{array}{l}
\vdash \forall{}x.~ nat~x \llimp even~x \llplus even~(s~x) \\
\vdash \forall{}x.~ nat~x \llimp \forall{}y\exists{}z.~ plus~x~y~z \\
\vdash \forall{}x.~ nat~x \llimp plus~x~0~x \\
\vdash \forall{}x.~ nat~x \llimp \forall{}y.~ nat~y \llimp
           \forall{}z.~ plus~x~y~z \llimp nat~z \\
\end{array} \]
In the last theorem, the assumption $(nat~x)$ is not needed and can
be weakened, thanks to Proposition~\ref{prop:struct}.
In order to prove
$(\forall{}x.~ nat~x \llimp \exists{}h.~ \half~x~h)$
the context does not provide an invariant that is strong enough.
A typical solution is to use complete induction,
\ie\ use the strengthened invariant
$(\lambda{}x.~ nat~x \lltens
               \forall{}y.~ leq~y~x \llimp \exists{}h.~ \half~y~h)$.

We do not know of any proof of totality for a non-primitive
recursive function in \mumall.
In particular, we have no proof of
$\forall x \forall y.~ nat~x \llimp nat~y \llimp \exists z.~ ack~x~y~z$.
The corresponding intuitionistic theorem can be proved using nested
inductions, but it does not lead to a linear proof since it requires to
contract an implication hypothesis (in \mumall, the dual of an implication
is a tensor, which is not negative and thus cannot \emph{a priori}
be contracted).

A typical example of co-induction involves the simulation
relation.  Assume that $step : state \ra label \ra state \ra o$ is an
inductively defined relation encoding a labeled transition system.
Simulation can be defined using the definition
\[
sim \stackrel{def}{=}
 \nu(\lambda{}S\lambda{}p\lambda{}q.~
   \forall{}a\forall{}p'.~ step~p~a~p'
     \llimp \exists{}q'.~ step~q~a~q' \lltens S~p'~q').
\]
Reflexivity of simulation ($\forall{}p.~ sim~p~p$) is proved easily
by co-induction with the co-invariant $(\lambda{}p\lambda{}q.~ p=q)$.
Instances of $step$ are not subject to induction but
are treated ``as atoms''. 
Proving transitivity, that is,
$$\forall{}p\forall{}q\forall{}r.~ sim~p~q \llimp sim~q~r \llimp sim~p~r$$
is done by co-induction on $(sim~p~r)$ with the co-invariant
$(\lambda{}p\lambda{}r.~ \exists{}q.~ sim~p~q \lltens sim~q~r)$.
The focus is first put on $(sim~p~q)^\bot$, then on $(sim~q~r)^\bot$.
The fixed points $(sim~p'~q')$ and $(sim~q'~r')$ appearing later in the proof
are treated ``as atoms'', as are all instances of $step$.
Notice that
these two examples are also cases where the context gives a coinvariant.

%% file: norm.tex
% vi:foldmarker=<<<,>>>:foldmethod=marker

In~\cite{baelde07lpar}, we provided an indirect proof of normalization
based on the second-order encoding of \mumall.
However, that proof relied on the normalization of second-order linear logic
extended with first-order quantifiers, and more importantly equality, but
this extension of Girard's result for propositional second-order linear logic
is only a (mild) conjecture.
Moreover, such an indirect proof does not provide cut reduction rules,
which usually illuminate the structure and meaning of a logic.
In this paper, we give the first direct and full proof of normalization
for \mumall:
we provide a system of reduction rules for eliminating cuts,
and show that it is weakly normalizing by
using the candidates of reducibility technique~\cite{girard87tcs}.
Establishing strong normalization would be useful,
but we leave it to further work.
Note that the candidates of reducibility technique
is quite modular in that respect:
in fact, \cite{girard87tcs} only provided a proof of weak normalizability
together with a conjectured standardization lemma from which strong
normalization would follow.
Also note, by the way, that Girard's proof applies to proof nets,
while we shall work directly within sequent calculus;
again, the adaptation is quite simple.
Finally, the candidate of reducibility is also modular in that it
relies on a compositional interpretation of connectives,
so that our normalization proof (unlike the earlier one)
should extend easily to exponentials and second-order quantification
using their usual interpretations.

Our proof can be related to similar work in other settings.
While it would technically have been possible to interpret fixed
points as candidates through their second-order encoding,
we found it more appealing
to directly interpret them as fixed point candidates.
In that respect, our work can be seen as an adaptation of the ideas
from \cite{mendler91apal,matthes98csl} to the classical linear setting,
where candidates of reducibility are more naturally expressed
as bi-orthogonals.
This adaptation turns out to work really well, and the interpretation
of least fixed points as least fixed points on candidates yields
a rather natural proof, notably proceeding by meta-level induction
on that fixed point construction.
Also related, of course, is the work on definitions;
although we consider a linear setting and definitions have been
studied in intuitionistic logic, we believe that our proof could be adapted,
and contributes to the understanding of similar notions.
In addition to the limitations of definitions over fixed points,
the only published proof of cut elimination~\cite{momigliano03types,tiu04phd}
further restricts definitions to strictly positive ones,
and limits the coinduction rule to coinvariants of smaller ``level''
than the considered coinductive object.
However, those two restrictions have been removed in
\cite{tiu10unp}, which relies (like our proof) on a full candidate
of reducibility argument rather than the earlier non-parametrized 
reducibility, and essentially follows (unlike our proof)
a second-order encoding of definitions.

%% Now, the proof

We now proceed with the proof, defining cut reductions and then
showing their normalization.
Instead of writing proof trees, we shall often use an informal term notation
for proofs, when missing details can be inferred from the context.
We notably write $cut(\Pi;\Pi')$ for a cut, and more generally
$cut(\Pi;\vec{\Pi'})$ for the sequence of cuts 
$cut(\ldots{}cut(\Pi;\Pi'_1)\ldots;\Pi'_n)$.
We also use notations such as
$\Pi\lltens\Pi'$, $\mu \Pi$, $\nu(\Pi,\Theta)$, etc.
Although the first-order structure does not play a role in the termination
and complexity of reductions, we decided to treat it directly in the proof,
rather than evacuating it in a first step. We tried to keep it readable, but
encourage the reader to translate the most technical parts for the
purely propositional case in order to extract their core.

% As in previous work on cut elimination involving a similar treatment of
% equality~\cite{tiu04phd},
% we shall use in this section a more flexible, but equivalent equality rule,
% that enumerates \emph{all} unifiers instead of only a complete set of them:
% \[
%    \infer{\vdash u\neq v, \Gamma}{\set{\vdash \Gamma\theta}{u\theta= v\theta}}
% \]
% Translating from the general to the $csu$ rule simply consists in dropping
% subderivations.
% Translating from the $csu$-based rule is done as follows:
% to get a derivation from $\theta$, we have $\theta = \sigma\sigma'$ where
% $\sigma\in csu(u\unif v)$, and a derivation for $\sigma$, on which we
% apply $\sigma'$.
% Note that all proof transformations described before, including
% proof instantiation, are trivially adapted to this style.
%  NOTE proof instantiation is actually easier,
%    but Pi->Pi' does not imply Pitheta -> Pi'theta (see main equality case)

% <<< Reduction rules

\subsection{Reduction rules}

Rules reduce instances of the cut rule,
and are separated into auxiliary and main rules.
Most of the rules are the same as for MALL.
For readability, we do not show the signatures $\Sigma$ when they
are not modified by reductions,
leaving to the reader the simple task of inferring them.

% <<< Aux
\subsubsection{Auxiliary cases}

If a subderivation does not start with a logical rule in which the cut
formula is principal,
its first rule is permuted with the cut.
We only present the commutations for the left subderivation,
the situation being perfectly symmetric.
% NOTE the text below makes it sound deterministic but really
%   several reductions can apply on a same redex and there is no
%   strategy forced at this point

\begin{longitem}

\item If the subderivation starts with a cut, splitting
  $\Gamma$ into $\Gamma',\Gamma''$, we reduce as follows:
 \disp{\infer[cut]{\vdash \Gamma',\Gamma'',\Delta}{
     \infer[cut]{\vdash \Gamma',\Gamma'',P^\perp}{
       \vdash \Gamma',P^\perp,Q^\perp &
       \vdash \Gamma'',Q}
     & \vdash P,\Delta}
    }{
    \infer[cut]{\vdash \Gamma',\Gamma'',\Delta}{
      \infer[cut]{\vdash \Gamma',\Delta,Q^\perp}{
        \vdash \Gamma',Q^\perp,P^\perp &
        \vdash P,\Delta} &
      \vdash Q,\Gamma''}
    }
  Note that this reduction alone leads to cycles,
  hence our system is trivially not strongly normalizing.
  This is only a minor issue, which could be solved, for example,
  by using proof nets or a classical multi-cut rule (which amounts to
  incorporate the required amount of proof net flexibility into
  sequent calculus).

\item % axiom
Identity between a cut formula and a formula from the conclusion:
$\Gamma$ is restricted to the formula $P$ and
the left subderivation is an axiom.
The cut is deleted and the right subderivation is now directly connected
to the conclusion instead of the cut formula:
\disp{\infer[cut]{\vdash P,\Delta}{
    \infer[init]{\vdash P,P^\perp}{} &
    \infer{\vdash P\vphantom{^\perp},\Delta}{\Pi}}
}{
   \infer{\vdash P,\Delta}{\Pi}
}

\item % tensor
When permuting a cut and a $\lltens$,
the cut is dispatched according to the splitting of the cut formula.
When permuting a cut and a $\llwith$, the cut is duplicated.
The rules $\llpar$ and $\llplus$ are easily commuted down the cut.

\item The commutations of $\top$ and $\perp$ are simple,
  and there is none for $\llone$ nor $\llzero$.

%\item When $\Gamma=\Gamma',\perp$ and $\perp$ is introduced:
%\[ \infer[cut]{\vdash \Gamma',\perp,\Delta}{
%      \infer[\perp]{\vdash \Gamma,\perp,P^\perp}{\vdash \Gamma,P^\perp}
%      & \vdash P,\Delta}
%   \longrightarrow
%   \infer[\perp]{\vdash \Gamma', \perp, \Delta}{
%   \infer[cut]{\vdash \Gamma', \Delta}{
%      \vdash \Gamma,P^\perp & \vdash P,\Delta}}
%\]
%
%\item When $\Gamma=\Gamma',\top$ and $\top$ is introduced:
%\[ \infer[cut]{\vdash \Gamma',\top,\Delta}{
%    \infer[\top]{\vdash \Gamma',\top,P^\perp}{} & \vdash P,\Delta}
%   \longrightarrow
%   \infer[\top]{\vdash \Gamma,\top,\Delta}{} \]

\item When $\forall$ is introduced, it is permuted down and
  the signature of the other derivation is extended.
  The $\exists$ rule is permuted down without any problem.

\item There is no commutation for equality ($=$).
  When a disequality (${\neq}$) is permuted down, the other premise is
  duplicated and instantiated:
  \disp{ \infer[cut]{\Sigma; \vdash \Gamma',u\neq v,\Delta}{
      \infer[\neq]{\Sigma; \vdash \Gamma',u\neq v,P^\perp}{
          \All{ 
             \infer{\Sigma\theta; \vdash \Gamma'\theta,P^\perp\theta}{
               \Pi_\theta}}
      } &
      \infer{\Sigma; \vdash P\vphantom{^\perp},\Delta}{\Pi'}}
}{
     \infer[\neq]{\Sigma; \vdash \Gamma',u\neq v,\Delta}{
      \All{
       \infer[cut]{\Sigma\theta; \vdash \Gamma'\theta,\Delta\theta}{
        \infer{\Sigma\theta; \vdash \Gamma'\theta, P^\perp\theta}{\Pi_\theta} 
        &
        \infer{\Sigma\theta; \vdash P\vphantom{^\perp}\theta,\Delta\theta}{
           \Pi'\theta}
       }
      }
     }
}

\item % mu
$\Gamma=\Gamma',\mu B\t$ and that least fixed point is introduced:
\disp{ \infer[cut]{\vdash \Gamma',\mu B\t,\Delta}{
      \infer[\mu]{\vdash \Gamma',\mu B\t,P^\perp}{
         \vdash \Gamma',B(\mu B)\t,P^\perp}
      & \vdash P,\Delta}
}{
   \infer[\mu]{\vdash \Gamma',\mu B\t,\Delta}{
   \infer[cut]{\vdash \Gamma',B(\mu B)\t,\Delta}{
      \vdash \Gamma',B(\mu B)\t,P^\perp
      & \vdash P,\Delta}}
}

\item % nu
$\Gamma=\Gamma',\nu B\t$ and that greatest fixed point is introduced:
\disp{
   \infer[cut]{\vdash \Gamma',\nu B\t,\Delta}{
     \infer[\nu]{\vdash \Gamma',\nu B\t,P^\perp}{
        \vdash \Gamma',S\t,P^\perp & \vdash S\x^\perp, BS\x}
     & \vdash P,\Delta
   }
}{
   \infer[\nu]{\vdash \Gamma',\nu B\t,\Delta}{
   \infer[cut]{\vdash \Gamma',S\t,\Delta}{
     \vdash \Gamma',S\t,P^\perp & \vdash P,\Delta}
   & \vdash S\x^\perp, BS\x}
}

\end{longitem}

% >>>
% <<< Main
\subsubsection{Main cases}
When a logical rule is applied on the cut formula on both sides,
one of the following reductions applies.
\begin{longitem}
\item
  In the multiplicative case,
  $\Gamma$ is split into $(\Gamma',\Gamma'')$
  and we cut the subformulas.
\disp{
  \infer[cut]{\vdash \Gamma',\Gamma'',\Delta}{
     \infer[\lltens]{\vdash \Gamma', \Gamma'',P'\lltens P''\vphantom{^\perp}}{
       \vdash \Gamma', P' &
       \vdash \Gamma'', P''} &
     \infer[\llpar]{\vdash P'^\perp\llpar P''^\perp,\Delta}{
       \vdash P'^\perp, P''^\perp, \Delta}}
}{
 \infer[cut]{\vdash \Gamma',\Gamma'',\Delta}{
     \vdash \Gamma', P' &
     \infer[cut]{\vdash P'^\perp, \Gamma'',\Delta}{
       \vdash \Gamma'', P'' &
       \vdash P'^\perp, P''^\perp, \Delta}}
}

\item
  In the additive case, we select the appropriate premise of $\llwith$.
\disp{
  \infer[cut]{\vdash \Gamma,\Delta}{
    \infer[\llplus]{\vdash \Gamma, P_0 \llplus P_1}{\vdash \Gamma, P_i} &
    \infer[\llwith]{\vdash \Delta, P_0^\perp \llwith P_1^\perp}{
      \vdash \Delta, P_0^\perp &
      \vdash \Delta, P_1^\perp}}
}{
  \infer[cut]{\vdash \Gamma,\Delta}{
      \vdash \Gamma, P_i &
      \vdash \Delta, P_i^\perp}
}

\item The $\llone/\bot$ case reduces to
  the subderivation of $\bot$.
  There is no case for $\top/\llzero$.
  
\item In the first-order quantification case,
  we perform a proof instantiation:
  \disp{ \infer[cut]{\Sigma; \vdash \Gamma,\Delta}{
      \infer[\exists]{\Sigma; \vdash \Gamma,\exists x.~ P x^\perp}{
        \infer{\Sigma; \vdash \Gamma,P t^\perp}{\Pi_l}} &
      \infer[\forall]{\Sigma; \vdash \forall x.~ P x\vphantom{^\perp}, \Delta}{
        \infer{\Sigma,x; \vdash P x\vphantom{^\perp}, \Delta}{\Pi_r}}}
}{
     \infer[cut]{\Sigma; \vdash \Gamma,\Delta}{
       \infer{\Sigma; \vdash \Gamma,P t^\perp}{\Pi_l} &
       \infer{\Sigma; \vdash P t, \Delta\vphantom{^\perp}}{\Pi_r[t/x]}}
}

\item The equality case is trivial, the interesting part concerning
  this connective lies in the proof instantiations triggered by other
  reductions. Since we are considering two terms that are already
  equal, we have $csu(u\unif u)=\{id\}$ and
  we can simply reduce to the subderivation corresponding to the
  identity substitution:
  \disp{ \infer[cut]{\Sigma; \vdash \Delta}{
       \infer[=]{\Sigma; \vdash u=u}{} &
       \infer[\neq]{\Sigma; \vdash u\neq u, \Delta}{
          \infer{\Sigma; \vdash \Delta}{\Pi_{id}}
        }}
}{
     \infer{\Sigma; \vdash \Delta}{\Pi_{id}}
}

\item Finally in the fixed point case,
  we make use of the functoriality transformation for propagating
  the coinduction/recursion under $B$:
\disp{
   \infer[cut]{\Sigma; \vdash \Gamma,\Delta}{
      \infer[\mu]{\Sigma; \vdash \Gamma, \mu B\t}{
        \infer{\Sigma; \vdash \Gamma, B(\mu B)\t}{\Pi'_l}} &
      \infer[\nu]{\Sigma; \vdash \Delta, \nu \B\t}{
        \infer{\Sigma; \vdash \Delta, S\t}{\Pi'_r} &
        \infer{\x; \vdash S\x^\perp, \B S\x\vphantom{\t}}{\Theta}
      }}
}{
   \infer[cut]{\Sigma; \vdash \Gamma,\Delta}{
      \infer{\Sigma; \vdash \Delta,S\t}{\Pi'_r} &
      \infer[cut]{\Sigma; \vdash S\t^\perp,\Gamma}{
        \infer{\Sigma; \vdash S\t^\perp, \B S\t}{\Theta[\t/\x]} &
        \infer[cut]{\Sigma; \vdash B S^\perp \t, \Gamma}{
           \infer{\Sigma; \vdash B S^\perp \t, \B(\nu\B)\t}{
             F_{B{\bullet}\t}(\nu(Id,\Theta))} &
           \infer{\Sigma; \vdash B (\mu B)\t, \Gamma}{\Pi'_l}
        }}}
}

\end{longitem}

% >>>

%\begin{remark}
%It is perfectly valid to see the $\nu$ rule as a compound of a cut
%and a more elementary coinduction rule:
%\[ \infer{\vdash S^\perp, \nu B}{\vdash S^\perp, BS} \]
%This point of view would seem to simplify some arguments, notably
%in the previous lemma. One would take the simpler $\nu$ rule
%in the reduction rules, and only recover the general one by
%composing it with a possible blocked cut below it.
%But this would not actually yield cut-free derivations,
%unless a rule is added to commute blocked cuts down the other ones
%so that they can be reduced.
%In the end such a rule is much more naturally expressed as the
%auxiliary cut reduction for $\nu$.
%\end{remark}

One-step reduction $\Pi\ra\Pi'$ is defined as the congruence
generated by the above rules.
We now seek to establish that such reductions can be applied
to transform any derivation into a cut-free one.
However, since we are dealing with transfinite (infinitely branching)
proof objects,
there are trivially derivations which cannot be reduced into
a cut-free form in a finite number of steps.
A possibility would be to consider transfinite reduction sequences,
relying on a notion of convergence for defining limits.
A simpler solution, enabled by the fact that our infinity
only happens ``in parallel'', is to define inductively
the transfinite reflexive transitive closure of one-step reduction.

\newcommand{\Ra}{\ra^*}

\begin{definition}[Reflexive transitive closure, $\WN$]
We define inductively $\Pi\Ra\Xi$ to hold when
(1) $\Pi\ra\Xi$,
(2) $\Pi\Ra\Pi'$ and $\Pi'\Ra\Xi$,
or
(3) $\Pi$ and $\Xi$ start with the same rule and their premises are in relation
   (\ie\ for some rule ${\cal R}$,
      $\Pi = {\cal R}(\Pi_i)_i$, $\Xi = {\cal R}(\Xi_i)_i$
      and each $\Pi_i\Ra\Xi_i$).
We say that $\Pi$ \emph{normalizes} when there exists a cut-free
derivation $\Pi'$ such that $\Pi\Ra\Pi'$.
We denote by $\WN$ the set of all normalizing derivations.
\end{definition}

% NOTE
%   it is trivial that Pi->*Pi' implies that they have the same conclusion
%   so that "Pi normalizes" means that its conclusion has a cut-free proof

From (1) and (2), it follows that if $\Pi$ reduces to $\Xi$ in $n>0$ steps,
then $\Pi\Ra\Xi$.
From (3) it follows that $\Pi\Ra\Pi$ for any $\Pi$.
In the finitely branching case, \ie\ if the $\neq$ connective was
removed or the system ensured finite $csu$, the role of (3) is only
to ensure reflexivity.
In the presence of infinitely branching rules, however,
it also plays the important role of packaging an infinite number of reductions.
In the finitely branching case, one can show that $\Pi\Ra\Xi$ implies
that there is a finite reduction sequence from $\Pi$ to $\Xi$
(by induction on $\Pi\Ra\Xi$),
and so our definition of normalization corresponds to the usual notion
of weak normalization in that case.

%% Properties of Ra

\begin{proposition}
If $\Pi\ra\Xi$ then $\Pi\theta\Ra\Xi\theta$.
\end{proposition}

\begin{proof}
By induction on $\Pi$.
If the redex is not at toplevel but in an immediate subderivation $\Pi'$,
then the corresponding subderivations in $\Pi\theta$ shall be reduced.
If the first rule of $\Pi$ is disequality, there may be zero, several
or infinitely many subderivations of $\Pi\theta$ of the form $\Pi'\theta'$.
Otherwise there is only one such subderivation.
In both cases,
we show $\Pi\theta\Ra\Xi\theta$ by (3),
using the induction hypothesis for the subderivations where the redex is,
and reflexivity of $\Ra$ for the others.

If the redex is at toplevel, then $\Pi\theta\ra\Xi\theta$.
The only non-trivial cases are the two reductions involving ${\neq}$.
In the auxiliary case, we have:
\[ \xymatrix{
   cut({\neq}(\Pi_\sigma)_{\sigma\in csu(u\unif v)};\Pi_r)
     \ar[r] \ar[d]^{\theta}
   & {\neq}(cut(\Pi_\sigma;\Pi_r\sigma))_\sigma \ar[d]^{\theta} \\
   cut({\neq}(\Pi'_{\sigma'})_{\sigma'\in csu(u\theta\unif v\theta)};
       \Pi_r\theta)
   \ar@{.>}[r]
   & {\neq}(cut(\Pi'_{\sigma'};(\Pi_r\theta)\sigma'))_{\sigma'}
} \]
By Definition~\ref{def:inst}, $\Pi'_{\sigma'} = \Pi_{\sigma}\sigma''$ for
$\theta\sigma' = \sigma\sigma''$, $\sigma\in csu(u\unif v)$.
% NOTE below we use determinacy of decomposition of theta sigma'
Applying $\theta$ on the reduct of $\Pi$, we obtain for each $\sigma'$
the subderivation
$cut(\Pi_\sigma;\Pi_r\sigma)\sigma'' =
 cut(\Pi_\sigma\sigma'';\Pi_r\sigma\sigma'') =
 cut(\Pi'_{\sigma'};\Pi_r\theta\sigma')$.
% NOTE below we use that csu(u=u)={id}
In the main case, $\Pi = cut({\neq}(\Pi_{id});u=u) \ra \Pi_{id}$
and $\Pi\theta = cut({\neq}(\Pi'_{id});u\theta=u\theta)
         \ra \Pi'_{id} = \Pi_{id}\theta$.
\end{proof}

\begin{proposition} \label{prop:theta_wn}
If $\Pi$ is normalizing then so is $\Pi\theta$.
\end{proposition}

\begin{proof}
Given a cut-free derivation $\Pi'$ such that $\Pi\Ra\Pi'$,
we show that $\Pi\theta\Ra\Pi'\theta$ by a simple induction on $\Pi\Ra\Pi'$,
making use of the previous proposition.
% NOTE a->b=>c give us a=>b=>c so we really need transitivity there
\end{proof}

\begin{proposition} \label{prop:id_wn}
We say that $\Xi$ is an $Id$-simplification of $\Pi$
if it is obtained from $\Pi$ by reducing an arbitrary,
potentially infinite number of redexes $cut(\Theta;Id)$ into $\Theta$.
If $\Xi$ is an $Id$-simplification of $\Pi$,
and $\Pi$ is normalizable then so is $\Xi$.
\end{proposition}

\begin{proof}
We show more generally
that if $\Xi$ is a simplification of $\Pi$ and $\Pi\Ra\Pi'$
then $\Xi\Ra\Xi'$ for some simplification $\Xi'$ of $\Pi'$.
This is easily done by induction on $\Pi\Ra\Pi'$, once we will have
established the following fact:
\emph{
  If $\Xi$ is a simplification of $\Pi$
  and $\Pi\ra\Pi'$,
  then $\Xi\Ra\Xi'$ for a simplification $\Xi'$ of $\Pi'$.}
If the redex in $\Pi$
does not involve simplified cuts, the same reduction can be 
performed in $\Xi$, and the result is a simplification of $\Pi'$
(note that this could erase or duplicate some simplifications).
If the reduction is one of the simplications
then $\Xi$ itself is a simplification of $\Pi'$. % one less cut to remove
If a simplified cut is permuted with another cut (simplified or not)
$\Xi$ is also a simplification of $\Pi'$. % same cuts, different positions
Finally, other auxiliary reductions on a simplified cut
also yield reducts of which $\Xi$ is already a simplification
(again, simplifications may be erased or duplicated).
\end{proof}

% >>>
% <<< Reducibility candidates

\subsection{Reducibility candidates} %%%%%%%%%%%%%%%%%%%%%%%%%%%%%%%%%%%%%%%

\begin{definition}[Type]
A proof of type $P$ is a proof with a distinguished formula $P$
among its conclusion sequent.
We denote by $Id_P$ the axiom rule between $P$ and $P^\perp$,
of type $P$.
%If $\Pi$ has type $P$ and $\Pi'$ has type $P^\perp$, we can
%form a cut on $P$ between them, denoted by $cut(\Pi;\Pi')$.
\end{definition}

In full details, a type should contain a signature under which the formula
is closed and well typed.
That extra level of information would be heavy,
and no real difficulty lies in dealing with it,
and so we prefer to leave it implicit.

If $X$ is a set of proofs,
we shall write $\Pi : P \in X$
as a shortcut for ``$\Pi\in X$ and $\Pi$ has type $P$''.
We say that $\Pi$ and $\Pi'$ are \emph{compatible} if their types are dual
of each other.

\begin{definition}[Orthogonality]
For $\Pi, \Pi' \in \WN$,
we say that $\Pi\Perp\Pi'$ when
for any $\theta$ and $\theta'$ such that $\Pi\theta$ and $\Pi'\theta'$
are compatible,
$cut(\Pi\theta;\Pi'\theta')\in\WN$.
For $\Pi\in\WN$ and $X\subseteq\WN$,
$\Pi\Perp X$ iff $\Pi\Perp\Pi'$ for any $\Pi'\in X$,
and $X^\perp$ is $\set{\Pi\in\WN}{\Pi\Perp X}$.
Finally, for $X,Y\subseteq\WN$, $X\Perp Y$ iff
$\Pi\Perp\Pi'$ for any $\Pi\in X$, $\Pi'\in Y$.
\end{definition}

\begin{definition}[Reducibility candidate]
A \emph{reducibility candidate} $X$ is a set of normalizing proofs
that is equal to its bi-orthogonal, \ie\ $X=X^{\perp\perp}$.
\end{definition}

\newcommand{\lfp}{\mathrm{lfp}}

That kind of construction has some well-known properties\footnote{
  This so-called \emph{polar} construction is used independently
  for reducibility candidates and phase semantics in \cite{girard87tcs},
  but also, for example,
  to define behaviors in ludics \cite{girard01mscs}.
},
which do not rely on the definition of the relation $\Perp$.
For any sets of normalizable derivations $X$ and $Y$,
$X\subseteq Y$ implies $Y^\perp \subseteq X^\perp$
and $(X\cup Y)^\perp=X^\perp \cap Y^\perp$;
moreover, the symmetry of $\Perp$ implies that $X\subseteq X^{\perp\perp}$,
and hence $X^\perp=X^{\perp\perp\perp}$
(in other words, $X^\perp$ is always a candidate).

Reducibility candidates, ordered by inclusion, form a complete lattice:
given an arbitrary collection of candidates $S$,
it is easy to check that
$({\bigcup} S)^{\perp\perp}$ is its least upper bound in the lattice,
and ${\bigcap} S$ its greatest lower bound.
We check the minimality of $({\bigcup} S)^{\perp\perp}$:
any upper bound $Y$ satisfies ${\bigcup} S \subseteq Y$,
and hence $({\bigcup} S)^{\perp\perp} \subseteq Y^{\perp\perp} = Y$.
Concerning the greatest lower bound, the only non-trivial thing
is that it is a candidate, but it suffices to observe that
${\bigcap} S = {\bigcap}_{X\in S} X^{\perp\perp} =
  ({\bigcup}_{X\in S} X^\perp)^\perp$.
The least candidate is $\emptyset^{\perp\perp}$
and the greatest is $\WN$.
Having a complete lattice, we can use the Knaster-Tarski theorem:
any monotonic operator $\phi$ on reducibility candidates
admits a least fixed point $\lfp(\phi)$ in the lattice of candidates.

Our definition of $\Perp$ yields some basic observations about candidates.
They are closed under substitution,
\ie\ $\Pi\in X$ implies that any $\Pi\theta\in X$.
Indeed, $\Pi\in X$ is equivalent to $\Pi\Perp X^\perp$
  which implies $\Pi\theta\Perp X^\perp$ by definition of $\Perp$
  and Proposition \ref{prop:theta_wn}.
Hence, $Id_P$ belongs to any candidate, since
 for any $\Pi\in X^\perp$,
 $cut(Id_{P\theta};\Pi\theta')\ra\Pi\theta'\in X^\perp \subseteq \WN$.
Candidates are also closed under expansion, \ie\
$\Pi'\ra\Pi$ and $\Pi\in X$ imply that $\Pi'\in X$.
Indeed,
  for any $\Xi\in X^\perp$,
  $cut(\Pi'\theta;\Xi\theta')\Ra cut(\Pi\theta;\Xi\theta')$
  by Proposition~\ref{prop:theta_wn},
  and the latter derivation normalizes.

A useful simplification follows from those properties:
for a candidate $X$, $\Pi\Perp X$ if for any $\theta$
and compatible $\Pi'\in X$, $cut(\Pi\theta;\Pi')$ normalizes | there
is no need to explicitly consider instantiations of members of $X$,
and since $Id\in X$, there is no need to show that $\Pi$ normalizes
by Proposition~\ref{prop:id_wn}.

The generalization over all substitutions is the only
novelty in our definitions. It is there to internalize the
fact that proof behaviors are essentially independent of their
first-order structure. By taking this into account from the beginning
in the definition of orthogonality, we obtain bi-orthogonals (behaviors)
that are closed under inessential transformations like substitution.
As a result,
unlike in most candidate of reducibility arguments, our candidates are
untyped. In fact, we could type them up-to first-order details,
\ie\ restrict to sets of proofs whose types have the same
propositional structure. Although that might look more familiar,
we prefer to avoid those unnecessary details.

\begin{definition}[Reducibility]
Let $\Pi$ be a proof of $\vdash P_1,\ldots,P_n$,
and $(X_i)_{i=1\ldots n}$ a collection of reducibility candidates.
We say that $\Pi$ is $(X_1,\ldots,X_n)$-reducible if for any $\theta$
and any derivations $(\Pi'_i : P_i\theta^\perp \in X_i^\perp)_{i=1\ldots n}$,
the derivation
$cut(\Pi\theta;\Pi'_1,\ldots,\Pi'_n)$ normalizes.
\end{definition}

From this definition, it immediately follows that
if $\Pi$ is $(X_1,\ldots,X_n)$-reducible
then so is $\Pi\theta$.
Also observe that $Id_P$ is $(X,X^\perp)$-reducible for any candidate $X$,
since for any $\Pi\in X$ and $\Pi'\in X^\perp$
$cut(Id_{P\theta};\Pi,\Pi')$ reduces to $cut(\Pi;\Pi')$ which normalizes.
Finally, any $(X_1,\ldots,X_n)$-reducible derivation $\Pi$ normalizes,
by Proposition~\ref{prop:id_wn} and the fact that
$cut(\Pi;Id,\ldots,Id)$ normalizes.

\begin{proposition} \label{prop:red-interp}
% Here n=0 is excluded, and n=1 is fine (overkill, actually).
Let $\Pi$ be a proof of $\vdash P_1,\ldots,P_n$,
let $(X_i)_{i=1\ldots n}$ be a family of candidates,
and let $j$ be an index in $1\ldots n$.
The two following statements are equivalent:
\emph{(1)} $\Pi$ is $(X_1,\ldots,X_n)$-reducible;
\emph{(2)}
  for any $\theta$ and $(\Pi'_i : P_i\theta^\perp \in X_i^\perp)_{i\neq j}$,
  $cut(\Pi\theta;(\Pi'_i)_{i\neq j})\in X_j$.
\end{proposition}

\begin{proof}
{(1) $\Rightarrow$ (2)}:
Given such $\theta$ and $(\Pi'_i)_{i\neq j}$,
we show that the derivation
$cut(\Pi\theta;(\Pi'_i)_{i\neq j}) \in X_j$.
Since $X_j=X_j^{\perp\perp}$,
it is equivalent to show that our derivation is in the orthogonal
of $X_j^\perp$.
For each $\sigma$ and $\Pi'': P_j\theta\sigma^\perp \in X_j^\perp$,
we have to show that
$cut(cut(\Pi\theta;(\Pi'_i)_{i\neq j})\sigma;\Pi'')$ normalizes.
Using cut permutation reductions, we reduce it into
$cut(\Pi\theta\sigma;
  \Pi'_1\sigma,\ldots,\Pi'',\ldots, \Pi'_n\sigma)$,
which normalizes by reducibility of $\Pi$.
{(2) $\Rightarrow$ (1)} is similar:
we have to show that
$cut(\Pi\theta;\Pi'_1,\ldots \Pi'_n)$ normalizes,
we reduce it into
$cut(cut(\Pi\theta;(\Pi'_i)_{i\neq j});\Pi'_j)$
which normalizes since $\Pi'_j\in X_j^\perp$
and the left subderivation belongs to $X_j$ by hypothesis.
\end{proof}

% >>>
% <<< Interpretation
\subsection{Interpretation} \label{sec:interpmu}

We interpret formulas as reducibility candidates,
extending Girard's interpretation of MALL connectives~\cite{girard87tcs}.

\begin{definition}[Interpretation] \label{def:interp}
Let $P$ be a formula
and $\E$ an environment
mapping each $n$-ary predicate variable $p$ occurring in $P$ to a candidate.
We define by induction on $P$ a candidate
called \emph{interpretation of $P$ under $\E$} and
denoted by $\interp{P}^\E$.
\[ % VARIABLES and UNITS
\interp{p\t}^\E = \E(p)
\quad
\interp{a\vec{u}}^\E =
  \All{\infer{\vdash a\vec{v}^\perp, a\vec{v}}{}}^{\perp\perp}
\quad
\interp{\llzero}^\E = \emptyset^{\perp\perp}
\quad
\interp{\llone}^\E = \All{\infer{\vdash \llone}{}}^{\perp\perp}
\]
\vspace{-0.5cm}{\allowdisplaybreaks
\begin{eqnarray*}
% MALL
\interp{P \lltens P'}^\E &=&
  \Set{
     \infer{\vdash \Delta,\Delta',Q\lltens Q'}{
       \infer{\vdash \Delta,Q\vphantom{'}}{\Pi} &
       \infer{\vdash \Delta',Q'}{\Pi'}}
  }{
     \Pi:Q \in\interp{P}^\E, \Pi':Q' \in\interp{P'}^\E
  }^{\perp\perp} \\
\interp{P_0 \llplus P_1}^\E &=&
  \Set{
     \infer{\vdash \Delta,Q_0\llplus Q_1}{\infer{\vdash \Delta,Q_i}{\Pi}}
  }{
     i \in \{0,1\}, \Pi : Q_i \in \interp{P_i}^\E
  }^{\perp\perp} \\
% FIRST-ORDER
\interp{\exists x.~ P x}^\E &=&
  \Set{
   \infer{\vdash \Gamma, \exists x.~ Q x}{
        \infer{\vdash \Gamma, Q t}{\Pi}}
  }{
   \Pi : Q t \in \interp{P t}^\E % NOTE [Px] ?
  }^{\perp\perp} \\
\interp{u=v}^\E &=& \All{\infer{\vdash t=t}{}}^{\perp\perp}
\\
% FIXED POINTS
\interp{\mu B\t}^\E &=&
   \lfp(
      X \mapsto
      \set{\mu \Pi}{
        \Pi : B (\mu B) \vec{t'} \in [Bp\t]^{\E,p\mapsto X}
      }^{\perp\perp}
   ) \\
% NEGATIVES
\interp{P}^\E &=& (\interp{P^\perp}^\E)^\perp
   \mbox{~ for all other cases}
\end{eqnarray*}}
The validity of that definition relies on a few observations.
It is easy to check that we do only form (bi-)orthogonals of
sets of proofs that are normalizing.
More importantly, the existence of least fixed point
candidates relies on the monotonicity of interpretations,
inherited from that of operators.
More generally,
$\interp{P}^\E$ is monotonic in $\E(p)$ if $p$ occurs only positively
in $P$, and antimonotonic in $\E(p)$ if $p$ occurs only negatively.
The two statements are proved simultaneously, following the
definition by induction on $P$.
Except for the least fixed point case,
it is trivial to check that (anti)monotonicity is preserved by the
first clauses of Definition~\ref{def:interp}, and in the case
of the last clause $\interp{P}^\E = (\interp{P^\perp}^\E)^\perp$
each of our two statements is derived from the other.
Let us now consider the definition of $[\mu B \t]^\E$,
written $\lfp(\phi_\E)$ for short.
First, the construction is well-defined:
by induction hypothesis and monotonicity of $B$,
$\interp{B q \t}^{\E,q\mapsto X}$ is monotonic in $X$,
and hence $\phi_\E$ is also monotonic and admits a least fixed point.
We then show that $\lfp(\phi_\E)$ is monotonic in $\E(p)$ when
$p$ occurs only positively in $B$ | antimonotonicity would be obtained
in a symmetric way.
If $\E$ and $\E'$ differ only on $p$ and $\E(p)\subseteq \E'(p)$,
we obtain by induction hypothesis that
$\phi_\E(X)\subseteq \phi_{\E'}(X)$ for any candidate $X$, and in particular
$\phi_\E(\lfp(\phi_{\E'})) \subseteq \phi_{\E'}(\lfp(\phi_{\E'})) =
   \lfp(\phi_{\E'})$,
\ie\ $\lfp(\phi_{\E'})$ is a prefixed point of $\phi_\E$,
and thus $\lfp(\phi_\E)\subseteq \lfp(\phi_{\E'})$, that is to say
$[\mu B \t]^\E$ is monotonic in $\E(p)$.
\end{definition}

\begin{proposition}
For any $P$ and $\E$, $([P]^\E)^\perp = [P^\perp]^\E$.
\end{proposition}

\begin{proposition} \label{prop:interp_subst}
For any $P$, $\theta$ and $\E$, $\interp{P}^\E = \interp{P\theta}^\E$.
\end{proposition}

\begin{proposition}
For any $\E$, monotonic $B$ and $S$,
$\interp{B S}^\E =
 \interp{B p}^{\E{}, p\mapsto \interp{S}^\E}$.
\end{proposition}

Those three propositions are easy to prove,
the first one immediately following from Definition~\ref{def:interp}
by involutivity of both negations (on formulas and on candidates),
the other two by induction (respectively on $P$ and $B$).
Proposition~\ref{prop:interp_subst} has an important
consequence: $\Pi\in[P]$ implies $\Pi\theta\in [P\theta]$,
\ie\ our interpretation is independent of first-order aspects.
This explains some probably surprising parts of the definition
such as the interpretation of least fixed points, where it
seems that we are not allowing the parameter of the fixed point
to change from one instance to its recursive occurrences.

In the following,
when the term structure is irrelevant or confusing, we shall write
$\interp{S}^\E$ for $\interp{S\t}^\E$.
For a predicate operator expression
$(\lambda \vec{p}.~ B\vec{p})$ of first-order arity $0$,
we shall write $\interp{B}^\E$ for
$\vec{X} \mapsto \interp{B\vec{p}}^{\E,(p_i\mapsto X_i)_i}$.
When even more concision is desirable,
we may also write $\interp{B\vec{X}}^\E$ for
$\interp{B}^\E\vec{X}$.
Finally, we simply write $[P]$ and $[B]$ when $\E$ is empty.

\begin{lemma} \label{lem:fun}
  Let $X$ and $Y$ be two reducibility candidates,
  and $\Pi$ be a proof of $\vdash P\x,Q\x$ that is $(X,Y)$-reducible.
  Then $F_B(\Pi)$ is $([B]X,[\B]Y)$-reducible.
\end{lemma}

\begin{lemma} \label{lem:nu}
  Let $X$ be a candidate
  and $\Theta$ a derivation of $\vdash S\x^\perp, \B S \x$ that
  is $(X^\perp,[\B]X)$-reducible.
  Then
  $\nu(Id_{S\t},\Theta)$ is $(X^\perp,[\nu \B\t])$-reducible for any $\t$.
\end{lemma}

% <<<
\begin{proof}[of Lemmas~\ref{lem:fun} and \ref{lem:nu}]
We prove them simultaneously, generalized as follows
for any monotonic operator $B$ of second-order arity $n+1$,
and any predicates $\vec{A}$ and candidates $\vec{Z}$:
\begin{enumerate}
\item For any $(X,Y)$-reducible $\Pi$,
  $F_{B\vec{A}}(\Pi)$ is
    $([B]\vec{Z}X,[\overline{B}]\vec{Z}^\perp Y)$-reducible.
\item For any $(X^\perp,[\B]\vec{Z}^\perp X)$-reducible $\Theta$,
  $\nu(Id_{S\t},\Theta)$ is
  $(X^\perp,[\nu(\B\vec{Z}^\perp)\t])$-reducible.
\end{enumerate}
We proceed by induction on $B$: we first establish (1),
relying on strictly smaller instances of both (1) and (2);
then we prove (2) by relying on (1) for the same $B$
(modulo size-preserving first-order details).
The purpose of the generalization is to separate the main part of $B$
from auxiliary parts $\vec{A}$, which may be large
and whose interpretations $\vec{Z}$ may depend on $X$ and $Y$,
but play a trivial role.

\begin{enumerate}
\item
  If $B$ is of the form $(\lambda \vec{p} \lambda q.~ B'\vec{p})$,
  then $F_{B\vec{A}}(\Pi)$ is simply $Id_{B'\vec{A}}$,
  which is trivially $([B'\vec{Z}],[\overline{B'}\vec{Z}^\perp])$-reducible
  since $[\overline{B'}\vec{Z}^\perp] = [B'\vec{Z}]^\perp$.
  If $B$ is of the form
  $(\lambda\vec{p}\lambda q.~ q\t)$,
  then $F_{B\vec{A}}(\Pi)$
  is $\Pi[\t/\x]$ which is $(X,Y)$-reducible.

  Otherwise, $B$ starts with a logical connective.
  Following the definition of $F_B$, dual connectives are treated
  in a symmetric way.
  The tensor case essentially consists in showing that
  if $\Pi'\vdash~P',Q'$ is $([P'],[Q'])$-reducible and
  $\Pi''\vdash~P'',Q''$ is $([P''],[Q''])$-reducible then
  the following derivation is $([P'\lltens P''], [Q'\llpar Q''])$-reducible:
  \[ \infer[\llpar]{\vdash P'\lltens P'', Q'\llpar Q''}{
      \infer[\lltens]{\vdash P'\lltens P'',Q',Q''}{
        \infer{\vdash P',Q'}{\Pi'} &
        \infer{\vdash P'',Q''}{\Pi''}}} \]
  \begin{longitem}
  \item
  The subderivation $\Pi'\lltens\Pi''$ is
  $([P'\lltens P''],[Q'],[Q''])$-reducible: By 
  Proposition~\ref{prop:red-interp} it suffices to show that
  for any $\theta$ and compatible $\Xi'\in[Q']^\perp$ and $\Xi''\in[Q'']^\perp$,
  $cut(\Pi\theta;\Xi',\Xi'')$ belongs to $[P'\lltens P'']$.
  This follows from: the fact that it reduces to
  $cut(\Pi'\theta;\Xi')\lltens cut(\Pi''\theta;\Xi'')$;
  that those two conjuncts are respectively in $[P']$ and $[P'']$ by 
  hypothesis;
  and that $\set{u\lltens v}{u\in[P'], v\in[P'']}$ is a subset
  of $[P'\lltens P'']$ by definition of the interpretation.
  \item
  We then prove that the full derivation, instantiated by $\theta$
  and cut against any compatible $\Xi\in[P'\lltens P'']^\bot$,
  is in $[Q'\llpar Q'']$.
  Since the interpretation of ${\llpar}$
  is $\set{u\lltens v}{u\in[Q']^\perp, v\in[Q'']^\perp}^\perp$,
  it suffices to show that
  % NOTE that we can't remove \sigma, because the set is not a candidate
  $cut((\llpar(\Pi'\lltens\Pi''))\theta;\Xi)$ normalizes (which
  follows from the reducibility of $\Pi'\lltens\Pi''$) and that
  for any substitutions $\sigma$ and $\sigma'$,
  $cut((\llpar(\Pi'\lltens\Pi''))\theta;\Xi)\sigma$ normalizes when cut against
  any such compatible $(u\lltens v)\sigma'$.
  Indeed, that cut reduces, using cut permutations and
  the main multiplicative reduction, into 
  $cut(cut((\Pi'\lltens\Pi'')\theta\sigma;\Xi\sigma);u\sigma',v\sigma')$
  which normalizes by reducibility of $\Pi'\lltens\Pi''$.
  \end{longitem}

  The additive case follows the same outline.
  There is no case for units, including $=$ and $\neq$,
  since they are treated with all formulas where $p$ does not occur.

  In the case of first-order quantifiers,
  say $B = \lambda \vec{p} \lambda q.~ \exists x.~ B' \vec{p} q x$,
  we essentially have to show that,
  assuming that $\Pi$ is $([P x], [Q x])$-reducible,
  the following derivation is
  $([\exists x.~ P x], [\forall x.~ Q x])$-reducible:
\[
    \infer[\forall]{\Sigma;\vdash \exists x.~ P x, \forall x.~ Q x}{
      \infer[\exists]{\Sigma,x;\vdash \exists x.~ P x, Q x}{
      \infer{\Sigma,x;\vdash P x, Q x}{\Pi}}}
\]
  \begin{longitem}
  \item
  We first establish that the immediate subderivation $\exists(\Pi)$
  is reducible, by considering $cut(\exists(\Pi)\theta;\Xi)$
  for any $\theta$ and compatible $\Xi\in[Q x]^\perp$.
  We reduce that derivation into
  $\exists(cut(\Pi\theta;\Xi))$ and conclude by 
  definition of $[\exists x.~ P x]$ and the fact that 
  $cut(\Pi\theta;\Xi)\in [P x]$.
  \item
  To prove that $\forall(\exists(\Pi))$ is reducible,
  we show that $cut(\forall(\exists(\Pi))\theta;\Xi)$ belongs to
  $[\forall x.~ Q x]$ for any $\theta$
  and compatible $\Xi \in [\exists x.~ P x]^\perp$.
  Since $[\forall x.~ Q x] = \set{\exists \Xi'}{\Xi'\in [Q t]^\perp}^\perp$,
  this amounts to show that our derivation normalizes
  (which follows from the reducibility of $\exists(\Pi)$)
  and that
  $cut(cut(\forall(\exists(\Pi))\theta;\Xi)\sigma;(\exists\Xi')\sigma')$
  normalizes for any $\sigma$, $\sigma'$ and compatible $\Xi' \in [Q t]^\perp$.
  Indeed, this derivation reduces,
  by permuting the cuts and performing the main $\forall/\exists$ reduction,
  into $cut(\exists(\Pi)\theta\sigma[t\sigma'/x];\Xi'\sigma',\Xi\sigma)$,
  which normalizes by reducibility of $\exists(\Pi)$.
  \end{longitem}

Finally, we show the fixed point case in full details since this
is where the generalization is really useful.
When $B$ is of the form $\lambda \vec{p} \lambda q.~ \mu(B'\vec{p}q)\t$,
we are considering the following derivation:
\[ \hspace{-0.3cm}\infer[\nu]{
      \vdash \mu(B'\vec{A}P)\t, \nu(\overline{B'}\vec{A}^\perp Q) \t}{
      \infer[init]{
        \vdash \mu(B'\vec{A}P)\t, \nu(\overline{B'}\vec{A}^\perp P^\perp)\t}{} &
      \infer[\mu]{
        \vdash \mu(B'\vec{A} P)\x,
             \overline{B'}\vec{A}^\perp Q
                (\nu(\overline{B'}\vec{A}^\perp P^\perp))\x}{
        \infer{\vdash B'\vec{A} P (\mu (B'\vec{A} P))\x,
             \overline{B'}\vec{A}^\perp Q
                (\nu(\overline{B'}\vec{A}^\perp P^\perp))\x}{
          F_{B'\vec{A}{\bullet}(\mu(B'\vec{A}P))\x}(\Pi)}}} \]
We apply induction hypothesis (1) on
$B'' := (\lambda \vec{p} \lambda p_{n+1} \lambda q.~ B' \vec{p} q p_{n+1} \x)$,
with
$A_{n+1} := \mu (B'\vec{A}P)$ and $Z_{n+1} := [\mu (B'\vec{Z}X)]$,
obtaining that the subderivation $F_{\ldots}(\Pi)$ is
$([B'']\vec{Z} Z_{n+1} X,
  [\overline{B''}]\vec{Z}^\perp Z_{n+1}^\perp Y)$-reducible.
Then, we establish that $\mu(F_{\ldots}(\Pi))$
is reducible: for any $\theta$ and compatible
$\Xi\in [B'']\vec{Z}Z_{n+1}Y^\perp$,
$cut(\mu(F_{\ldots}(\Pi))\theta;\Xi)$ reduces to
$\mu(cut(F_{\ldots}(\Pi)\theta;\Xi))$ which
belongs to
$ [\mu (B'\vec{Z}X)\x] =
 \set{\mu \Pi'}{
   \Pi' \in [B'\vec{Z}X(\mu(B'\vec{Z}X))\x]}^{\perp\perp} $
by reducibility of $F_{\ldots}(\Pi)$.
% since $cut(F_{\ldots}(\Pi);\Xi)$ belongs to
% $[B'\vec{A}P(\mu(B'\vec{A}P))\x]$ by induction hypothesis (1).
We finally obtain the reducibility of the whole derivation by
applying induction hypothesis (2) on $B'$ with $A_{n+1}:=Q^\perp$,
$Z_{n+1} := Y^\perp$ and $X := \interp{\mu (B'\vec{Z}X) \x}^\perp$.

\item
Here we have to show that for any $\theta$ and any compatible $\Xi\in X$,
the derivation
$cut(\nu(Id_{S\t},\Theta)\theta;\Xi)$ belongs to $[\mu(B\vec{Z})\t]^\perp$.
Since only $\t$ is affected by $\theta$ in such derivations,
we generalize on it directly, and consider the following set:
\[ Y := \set{cut(\nu(Id_{S\vec{t'}},\Theta);\Xi)}{\Xi: S\vec{t'}\in X}^\perp \]
Note that we can form the orthogonal to obtain $Y$, since we are indeed
considering a subset of $\WN$: any $cut(\nu(Id;\Theta);\Xi)$ reduces
to $\nu(\Xi;\Theta)$, and $\Xi$ and $\Theta$ normalize.
We shall establish that $Y$ is a pre-fixed point of the operator $\phi$
such that $[\mu(B\vec{Z})\t]$ has been defined as $\lfp(\phi)$,
from which it follows that $[\mu(B\vec{Z})\t]\subseteq Y$,
which entails our goal |
note that this is essentially a proof by induction on $[\mu(B\vec{Z})]$.

So we prove the pre-fixed point property:
\[ \set{\mu\Pi}{\Pi: B\vec{A}(\mu (B\vec{A}))\vec{t''}
                 \in[B\vec{Z} Y \t']}^{\perp\perp} \subseteq Y \]
Observing that, for any $A,B\subseteq\WN$, we have
  $A^{\perp\perp} \subseteq B^\perp \Leftrightarrow
   A^{\perp\perp} \Perp B \Leftrightarrow
   B \subseteq A^\perp \Leftrightarrow
   B \Perp A$,
our property can be rephrased equivalently:
\[ \set{cut(\nu(Id_{S\vec{t'}},\Theta);\Xi)}{\Xi:S\vec{t'}\in X}
 \Perp \set{\mu \Pi}{\Pi\in [B\vec{Z}Y\t']} \]
Since both sides are stable by substitution, there is no need
to consider compatibility substitutions here, and it suffices to consider cuts
between any compatible left and right-hand side derivations:
$cut(cut(\nu(Id,\Theta);\Xi);\mu \Pi)$.
It reduces, using cut exchange, the main fixed point reduction
and finally the identity reduction, into:
\[ \hspace{-0.5cm}\infer[cut]{\vdash \Gamma,\Delta}{
     \infer{\vdash \Gamma, S\vec{t'}}{\Xi} &
     \infer[cut]{\vdash S^\perp\vec{t'}, \Delta}{
       \infer{\vdash S^\perp\vec{t'}, \B\vec{A}^\perp S\vec{t'}}{
         \Theta[\vec{t'}/\x]} &
       \infer[cut]{\vdash B\vec{A} S^\perp \vec{t'}, \Delta}{
         \infer{\vdash B\vec{A} S^\perp \vec{t'},
                       \B\vec{A}^\perp(\nu (\B\vec{A}^\perp))\vec{t'}}{
                F_{B\vec{A}{\bullet}\vec{t'}}(\nu(Id_{S\x},\Theta))} &
         \infer{\vdash B\vec{A}(\mu (B\vec{A}))\vec{t'},\Delta}{\Pi}
       }
     }
   } \]
By hypothesis, $\Xi\in X$,
$\Pi\in[B\vec{Z}Y\vec{t'}]$
and $\Theta[\vec{t'}/\x]$ is $(X^\perp,[\B\vec{Z}^\perp X \t'])$-reducible.
Moreover,
$\nu(Id_{S\x},\Theta)$ is $(X^\perp,Y^\perp)$-reducible by definition
of $Y$,
and thus, by applying (1) on the operator
$\lambda \vec{p} \lambda q.~ B \vec{p} q \vec{t'}$,
which has the same size as $B$, we obtain that
$F_{B\vec{A}{\bullet}\vec{t'}}(\nu(Id_{S\x},\Theta))$ is
$([B\vec{Z} X^\perp \vec{t'}],
 [\B\vec{Z}^\perp Y^\perp \vec{t'}])$-reducible\footnote{
  This use of (1) involving $Y$ is the reason why our two lemmas
  need to deal with arbitrary candidates
  and not only interpretations of formulas.
}.
We can finally compose all that to conclude that our derivation normalizes.
\end{enumerate}
\vspace{-0.5cm}\end{proof}
% >>>

\subsection{Normalization}

\begin{lemma} \label{lem:mainbis}
Any proof of $\;\vdash P_1,\ldots,P_n$ is $([P_1],\ldots,[P_n])$-reducible.
\end{lemma}

\begin{proof}
By induction on the height of the derivation $\Pi$,
with a case analysis on the first rule.
We are establishing that for any $\theta$ and compatible
$(\gamma_i\in [P_i]^\perp)_{i=1\ldots n}$,
$cut(\Pi\theta;\vec{\gamma})$ normalizes.
If $\Pi\theta$ is an axiom on $P \equiv P_1\theta \equiv P_2^\perp\theta$,
the cut against a proof of $[P]$ and a proof of $[P]^\perp$
reduces into a cut between those two proofs, which normalizes.
If $\Pi\theta = cut(\Pi'\theta;\Pi''\theta)$ is a cut on the formula $P$,
$cut(\Pi\theta;\vec{\gamma})$ reduces to
$cut(cut(\Pi'\theta;\vec{\gamma}');cut(\Pi''\theta;\vec{\gamma}''))$
and the two subderivations belong to dual
candidates $[P]$ and $[P]^\perp$ by induction hypothesis
and Proposition~\ref{prop:red-interp}.

Otherwise, $\Pi$ starts with a rule from the logical group,
the end sequent is of the form $\vdash\Gamma,P$ where $P$ is the principal
formula, and we shall prove that $cut(\Pi\theta;\vec{\gamma}) \in [P]$
when $\vec{\gamma}$ is taken
in the duals of the interpretations of $\Gamma\theta$,
which allows to conclude again using Proposition~\ref{prop:red-interp}.

\begin{longitem}
\item
  The rules $\llone$, $\lltens$, $\llplus$, $\exists$, $=$
  and $\mu$ are treated similarly,
  the result coming directly from the definition of the interpretation.

  Let us consider, for example, the fixed point case: $\Pi = \mu \Pi'$.
  By induction hypothesis, $cut(\Pi'\theta;\vec{\gamma})\in [B(\mu B)\t]$.
  By definition,
  $[\mu B\t] = \lfp(\phi) = \phi(\lfp(\phi)) = X^{\perp\perp}$
  where $X:=\set{\mu \Xi}{\Xi\in[B{\bullet}\t][\mu B]}$.
  Since $[B(\mu B)\t] = [B{\bullet}\t][\mu B]$,
  we obtain that $\mu (cut(\Pi'\theta;\vec{\gamma})) \in X$ and thus also
  in $X^{\perp\perp}$.
  Hence
  $cut(\Pi\theta;\vec{\gamma})$, which reduces to the former, is
  also in $[\mu B\t]$.

\item
  The rules $\perp$, $\llpar$, $\top$, $\llwith$, $\forall$, ${\neq}$,
  and $\nu$ are treated similarly:
  we establish that $cut(\Pi\theta;\vec{\gamma})\Perp X$ for some $X$
  such that $[P] = X^\perp$.
  First, we have to show that our derivation normalizes, which comes
  by permuting up the cuts, and concluding by induction hypothesis |
  this requires that after the permutation the derivations $\vec{\gamma}$
  are still in the right candidates, which relies on closure under
  substitution and hence signature extension for the case of disequality
  and $\forall$.
  Then we have to show that for any $\sigma$ and $\sigma'$, and any
  compatible $\Xi \in X$, the derivation 
  $cut(cut(\Pi\theta;\vec{\gamma})\sigma;\Xi\sigma')$ normalizes too.
  We detail this last step for two key cases.

  In the $\forall$ case we have
  $[\forall x.~ P x] = \set{\exists\Xi'}{\Xi'\in [P t^\perp]}^\perp$,
  so we consider
  $cut(cut((\forall\Pi')\theta;\vec{\gamma})\sigma;(\exists\Xi')\sigma')$,
  which reduces to $cut(\Pi'\theta[t/x];$ $\vec{\gamma}\sigma,\Xi\sigma')$.
  % Note that $\theta$ and $[t/x]$ commute since $x$ is freshly bound in $\Pi'$.
  This normalizes by induction hypothesis on
  $\Pi'[t/x]$, which remains smaller than $\Pi$.
  
  The case of $\nu$ is the most complex, but is similar to the
  argument developed for Lemma~\ref{lem:nu}.
  If $\Pi$ is of the form $\nu(\Pi',\Theta)$ and $P \equiv \nu B \t$ then
  $cut(\Pi;\gamma)\theta$ has type $\nu B \vec{u}$ for $u := \t\theta$.
  Since $[\nu B\vec{u}] =
  \set{\mu \Xi}{\Xi\in[\B{\bullet}\vec{u}][\mu \B]}^\perp$,
  we show that for any $\sigma$, $\sigma'$ and compatible
  $\Xi\in[\B(\mu\B)\vec{u}]$,
  the derivation
  $cut(cut(\nu(\Pi',\Theta)\theta;\vec{\gamma})\sigma;(\mu \Xi)\sigma')$
  normalizes.
  Let $\vec{v}$ be $\vec{u}\sigma$,
  the derivation reduces to:
  \[ cut(cut(\Pi'\theta\sigma;\vec{\gamma}\sigma);
       cut(\Theta[\vec{v}/\x];
           cut(F_{\B{\bullet}\vec{v}}(\nu(Id,\Theta));\Xi\sigma'))) \]
  By induction hypothesis,
  $cut(\Pi'\theta\sigma;\vec{\gamma}\sigma)\in [S\vec{v}]$,
  and $\Theta$ is $([S\x]^\perp,[BS\x])$-reducible.
  By Lemmas~\ref{lem:fun} and \ref{lem:nu} we obtain
  that $F_{\B{\bullet}\vec{v}}(\nu(Id,\Theta))$ is
  $([\B{}S^\perp\vec{v}],[B(\nu B)\vec{v}])$-reducible.
  Finally, $\Xi\in[\B(\mu\B)\vec{v}]$.
  We conclude by composing all these reducibilities
  using Proposition~\ref{prop:red-interp}.
\end{longitem}
\vspace{-0.5cm}\end{proof}

\begin{theorem}[Cut elimination]
Any derivation can be reduced into a cut-free derivation.
\end{theorem}

\begin{proof}
By Lemma~\ref{lem:mainbis}, any derivation is reducible, and hence normalizes.
\end{proof}

The usual immediate corollary of the cut elimination result is that
\mumall\ is consistent, since there is obviously no cut-free derivation
of the empty sequent.
However, note that unlike in simpler logics, cut-free derivations
do not enjoy the subformula property,
because of the $\mu$ and $\nu$ rules.
While it is easy to characterize the new formulas that can arise from $\mu$,
nothing really useful can be said for $\nu$,
for which no non-trivial restriction is known.
Hence, \mumall\ only enjoys restricted forms of the subformula property,
applying only to (parts of) derivations that do not involve coinductions.

% >>>

%% file: focus.tex
% NOTE
%   In here I need to talk about initial rules in general, not only
%   the axiom. So I call $init$ the axiom.

\label{sec:foc_mumall}

In \cite{andreoli92jlc}, Andreoli identified some important structures
in linear logic, which led to the design of his focused proof system.
This complete proof system for (second-order) linear logic structures
proofs in stripes of \emph{asynchronous} and \emph{synchronous} rules.
Choices in the order of application of asynchronous rules do not matter,
so that the real non-determinism lies in the synchronous phase.
However, the focused system tames this non-determinism by forcing
to hereditarily chain these choices: once the focus is set on a
synchronous formula, it remains on its subformulas
as its connectives are introduced, and so on,
to be released only on asynchronous subformulas.
We refer the reader to \cite{andreoli92jlc} for a complete
description of that system,
but note that Figure~\ref{fig:focused}, without the fixed point rules,
can be used as a fairly good reminder:
it follows the exact same structure, only missing the rules for exponentials.

Focusing \mumall\ can be approached simply by reading the
focusing of second-order linear logic through the encoding of fixed points.
But this naive approach yields a poorly
structured system.
Let us recall the second-order encoding of $\mu B \t$:
\[ \forall S.~ \oc(\forall \x.~ B S \x \llimp S \x) \llimp S \t \]
This formula starts with a layer of asynchronous connectives:
$\forall$, $\llimp$ and ${\wn}$, the dual of $\,{\oc}$.
Once the asynchronous layer has been 
processed, the second-order eigenvariable $S$ represents $\mu B$
and one obtains unfoldings of $S$ into $BS$
by focusing on the pre-fixed point hypothesis.
Through that encoding, one would thus obtain a system where several unfoldings
necessarily require several phase alternations.
This is not satisfying:
the game-based reading of focusing identifies fully synchronous (positive)
formulas with data types, which should be built in one step by the player,
\ie\ in one synchronous phase.
In \mumall, least fixed points over fully synchronous operators
should be seen as data types.
That intuition, visible in previous examples, is also justified
by the classification of connectives in Definition~\ref{def:connectives},
and is indeed accounted for in the focused system
presented in Figure~\ref{fig:focused}.

It is commonly believed that asynchrony corresponds to invertibility.
The two notions do coincide in many cases but it should not be taken too
seriously, since this does not explain, for example,
the treatment of exponentials,
or the fact that $init$ has to be synchronous while it is trivially
invertible.
  % NOTE even when there's weakening
  % NOTE This must be, since there might be different ways to apply init
  % for a given positive atom in a proof, and it is not possible to 
  % (this matter when doing proof search, with logic variables)
  % and unlike fixed points those different uses cannot be unified
  % into a single use in a given proof
  % (without violating the essence of the derivation,
  %  or even the first-order structure when we have disequality).
In the particular case of fixed points,
invertibility is of no help in designing a complete focused proof system.
Both $\mu$ and $\nu$ are invertible (in the case of $\nu$, this is
obtained by using the unfolding coinvariant) but
this does not capture the essential aspect of fixed points, that is
their infinite behavior.
As a result, a system requiring that the $\mu$ rule is applied
whenever possible would not be complete, notably failing
on $\;\vdash \top \lltens \llone, \mu p. p$ or $\;\vdash nat~x \llimp nat~x$.
% A natural proof of the latter sequent consists in an axiom.
% From that proof, we can obtain a proof starting by the $\mu$ rule,
% since it is invertible:
% it is the $\eta$-expansion of the axiom.
% Such invertibility at the cost of repeated expansions is of no help
% when attempting to design a complete proof system.
As we shall see, the key to obtaining focused systems is to consider
the permutability of asynchronous rules, rather than their invertibility,
as the fundamental guiding principle.

We first design the $\mu$-focused system in Section~\ref{sec:mufoc},
treating $\mu$ synchronously,
which is satisfying for several reasons starting with its positive nature.
We show in Section~\ref{sec:nufoc} that it is also possible to consider
a focused system for \mumall\ where $\nu$ is treated synchronously.
In Section~\ref{sec:foc_mulj}, we apply the $\mu$-focused system
to a fragment of \muLJ.

%% %%%%%%%%%%%%%%%%%%%%%%%%%%%%%%%%%%%%%%%%%%%%%%%%%%%%%%%%%%%%%%%%%%%%%%%%%%*
\subsection{A complete $\mu$-focused calculus} \label{sec:mufoc}

In this section,
we call \emph{asynchronous} (\resp \emph{synchronous})
the negative (\resp positive) connectives of Definition~\ref{def:connectives}
and the formulas whose top-level connective is asynchronous (\resp 
synchronous).
Moreover, we classify non-negated atoms as synchronous and negated
ones as asynchronous. As with Andreoli's original system, this latter choice
is arbitrary and can easily be changed for a case-by-case 
assignment~\cite{miller07cslb,chaudhuri08jar}.

We present the system in Figure~\ref{fig:focused} as a good
candidate for a focused proof system for \mumall.
In addition to asynchronous and synchronous formulas as defined above,
focused sequents can contain \emph{frozen formulas} $P^*$
where $P$ is an asynchronous atom or fixed point.
Frozen formulas may only be found at toplevel in sequents.
We use explicit annotations of the sequents in the style of Andreoli:
in the synchronous phase, sequents have the form
$\;\vdash \Gamma \Downarrow P$;
in the asynchronous phase, they have the form
$\;\vdash \Gamma \Uparrow \Delta$.
In both cases,
$\Gamma$ and $\Delta$ are sets of formulas of disjoint locations,
and $\Gamma$ is a multiset of synchronous or frozen formulas.
The convention on $\Delta$ is a slight departure from Andreoli's
original proof system where $\Delta$ is a list: we shall emphasize
the irrelevance of the order of asynchronous rules without
forcing a particular, arbitrary ordering.
Although we use an explicit freezing annotation,
our treatment of atoms is really the same one as Andreoli's;
the notion of freezing is introduced here as a technical device for
dealing precisely with fixed points,
and we also use it for atoms for a more uniform presentation.

\begin{figure}[htpb]
\begin{center}
$\begin{array}{c}
\mbox{Asynchronous phase}
\\[6pt]
\infer{\vdash\Gamma\Uparrow P\llpar Q,\Delta}{\vdash\Gamma\Uparrow P,Q,\Delta}
\quad
\infer{\vdash\Gamma\Uparrow P\llwith Q, \Delta}{
  \vdash\Gamma\Uparrow P,\Delta & \vdash\Gamma\Uparrow Q,\Delta
}
\quad
\infer{\vdash\Gamma\Uparrow a^\perp\t, \Delta}{
  \vdash\Gamma,(a^\perp\t)^*\Uparrow\Delta}
\\[6pt]
\infer{\vdash \Gamma \Uparrow \bot, \Delta}{\vdash \Gamma\Uparrow\Delta}
\quad
\infer{\vdash \Gamma \Uparrow \top, \Delta}{}
\quad
\infer{\vdash \Gamma \Uparrow s\neq t, \Delta}{
 \{ \vdash \Gamma\theta \Uparrow \Delta\theta :
      \theta\in csu(s\unif t) \} }
\\[6pt]
\infer{\vdash\Gamma\Uparrow\forall{}x. P x,\Delta}{
  \vdash\Gamma\Uparrow P c,\Delta}
\\[6pt]
\infer{\vdash\Gamma\Uparrow \nu{}B\t,\Delta}{
  \vdash\Gamma\Uparrow S\t,\Delta &
  \vdash\Uparrow BS\x, S\x^\bot
}
\quad
\infer{\vdash\Gamma\Uparrow \nu{}B\t,\Delta}{
  \vdash\Gamma,(\nu{}B\t)^*\Uparrow\Delta}
\end{array}
$

\vspace{10pt}
$
\begin{array}{c}
\mbox{Synchronous phase}
\\[6pt]
\infer{\vdash\Gamma,\Gamma'\Downarrow P\lltens Q}{
  \vdash\Gamma\Downarrow P &
  \vdash \Gamma'\Downarrow Q
}
\quad
\infer{\vdash \Gamma\Downarrow P_0\llplus P_1}{\vdash\Gamma\Downarrow P_i}
\quad
\infer{\vdash (a^\perp\t)^*\Downarrow a\t}{}
\\[6pt]
\infer{\vdash \Downarrow \llone}{} \quad \infer{\vdash \Downarrow t=t}{}
\\[6pt]
\infer{\vdash \Gamma\Downarrow\exists{}x. P x}{\vdash\Gamma\Downarrow P t}
\\[6pt]
\infer{\vdash\Gamma\Downarrow \mu{}B\t}{\vdash\Gamma\Downarrow B(\mu{}B)\t}
\quad
\infer{\vdash (\nu{}\B\t)^*\Downarrow \mu{}B\t}{}
\end{array}$

\vspace{10pt}
Switching rules (where $P$ is synchronous, $Q$ asynchronous)
\[
\infer{\vdash\Gamma\Uparrow P,\Delta}{\vdash\Gamma,P\Uparrow\Delta}
\quad
\infer{\vdash \Gamma, P \Uparrow}{\vdash \Gamma \Downarrow P}
\quad
\infer{\vdash \Gamma \Downarrow Q}{\vdash \Gamma \Uparrow Q}
\]
\end{center}
\caption{The $\mu$-focused proof-system for \mumall}
\label{fig:focused}
\end{figure}

The $\mu$-focused system extends the usual focused system for MALL.
The rules for equality are not surprising,
the main novelty here is the treatment of fixed points.
% Depending on the operator on which the fixed point is formed,
% both $\mu$ and $\nu$ rules can be applied any number of times
% | but not with any coinvariant concerning $\nu$.
% Notice for example that an instance of $\mu\nu$ can be $\eta$-expanded
% into a larger derivation,
% unfolding both fixed points to apply $\mu\nu$ on the recursive occurrences.
Each of the fixed point connectives has two rules in the focused 
system:  one treats it ``as an atom'' and the other one as an expression with
internal logical structure.
In accordance with Definition~\ref{def:connectives},
$\mu$ is treated during the synchronous phase
and $\nu$ during the asynchronous phase.

Roughly, what the focused system implies is that
if a proof involving a $\nu$-expression
proceeds by coinduction on it, then this coinduction can be done at the 
beginning;
otherwise that formula can be ignored in the whole derivation,
except for the $init$ rule.
The latter case is expressed by the rule which moves the greatest fixed
point to the left zone, freezing it.
Focusing on a $\mu$-expression yields two choices: unfolding or applying the 
initial rule for fixed points.
If the considered operator is fully synchronous, the focus will never be lost.
For example, if $nat$ is the (fully synchronous) expression
$\mu N. \lambda{}x.~ x=0 \llplus \exists{}y.~ x=s~y \lltens N~y$,
then focusing puts a lot of structure on a proof of
$\vdash \Gamma\Downarrow nat~t$: 
either $t$ is a closed term representing a natural number and $\Gamma$ is empty,
or $t = s^n t'$ for some $n\geq 0$ and $\Gamma$ only contains $(nat~t')^\bot$.

We shall now establish the completeness of our focused proof system:
If the unfocused sequent $\;\vdash\Gamma$ is provable then so is
$\;\vdash\Uparrow\Gamma$, and the order of application of asynchronous
rules does not affect provability.
From the perspective of proofs rather than provability,
we are actually going to provide transformations from unfocused to focused
derivations (and back) which can reorder asynchronous rules arbitrarily.
% <<< quasi finite
However, this result cannot hold without a simple condition
avoiding pathological uses of infinite branching, as illustrated with
the following counter-example.
The unification problem $s~(f~0)\unif f~(s~0)$, where $s$ and $0$ are constants,
has infinitely many solutions $[(\lambda x.~ s^n x) / f]$.
Using this, we build a derivation $\Pi_\omega$
with infinitely many branches, each $\Pi_n$
unfolding a greatest fixed point $n$ times:
\[
 \Pi_0 \eqdef \infer[\top]{\vdash \nu p. p, \top}{} \quad\quad
 \Pi_{n+1} \eqdef \infer[\nu]{\vdash \nu p. p, \top}{
                  \infer{\vdash \nu p. p, \top}{\Pi_n} &
                  \infer[init]{\vdash \mu p. p, \nu p. p}{}} \]
 \[ \Pi_\omega \eqdef \infer[\neq]{
                     f; \vdash s~(f~0) \neq f~(s~0), \nu p. p, \top}{
                     \Pi_0 & \Pi_1 & \ldots & \Pi_n & \ldots} \]
Although this proof happens to be already in a focused form,
in the sense that focusing annotations can be added in a straightforward
way,
the focused transformation must also provide a way to change the order
of application of asynchronous rules.
In particular it must allow to permute down the introduction of
the first $\nu p. p$. The only reasonable way to do so is as follows,
% SPACE ("first nu p.p / outermost application: why is it duplicated")
expanding $\Pi_0$ into $\Pi_1$ and then pulling down the $\nu$ rule
from each subderivation, changing $\Pi_{n+1}$ into $\Pi_n$:
\[ \Pi_\omega \quad\rightsquigarrow\quad
      \infer[\nu]{f; \vdash s~(f~0) \neq f~(s~0), \nu p. p, \top}{
         \infer{f; \vdash s~(f~0) \neq f~(s~0), \nu p. p, \top}{\Pi_\omega} &
         \infer[init]{\vdash \mu p. p, \nu p. p}{}} \]
This leads to a focusing transformation that may not terminate.
% explain the link between quasi-finiteness and pi_w (SPACE)
The fundamental problem here is that although each additive
branch only finitely explores the asynchronous formula $\nu p.p$,
the overall use is infinite.
A solution would be to admit infinitely deep derivations,
with which such infinite balancing process may have a limit.
% Although this is actually interesting~\cite{baelde10pstt},
But our goal here is to develop finite proof representations
(this is the whole point of (co)induction rules)
so we take an opposite approach and require a minimum amount
of finiteness in our proofs.

\begin{definition}[Quasi-finite derivation]
A derivation is said to be quasi-finite if
it is cut-free,
has a finite height
and only uses a finite number of different coinvariants.
%  The first assumption is only used in the proof of balancing,
%  but it is crucial there.
%  The second assumption is not implied by the first because
%  of infinite branching, and is unnecessary in the $\nu$-focused
%  system.
\end{definition}

This condition may seem unfortunate, but it appears to be essential
when dealing with transfinite proof systems involving fixed points. More
precisely, it is related to the choice regarding the introduction of
asynchronous fixed points, be they greatest fixed points in $\mu$-focusing
or least fixed points in $\nu$-focusing.
Note that quasi-finiteness is trivially satisfied
for any cut-free derivation that is finitely branching,
and that any derivation which does not involve the $\neq$ rule
can be normalized into a quasi-finite one.
Moreover,
quasi-finiteness is a natural condition from a practical perspective,
for example in the context of automated or interactive theorem proving,
where $\neq$ is restricted to finitely branching instances anyway.
However, it would be desirable to refine the notion of quasi-finite
derivation in a way that allows cuts and is preserved by cut elimination,
so that quasi-finite proofs could be considered a proper proof fragment.
Indeed, the essential idea behind quasi-finiteness is that
only a finite number of locations are explored in a proof,
and the cut-free condition is only added because cut reductions
do not obviously preserve this.
We conjecture that a proper, self-contained notion of quasi-finite
derivation can be attained,
but leave this technical development to further work.

% >>>

The core of the completeness proof follows~\cite{miller07cslb}.
This proof technique proceeds by transforming standard derivations
into a form where focused annotations can be added to obtain a focused
derivation. Conceptually, focused proofs are simply special cases
of standard proofs, the annotated sequents of the focused proof system
being a concise way of describing their shape.
The proof transformation proceeds by iterating two lemmas which
perform rule permutations: the first lemma expresses that 
asynchronous rules can always be applied first, while the second one
expresses that synchronous rules can be applied in a hereditary fashion
once the focus has been chosen.
The key ingredient of \cite{miller07cslb} is the notion of focalization
graph, analyzing dependencies in a proof and showing that there is always
at least one possible focus.

In order to ease the proof, we shall consider an intermediate
proof system whose rules enjoy a one-to-one correspondence with
the focused rules.
This involves getting rid of the cut, non-atomic axioms,
and also explicitly performing freezing.

\begin{definition}[Freezing-annotated derivation]
The freezing-annotated variant of \mumall\ is obtained by
removing the cut rule,
enriching the sequent structure with an annotation for frozen fixed points
or atoms,
restricting the initial rule to be applied only on frozen asynchronous
formulas,
and adding explicit annotation rules:
\[
   \infer{\vdash \freeze{a^\perp\t}, a\t}{}
\quad\quad
   \infer{\vdash \freeze{\nu\B\t}, \mu B \t}{}
\quad\quad
   \infer{\vdash \Gamma, \nu B \t}{\vdash \Gamma, \freeze{\nu B \t}}
\quad\quad
   \infer{\vdash \Gamma, a^\perp\t}{\vdash \Gamma, \freeze{a^\perp\t}}
\]

Atomic instances of $init$ can be translated into freezing-annotated
derivations:
\[ \infer{\vdash \nu B\t, \mu \B\t}{}
\quad \longrightarrow \quad
   \infer{\vdash \nu B\t , \mu \B\t }{
   \infer{\vdash \freeze{\nu B\t}, \mu \B\t }{}}
\quad\quad\quad\quad
   \infer{\vdash a^\perp\t, a\t}{}
\quad \longrightarrow \quad
   \infer{\vdash a^\perp\t , a\t }{
   \infer{\vdash \freeze{a^\perp\t}, a\t }{}} \]
Arbitrary instances of $init$ can also be obtained by first expanding them
to rely only on atomic $init$, using Proposition~\ref{def:atomicinit},
and then translating atomic $init$ as shown above.
We shall denote by $init*$ this derived generalized axiom.
Any \mumall\ derivation can be transformed into a freezing-annotated one
by normalizing it and translating $init$ into $init*$.
\end{definition}

The asynchronous freezing-annotated rules (that is,
those whose principal formula is asynchronous) correspond naturally
to asynchronous rules of the $\mu$-focused system.
Similarly, synchronous freezing-annotated rules correspond to
synchronous focused rules, which includes the axiom rule.
The switching rules of the $\mu$-focused system
do not have a freezing-annotated equivalent:
they are just book-keeping devices marking phase transitions.

From now on we shall work on freezing-annotated derivations,
simply calling them derivations.

\subsubsection{Balanced derivations}

In order to ensure that the focalization process terminates, we have to 
guarantee that the permutation steps preserve some measure over derivations.
The main problem here comes from the treatment of fixed points,
and more precisely from the fact that there is a choice in the asynchronous
phase regarding greatest fixed points.
% In the focused system, a greatest fixed point can be used for coinduction
% or frozen for later use ``as an atom'' in the axiom rule.
% To reflect this step, we shall enrich our derivations with an extra-logical
% annotation for frozen fixed points.
% | this is only a slight departure from the idea of \cite{miller07cslb}.
% The heart of our problem is that
We must ensure that a given greatest fixed point formula is always used in 
the same way in all additive branches of a proof:
if a greatest fixed point is copied by an additive conjunction or $\neq$,
then it should either be used for coinduction in all branches,
or frozen and used for axiom in all branches.
Otherwise it would not be possible to permute the treatment of the 
$\nu$ under that of the $\llwith$ or $\neq$ while controlling the size of
the transformed derivation.

\begin{definition}[Balanced derivation]
A greatest fixed point occurrence is \emph{used in a balanced way}
if all of its principal occurrences are used consistently:
either they are all frozen or they are all used for coinduction, 
with the same coinvariant.
We say that a derivation is \emph{balanced} if it is quasi-finite
and all greatest fixed points occurring in it are used in a balanced way.
\end{definition}

\begin{lemma}\label{lem:invwith}
If $S_0$ and $S_1$ are both coinvariants for $B$
then so is $S_0\llplus S_1$.
%NOTE we used to say this but with cut elim it's a long call:
%Moreover the resulting coinvariance derivation has the same $\mu$-height
%as the highest original derivation of coinvariance.
\end{lemma}

\begin{proof}
Let $\Pi_i$ be the derivation of coinvariance for $S_i$.
The proof of coinvariance of $S_0\llplus S_1$ is as follows:
\[ \infer[\llwith]{\vdash S_0^\perp\x \llwith S_1^\perp\x,
                          B (S_0\llplus S_1) \x}{
   \infer{\vdash S_0^\perp\x, B (S_0\llplus S_1)\x}{
   \phi_0(\Pi_0)}
   &
   \infer{\vdash S_1^\perp\x, B (S_0\llplus S_1)\x}{
   \phi_1(\Pi_1)}
} \]
The transformed derivations $\phi_i(\Pi_i)$ are obtained by functoriality:
\[ \phi_i(\Pi_i) = \infer[cut]{\vdash S_i^\perp\x, B (S_0\llplus S_1)\x}{
                  \infer{\vdash S_i^\perp\x, B S_i \x}{\Pi_i} &
                  \infer[B]{\vdash \B S_i^\perp \x,
                                   B (S_0\llplus S_1) \x}{
                  \infer[\llplus]{\vdash S_i^\perp \y, S_0\y\llplus S_1\y}{
                  \infer[init]{\vdash S_i^\perp\y, S_i\y}{}}}} \]
Notice that after the elimination of cuts, the proof of coinvariance
that we built can be larger than the original ones:
this is why this transformation cannot be done as part of
the rule permutation process.
\end{proof}

\begin{lemma} \label{lem:balance}
Any quasi-finite derivation of $\;\vdash\Gamma$
can be transformed into a balanced
derivation of $\;\vdash\Gamma$.
\end{lemma}

\begin{proof}
We first ensure that all coinvariants used for the same (locatively identical)
greatest fixed point are the same.
For each $\nu B$ on which at least one coinduction is performed
in the proof, this is achieved by taking the union of all coinvariants
used in the derivation,
thanks to Lemma~\ref{lem:invwith},
adding to this union the unfolding coinvariant $B (\nu B)$.
Note that quasi-finiteness is needed here to
ensure that we are only combining finitely many coinvariants.
Let $S_{\nu B}$ be the resulting coinvariant,
of the form $S_0 \llplus \ldots \llplus S_n \llplus B (\nu B)$,
and $\Theta_{\nu B}$ be the proof of its coinvariance.
We adapt our derivation by
changing every instance of the $\nu$ rule as follows:
\[ \infer{\vdash \Gamma, \nu B \t}{
      \vdash \Gamma, S_i \t &
      \infer{\vdash S_i^\perp \x, B S_i \x}{\Theta_i}}
   \quad\longrightarrow\quad
   \infer{\vdash \Gamma, \nu B \t}{
     \infer=[\llplus]{\vdash
       \Gamma, S_{\nu B}\t}{
       \vdash \Gamma, S_i\t}
     &
     \infer{\vdash S_{\nu B}^\perp\x, BS_{\nu B}\x}{\Theta_{\nu B}}
  } \]

It remains to ensure that a given fixed point is either always coinducted on
or always frozen in the derivation.
% Greatest fixed points are ordered by the transitive closure of
% the sublocation relation, and that order
% is well-founded for the formulas appearing in our derivation
% since it is assumed to have finite height.
We shall balance greatest fixed points,
starting with unbalanced fixed points closest to the root,
and potentially unbalancing deeper fixed points in that process,
but without ever introducing unbalanced fixed points that were not initially
occurring in the proof.

% DETAIL more? (reviewer #3)
Let $\Pi_0$ be the derivation obtained at this point.
We define the degree of a greatest fixed point
to be the maximum distance in the sublocation ordering
to a greatest fixed point sublocation occurring in $\Pi_0$,
$0$ if there is none.
Quasi-finiteness ensures that degrees are finite,
since there are only finitely many locations occurring at toplevel
in the sequents of a quasi-finite derivation.
We shall only consider derivations in which greatest fixed points that
are coinducted on are also coinducted on with the same coinvariant
in $\Pi_0$, and maintain this condition
while transforming any such derivation into a balanced one.
We proceed by induction on
the multiset of the degrees of unbalanced fixed points in the derivation,
ordered using the standard multiset ordering |
note that degrees are well defined for all unbalanced fixed points since they
must also occur in $\Pi_0$.
%
% These two transformations restore a balanced usage for $\nu B$,
% but may introduce freezing on greatest fixed point
% subformulas of $\nu B$: this may occur in $init*$,
% but also with the $\nu B_i$ in our expansion of the erasure.
% If necessary, \ie\ if these subformulas are not frozen in $\Pi'$,
% their freezing should be corrected by applying the same 
% transformations on them.
% At this point the distinction between location and structure
% is crucial:
% speaking about structure,
% the subformula relation is cyclic on fixed points ($nat$ is a 
% subformula of $nat$), but from the locative point of view it coincides
% with the tree of sublocations.
% Moreover, that tree is only finitely
% explored in $\Pi'$, so there is a well-founded order on formulas, based on
% the distance to a leaf in that finite part of the tree.
% A multiset ordering argument on top of that measure shows that our iterative
% process terminates.
%
If there is no unbalanced fixed point, we have a balanced proof.
Otherwise, pick an unbalanced fixed point of maximal degree.
It is frozen in some branches % (of $\llwith$ or $\neq$)
and coinducted on in others.
We remove all applications of freezing on that fixed point,
which requires to adapt axioms\footnote{
  Note that instead of the unfolding coinvariant $B(\nu B)$ we could have
  used the coinvariant $\nu B$. This would yield a simpler proof,
  but that would not be so easy to adapt for $\nu$-focusing in
  Section~\ref{sec:nufoc}.
}:
\[ \infer{\vdash \freeze{\nu B \t}, \mu \B \t}{}
\quad\longrightarrow\quad
   \infer[\nu]{\vdash \nu B \t, \mu\B\t}{
   \infer=[\llplus]{\vdash S_{\nu B}\t,\mu\B\t}{
   \infer[\mu]{\vdash B (\nu B) \t, \mu\B\t}{
   \infer[init*]{\vdash B (\nu B) \t, \B(\mu\B)\t}{}}}
   & \infer{\vdash S^\perp_{\nu B}\x, BS_{\nu B}\x}{\Theta_{\nu B}}} \]
The fixed point $\nu B$ is used in a balanced way in the resulting derivation.
Our use of the derived rule $init*$ might have introduced
some new freezing rules on greatest fixed point
sublocations of $B(\nu B)$ or $\B (\mu \B)$.
Such sublocations, if already present in the proof,
may become unbalanced, but have a smaller degree.
Some new sublocations may also be introduced,
but they are only frozen as required.
The new derivation has a smaller multiset
of unbalanced fixed points, and we can conclude by induction hypothesis.
\end{proof}

Balancing is the most novel part of our focalization process.
This preprocessing is a technical device ensuring
termination in the proof of completeness,
whatever rule permutations are performed.
It should be noted that balancing is often too strong,
and that many focused proofs are indeed not balanced.
For example,
it is possible to obtain unbalanced focused proofs
by introducing an additive conjunction before treating a greatest 
fixed point differently in each branch.

\subsubsection{Focalization graph}

We shall now present the notion of focalization graph
and its main properties~\cite{miller07cslb}.
As we shall see, their adaptation to \mumall{}
is trivial\footnote{
  Note that we do not use the same notations:
  in \cite{miller07cslb}, ${\prec}$ denotes the subformula relation
  while it represents accessibility in the focalization graph in our case.
}.

\begin{definition}
The \emph{synchronous trunk} of a derivation is its largest prefix
containing only applications of synchronous rules.
It is a potentially open subderivation having the same conclusion sequent.
The open sequents of the synchronous trunk (which are conclusions
of asynchronous rules in the full derivation) and its initial sequents
(which are conclusions of $init$, $\llone$ or ${=}$)
are called \emph{leaf sequents} of the trunk.
\end{definition}

\begin{definition}
We define the relation $\prec$ on the formulas of
the base sequent of a derivation $\Pi$:
$P\prec Q$ if and only if
there exists $P'$, asynchronous subformula\footnote{
  This does mean subformula in the locative sense,
  in particular with (co)invariants being subformulas of
  the associated fixed points.
} of $P$,
and $Q'$, synchronous subformula of $Q$,
such that $P'$ and $Q'$ occur in the same
leaf sequent of the synchronous trunk of $\Pi$.
% NOTE that if we say "in the same sequent" and not necessarily "leaf sequent"
% then we have a different (but seemingly appropriate, acyclic) relation:
% it is possible that an async formula is temporarily with a sync one
% but ends up being split apart, e.g. in _+N, P+Q, A*B (N,A,B neg)
% doing first plus, then split tensor and get N,A and P+Q,B
\end{definition}

The intended meaning of $P\prec Q$ is that we must focus on $P$ before $Q$.
Therefore, the natural question is the existence of minimal elements for that 
relation, equivalent to its acyclicity.

\begin{proposition} \label{prop:mini_subform}
If $\Pi$ starts with a synchronous rule,
and $P$ is minimal for $\prec$ in $\Pi$,
then so are its subformulas in their respective subderivations.
\end{proposition}

\begin{proof}
There is nothing to do
if $\Pi$ simply consists of an initial rule.
In all other cases
($\lltens$, $\llplus$, $\exists$ and $\mu$)
let us consider any subderivation $\Pi'$ in which
the minimal element $P$ or one of its subformulas $P'$ occurs
| there will be exactly one such $\Pi'$, except in the case of a tensor
applied on $P$.
The other formulas occurring in the conclusion of $\Pi'$
either occur in the conclusion of $\Pi$ or are subformulas
of the principal formula occurring in it.
This implies that a $Q\prec P$ or $Q\prec P'$ in $\Pi'$
would yield a $Q'\prec P$ in $\Pi$,
which contradicts the minimality hypothesis.
\end{proof}

\begin{lemma}\label{lem:mini}
The relation $\prec$ is acyclic.
\end{lemma}

\begin{proof}
We proceed by induction on the derivation $\Pi$.
If it starts with an asynchronous rule or an initial synchronous rule,
\ie\ its conclusion sequent is a leaf of its synchronous trunk,
acyclicity is obvious since $P\prec Q$ iff $P$ is asynchronous
and $Q$ is synchronous.
If $\Pi$ starts with $\llplus$, $\exists$ or $\mu$,
the relations $\prec$ in $\Pi$ and its subderivation
are isomorphic (only the principal formula changes)
and we conclude by induction hypothesis.
In the case of $\lltens$,
say $\Pi$ derives $\,\vdash\Gamma,\Gamma',P\lltens P'$,
only the principal formula $P\lltens P'$ has subformulas in both premises
$\,\vdash\Gamma,P$ and $\,\vdash\Gamma',P'$.
Hence there cannot be any $\prec$ relation between a formula of $\Gamma$
and one of $\Gamma'$.
In fact, the graph of $\prec$ in the conclusion is obtained by taking
the union of the graphs in the premises
and merging $P$ and $P'$ into $P\lltens P'$.
Suppose, \emph{ab absurdo}, that $\prec$ has cycles in $\Pi$,
and consider a cycle of minimal length.
It cannot involve nodes from both $\Gamma$ and $\Gamma'$:
since only $P\lltens P'$ connects those two components,
the cycle would have to go twice through it,
which contradicts the minimality of the cycle's length.
Hence the cycle must lie within
$(\Gamma,P\lltens P')$ or $(\Gamma',P\lltens P')$
but then there would also be a cycle in the corresponding premise
(obtained by replacing $P\lltens P'$ by its subformula)
which is absurd by induction hypothesis.
\end{proof}

\subsubsection{Permutation lemmas and completeness} %% =====================

We are now ready to describe the transformation of a balanced derivation
into a $\mu$-focused derivation.

\begin{definition}
We define the \emph{reachable locations} of a balanced
derivation $\Pi$, denoted by $|\Pi|$, by taking
the finitely many locations occurring at toplevel in sequents of $\Pi$,
ignoring coinvariance subderivations,
and saturating this set by adding the sublocations of
locations that do not correspond to fixed point expressions.
\end{definition}

It is easy to see that $|\Pi|$ is a finite set.
Hence $|\Pi|$, ordered by strict inclusion, is a well-founded measure
on balanced derivations.

Let us illustrate the role of reachable locations with the following
derivations:
\[ \infer[\nu]{\vdash \nu B \t, a \llpar b, \top}{
   \infer[\top]{\vdash S \t, a\llpar b, \top}{} &
   \infer{\vdash S^\perp\x, BS\x}{\vdots}}
\quad\quad\quad
   \infer[\llpar]{\vdash \nu B \t, a \llpar b, \top}{
   \infer[\top]{\vdash \nu B \t, a, b, \top}{}} \]
For the first derivation, the set of reachable locations is
$\{ \nu B \t, a \llpar b, \top, S\t, a, b \}$.
For the second one, it is $\{ \nu B \t, a\llpar b, \top, a, b \}$.
As we shall see, the focalization process may involve transforming
the first derivation into the second one, thus loosing reachable locations,
but it will never introduce new ones.
In that process, the asynchronous rule $\llpar$ is ``permuted'' under
the $\top$, \ie\ the application of $\top$ is delayed by the insertion
of a new $\llpar$ rule.
This limited kind of proof expansion does not affect reachable locations.
A more subtle case is that of ``permuting'' a fixed point rule under $\top$.
This will never happen for $\mu$. For $\nu$, the permutation
will be guided by the existing reachable locations:
if $\nu$ currently has no reachable sublocation it will be frozen,
otherwise it will be coinducted on,
leaving reachable sublocations unchanged in both cases.
The set of reachable locations is therefore
a skeleton that guides the focusing process,
and a measure which ensures its termination.

\begin{lemma} \label{lem:inst}
For any balanced derivation $\Pi$,
$|\Pi\theta|$ is balanced and $|\Pi\theta|\subseteq|\Pi|$.
\end{lemma}

\begin{proof}
By induction on $\Pi$, following the definition of $\Pi\theta$.
The preservation of balancing and reachable locations is obvious since
the rule applications in $\Pi\theta$ are the same as in $\Pi$,
except for branches that are erased by $\theta$
(which can lead to a strict inclusion of reachable locations).
\end{proof}

\begin{lemma}[Asynchronous permutability] \label{lem:async}
Let $P$ be an asynchronous formula.
If $\;\vdash \Gamma, P$ has a balanced derivation $\Pi$,
then it also has a balanced derivation $\Pi'$ where $P$ is principal in the
conclusion sequent, and such that $|\Pi'|\subseteq|\Pi|$.
\end{lemma}

\begin{proof}
Let $\Pi_0$ be the initial derivation.
We proceed by induction on its subderivations,
transforming them while respecting the balanced use of fixed points
in $\Pi_0$.
If $P$ is already principal in the conclusion, there is nothing to do.
Otherwise, by induction hypothesis
we make $P$ principal in the immediate subderivations where it occurs,
and we shall then permute the first two rules.

If the first rule ${\cal R}$
is $\top$ or a non-unifiable instance of $\neq$, there is no subderivation,
and \emph{a fortiori} no subderivation where $P$ occurs.
In that case we apply an introduction rule for $P$,
followed by ${\cal R}$ in each subderivation.
This is obvious in the case of $\llpar$, $\llwith$, $\forall$, $\bot$,
$\neq$ and $\top$ (note that there may not be any subderivation in the last
two cases, in which case the introduction of $P$ replaces ${\cal R}$).
% DETAIL reviewer #2 missed that we take the nu rule from Pi_0
If $P$ is a greatest fixed point that is coinducted on in $\Pi_0$,
we apply the coinduction rule with the coinvariance premise taken in $\Pi_0$,
followed by ${\cal R}$.
Otherwise, we freeze $P$ and apply ${\cal R}$.
By construction, the resulting derivation is balanced in the same way as
$\Pi_0$, and its reachable locations are contained in $|\Pi_0|$.

In all other cases we permute the introduction of $P$ under the first rule.
The permutations of MALL rules are simple. We shall not detail them, but
note that if $P$ is $\top$ or a non-unifiable $u\neq v$, permuting
its introduction under the first rule erases that rule.
The permutations involving freezing rules are obvious,
and most of the ones involving fixed points, such as ${\lltens}/\nu$,
are not surprising:
\[\infer{\vdash \Gamma,\Gamma', P\lltens P', \nu{}B\t}{
  \infer{\vdash\Gamma,P,\nu{}B\t}{
    \vdash\Gamma,P,S\t &
    \vdash BS\x ,S\x^\bot
  } &
  \vdash\Gamma',P'
  }
\quad \longrightarrow \quad
  \infer{\vdash \Gamma,\Gamma',P\lltens P', \nu{}B\t}{
  \infer{\vdash \Gamma,\Gamma',P\lltens P', S\t}{
    \vdash \Gamma,P,S\t &
    \vdash \Gamma',P'}
  &
    \vdash BS\x, S~\x^\bot
  }
\]
The ${\llwith}/\nu$ and ${\neq}/\nu$ permutations
rely on the fact that the subderivations obtained by induction hypothesis
are balanced in the same way,
with one case for freezing in all additive branches
and one case for coinduction in all branches:
\disp{
\infer{\vdash\Gamma,P\llwith P', \nu{}B\t}{
  \infer{\vdash\Gamma,P,\nu{}B\t}{
    \infer{\vdash\Gamma,P,S\t}{\Pi} &
    \infer{\vdash BS\x, S\x^\bot\vphantom{\t}}{\Theta}
  } &
  \infer{\vdash\Gamma,P',\nu{}B\t}{
    \infer{\vdash\Gamma,P',S\t}{\Pi'} &
    \infer{\vdash BS\x, S\x^\bot\vphantom{\t}}{\Theta}
  }
}
}{
\infer{\vdash\Gamma,P\llwith P', \nu{}B\t}{
  \infer{\vdash\Gamma,P\llwith P', S\t}{
    \infer{\vdash \Gamma, P, S\t}{\Pi} &
    \infer{\vdash \Gamma, P', S\t}{\Pi'}}
  &
  \infer{\vdash BS\x, (S\x)^\bot\vphantom{\t}}{\Theta}}
}
Another non-trivial case is ${\lltens}/{\neq}$ which makes use of
Lemma~\ref{lem:inst}: %
\disp{
  \infer{\vdash \Gamma,\Gamma',P\lltens Q, u\neq v}{
  \infer{\vdash \Gamma,P,u\neq v}{
  \Set{\infer{\vdash (\Gamma,P)\sigma}{\Pi_\sigma}}{\sigma\in csu(u\unif v)}} &
  \infer{\vdash \Gamma',Q}{\Pi'}}
}{
  \infer{\vdash \Gamma,\Gamma',P\lltens Q, u\neq v}{
    \Set{
      \infer{\vdash (\Gamma,\Gamma',P\lltens Q)\sigma}{
          \infer{\vdash (\Gamma,P)\sigma}{\Pi_\sigma} &
          \infer{\vdash (\Gamma',Q)\sigma}{\Pi'\sigma}}
      }{\sigma\in csu(u\unif v)}
  }
} %
A simple inspection shows that in each case,
the resulting derivation is balanced in the same way as $\Pi_0$,
and does not have any new reachable location |
the set of reachable locations may strictly decrease only
upon proof instantiation in ${\lltens}/{\neq}$,
or when permuting $\top$ and trivial instances of $\neq$ under
other rules.
\end{proof}

\begin{lemma}[Synchronous permutability] \label{lem:sync}
Let $\Gamma$ be a sequent of synchronous and frozen formulas.
If $\;\vdash\Gamma$
has a balanced derivation $\Pi$ in which $P$ is minimal for $\prec$
then it also has a balanced derivation $\Pi'$ such that
$P$ is minimal and principal in the conclusion sequent of $\Pi'$,
and $|\Pi'|=|\Pi|$.
\end{lemma}

\begin{proof}
We proceed by induction on the derivation.
If $P$ is already principal, there is nothing to do.
% NOTE there is no initial case, P would be principal
Otherwise, since the first rule must be synchronous,
$P$ occurs in a single subderivation.
We can apply our induction hypothesis on that subderivation:
its conclusion sequent still cannot contain any asynchronous formula by
minimality of $P$ and,
by Proposition~\ref{prop:mini_subform}, $P$ is still minimal in it.
We shall now permute the first two rules, which are both synchronous.
The permutations of synchronous MALL rules are simple.
As for $\llone$, there is no permutation involving $=$.
The permutations for $\mu$ follow the same geometry as those for $\exists$
or $\llplus$. For instance, ${\lltens}/{\mu}$ is as follows:
\disp{
  \infer[\lltens]{\vdash \Gamma, \Gamma', P\lltens P', \mu B \t}{
    \vdash \Gamma,P &
    \infer[\mu]{\vdash \Gamma',P',\mu B\t}{\vdash \Gamma',P', B(\mu B)\t}}
}{
  \infer[\mu]{\vdash \Gamma, \Gamma', P\lltens P', \mu B \t}{
  \infer[\lltens]{\vdash \Gamma, \Gamma', P\lltens P', B(\mu B) \t}{
    \vdash \Gamma,P &
    \vdash \Gamma',P', B(\mu B)\t}}
}
All those permutations preserve $|\Pi|$.
Balancing and minimality are obviously preserved, respectively
because asynchronous rule applications and
the leaf sequents of the synchronous trunk are left unchanged.
\end{proof}

\begin{theorem}
The $\mu$-focused system is sound and complete with respect to \mumall:
If $\;\vdash\Uparrow\Gamma$ is provable, then $\;\vdash\Gamma$
  is provable in \mumall.
If $\;\vdash\Gamma$ has a quasi-finite \mumall\ derivation,
  then $\;\vdash\Uparrow\Gamma$ has a (focused) derivation.
\end{theorem}

\begin{proof}
For soundness, we observe that an unfocused derivation can be obtained
simply from a focused one by erasing focusing annotations
and removing switching rules
($\;\vdash\Delta\Uparrow\Gamma$ gives $\;\vdash\Delta,\Gamma$ and
 $\;\vdash\Gamma\Downarrow P$ gives $\;\vdash \Gamma,P$).
To prove completeness, we first obtain a balanced derivation using
Lemma~\ref{lem:balance}. Then, we use permutation lemmas to reorder rules
in the freezing-annotated derivation so that we
can translate it to a $\mu$-focused derivation.
Formally, we first use an induction on the height of the derivation.
This allows us to assume that coinvariance proofs can be focused,
which will be preserved since those subderivations are left untouched
by the following transformations.
Then, we prove simultaneously the following two statements:
\begin{enumerate}
\item
  If $\;\vdash\Gamma,\Delta$ has a balanced derivation $\Pi$,
  where $\Gamma$ contains only synchronous and frozen formulas,
  then $\;\vdash\Gamma\Uparrow\Delta$ has a derivation.
\item
  If $\vdash\Gamma,P$ has a balanced derivation $\Pi$
  in which $P$ is minimal for ${\prec}$,
  and there is no asynchronous formula in its conclusion,
  then there is a focused derivation of $\vdash\Gamma\Downarrow P$.
\end{enumerate}
We proceed by well-founded induction on $|\Pi|$
with a sub-induction on the number of non-frozen formulas in the
conclusion of $\Pi$.
Note that (1) can rely on (2) for the same $|\Pi|$ but
(2) only relies on strictly smaller instances of (1) and (2).
\begin{enumerate}
\item
If there is any, pick \emph{arbitrarily} an asynchronous formula $P$,
and apply Lemma \ref{lem:async} to make it principal in the first rule.
The subderivations of the obtained proof can be focused,
either by the outer induction in the case of coinvariance proofs,
or by induction hypothesis (1) for the other subderivations:
if the first rule is a freezing, then the reachable locations of the
subderivation and the full derivation are the same, but there is one
less non-frozen formula;
with all other rules, the principal location is consumed
and reachable locations strictly decrease.
Finally, we obtain the full focused derivation by composing those
subderivations using the focused equivalent of the rule applied on $P$.

When there is no asynchronous formula left, we have shown
in Lemma~\ref{lem:mini} that there is a minimal synchronous formula $P$
in $\Gamma,\Delta$.
Let $\Gamma'$ denote $\Gamma,\Delta$ without $P$.
Using switching rules,
we build the derivation of $\vdash\Gamma\Uparrow\Delta$
from $\vdash\Gamma'\Downarrow P$,
the latter derivation being obtained by (2) with $\Pi$ unchanged.

\item
Given such a derivation,
we apply Lemma~\ref{lem:sync} to make the formula $P$ principal.
Each of its subderivations has strictly less reachable locations,
and a conclusion of the form $\;\vdash\Gamma'', P'$
where $P'$ is a subformula of $P$
that is still minimal by Proposition~\ref{prop:mini_subform}.
For each of those we build a focused derivation of
$\;\vdash\Gamma''\Downarrow P'$:
if the subderivation still has no asynchronous formula in its conclusion,
we can apply induction hypothesis (2);
otherwise $P'$ is asynchronous by minimality
and we use the switching rule releasing focus on $P'$,
followed by a derivation of $\vdash\Gamma''\Uparrow P'$
obtained by induction hypothesis (1).
Finally, we build the expected focused derivation from those
subderivations by using the focused equivalent of the
synchronous freezing-annotated rule applied on $P$.
\end{enumerate}
\vspace{-0.6cm}\end{proof}

In addition to a proof of completeness, we have actually defined
a transformation that turns any unfocused proof into a focused one.
This process is in three parts:
first, balancing a quasi-finite unfocused derivation;
then, applying rule permutations on unfocused balanced derivations;
finally, adding focusing annotations to obtain a focused proof.
The core permutation process allows to reorder asynchronous rules
arbitrarily, establishing that, from the proof search viewpoint,
this phase consists of inessential non-determinism as usual,
except for the choice concerning greatest fixed points.

In the absence of fixed points, balancing disappears,
and the core permutation process is known to preserve the essence of
proofs, \ie\ the resulting derivation behaves the same as the original
one with respect to cut elimination.
A natural question is whether our process enjoys the same property.
This is not a trivial question,
because of the merging of coinvariants which is performed during balancing,
and to a smaller extent the unfoldings also performed in that process.
We conjecture that those new transformations, which are essentially
loop fusions and unrolling, do also preserve the
cut elimination behavior of proofs.

A different proof technique for establishing completeness
consists in focusing a proof by cutting it against focused 
identities~\cite{laurent04unp,chaudhuri08jar}.
The preservation of the essence of proofs is thus an immediate
corollary of that method.
However, the merging of coinvariants cannot be performed through
cut elimination, so this proof technique (alone) cannot be used
in our case.

%% %%%%%%%%%%%%%%%%%%%%%%%%%%%%%%%%%%%%%%%%%%%%%%%%%%%%%%%%%%%%%%%%%%%%%%%%%%*

\subsection{The $\nu$-focused system}
\label{sec:nufoc}

While the classification of $\mu$ as synchronous and $\nu$ as
asynchronous is rather satisfying and coincides with several other observations,
that choice does not seem to be forced from the focusing point of view alone.
After all, the $\mu$ rule also commutes with all other rules.
It turns out that one can design a $\nu$-focused system
treating $\mu$ as asynchronous and $\nu$ as synchronous,
and still obtain completeness.
That system is obtained from the previous one by changing only
the rules working on fixed points:
\[ \renewcommand\arraystretch{2}\begin{array}{cp{0.3cm}c}
\infer{\vdash\Gamma\Uparrow\mu{}B\t,\Delta}{
       \vdash\Gamma\Uparrow B(\mu{}B)\t,\Delta}
& &
\infer{\vdash\Gamma\Uparrow\mu{}B\t,\Delta}{
   \vdash\Gamma,(\mu{}B\t)^*\Uparrow\Delta}
\\
\infer{\vdash\Gamma\Downarrow\nu{}B\t}{
       \vdash\Gamma\Downarrow S\t & \vdash\Uparrow BS\x, (S\x)^\bot}
& &
\infer{\vdash(\mu\B\t)^*\Downarrow\nu{}B\t}{}
\end{array} \]

Note that a new asynchronous phase must start in the coinvariance premise:
asynchronous connectives in $BS\x$ or $(S\x)^\perp$ might have to be
introduced before a focus can be picked.
For example, if $B$ is $(\lambda p.~ a^\perp \llpar \bot)$ and
$S$ is $a^\perp$, one cannot focus on $S^\perp$ immediately
since $a^\perp$ is not yet available for applying the $init$;
conversely, if $B$ is $(\lambda p.~ a)$ and $S$ is
$a\lltens\llone$, one cannot focus on $BS$ immediately.

\begin{theorem}
The $\nu$-focused system is sound and complete with respect to \mumall:
If $\;\vdash\Uparrow\Gamma$ is provable, then $\;\vdash\Gamma$
  is provable in \mumall.
If $\;\vdash\Gamma$ has a quasi-finite \mumall\ derivation,
  then $\;\vdash\Uparrow\Gamma$ has a (focused) derivation.
\end{theorem}

\begin{proof}[sketch]
The proof follows the same argument as for the $\mu$-focused system.
We place ourselves in a logic with explicit freezing annotations
for atoms and least fixed points,
and define balanced annotated derivations, requiring
that any instance of a least fixed point is used consistently throughout
a derivation, either always frozen or always unfolded;
together with the constraint on its sublocations, this means that
a least fixed point has to be unfolded the same number of times in
all (additive) branches of a derivation.
We then show that any quasi-finite annotated derivation can be balanced;
the proof of Lemma~\ref{lem:balance} can be adapted easily.
Finally, balanced derivations can be transformed into focused
derivations using permutations: the focalization graph technique
extends trivially, the new asynchronous permutations involving the
$\mu$ rule are simple thanks to balancing, and the new synchronous
permutations involving the $\nu$ rule are trivial.
\end{proof}

This flexibility in the design of a focusing system is unusual.
It is not of the same nature as the arbitrary bias assignment that
can be used in Andreoli's system: atoms are non-canonical, and the bias
can be seen as a way to indicate what is the synchrony of the formula
that a given atom might be instantiated with. But our fixed points
have a fully defined logical meaning, they are canonical.
The flexibility highlights the fact that focusing is a somewhat shallow
property, accounting for local rule permutability
independently of deeper properties such as positivity.
% Locally, least and greatest fixed points are just fixed points.
% In a sense, their difference is a more global aspect, which focusing
% cannot grasp.
%
Although we do not see any practical use of such flexibility,
it is not excluded that one is discovered in the future,
like with the arbitrary bias assignment on atoms in Andreoli's original 
system.

It is not possible to treat both least and greatest fixed points
as asynchronous. Besides creating an unclear situation regarding $init$,
this would require to balance both kinds of fixed points, which is
impossible. In $\mu$-focusing, balancing greatest fixed points unfolds
least fixed points as a side effect, which is harmless since there is
no balancing constraint on those. The situation is symmetric in
$\nu$-focusing. But if both least and greatest fixed points
have to be balanced, the two unfolding processes interfere
and may not terminate anymore.
%\[ \infer{\vdash \bot\llwith\bot, \nu p.p, \mu p.p}{
%      \infer[\bot]{\vdash \bot,\nu p.p,\mu p.p}{
%      \infer[\nu]{\vdash \nu p.p,\mu p.p}{
%      \infer[init]{\vdash \nu p.p,\mu p.p}{} & Id_{\nu p.p}}} &
%      \infer[\bot]{\vdash \bot,\nu p.p,\mu p.p}{
%      \infer[\mu]{\vdash \nu p.p,\mu p.p}{
%      \infer[init]{\vdash \nu p.p,\mu p.p}{}}}} \]
%
It is nevertheless possible to consider mixed bias assignments
for fixed point formulas, if the $init$ rule is restricted accordingly.
We would consider two logically identical variants
of each fixed point: $\mu^+$ and $\nu^+$ being treated synchronously,
$\mu^-$ and $\nu^-$ asynchronously, and the axiom rule would be
restricted to dual fixed points of opposite bias:
\[ \infer{\vdash (\mu B\t)^+, (\nu \B \t)^-}{} \quad
   \infer{\vdash (\nu B\t)^+, (\mu \B \t)^-}{} \]
This restriction allows to perform simultaneously the balancing of
$\nu^-$ and $\mu^-$ without interferences. Further, we conjecture
that a sound and complete focused proof system for that logic would be
obtained by superposing
the $\mu$-focused system for $\mu^+$, $\nu^-$ and the $\nu$-focused
system for $\mu^-$, $\nu^+$.

%% %%%%%%%%%%%%%%%%%%%%%%%%%%%%%%%%%%%%%%%%%%%%%%%%%%%%%%%%%%%%%%%%%%%%%%%%%%*
\subsection{Application to \muLJL} \label{sec:foc_mulj}

The examples of Section~\ref{sec:mumall_examples} showed that despite
its simplicity and linearity, \mumall\ can be related to a more
conventional logic.
In particular we are interested in drawing some connections with
\muLJ~\cite{baelde08phd},
the extension of LJ with least and greatest fixed points.
%We will show that the focusing of \mumall\ derivations
%yields a similar result in the intuitionistic setting.
%A general approach for making such a connection is to first encode
%intuitionistic logic in \mumall{}, focus the derivations of encodings,
%and translate them back to intuitionistic derivations.
%When doing so, it is interesting to minimize the use of exponentials
%in the encoding since these connectives weaken the focusing discipline.
%This is precisely what the extension of the positive/negative
%classification allows.
In the following, we show a simple first step to this program,
in which we capture a rich fragment of \muLJ{}
even though \mumall\ does not have exponentials.
In this section, we make use of the properties of negative formulas
(Definition~\ref{def:connectives}), which has two important consequences:
we shall use the $\mu$-focused system, and could not use the
alternative $\nu$-focused one, since it does not agree with the 
classification;
moreover, we shall work in a fragment of \mumall\ without atoms,
since atoms do not have any polarity.

We have observed (Proposition~\ref{prop:struct}) that structural rules are 
admissible for negative formulas of \mumall.
This property allows us to obtain a 
faithful encoding of a fragment of \muLJ\ in \mumall{}
despite the absence of exponentials.
The encoding must be organized so that formulas appearing on
the left-hand side of intuitionistic sequents can be encoded positively
in \mumall.
The only connectives allowed to appear negatively shall thus be
$\wedge$, $\vee$, $=$, $\mu{}$ and $\exists$.
Moreover, the encoding must commute with negation,
in order to translate the (co)induction rules correctly.
This leaves no choice in the following design.

\begin{definition}[$\H$, $\G$, \muLJL]
The fragments $\H$ and $\G$ are given by the following grammar:
{\allowdisplaybreaks\begin{eqnarray*}
\G &::=& \G\wedge \G \| \G \vee \G \| s=t \| \exists x. \G x \|
         \mu{}(\lambda{}p\lambda\x.\G p \x)\t \| p \t
        \\
  &\|& \forall{}x. \G x \| \H \supset \G \|
       \nu{}(\lambda{}p\lambda\x.\G p \x)\t \\
\H &::=& \H\wedge \H \| \H \vee \H \| s=t \|
         \exists{}x. \H x \|
         \mu{}(\lambda{}p\lambda\x.\H p \x)\t \| p \t
\end{eqnarray*}}
The logic \muLJL\ is the restriction of \muLJ\ to sequents
where all hypotheses are in the fragment $\H$,
and the goal is in the fragment $\G$.
This implies a restriction of induction and coinduction rules to
(co)invariants in $\H$.

Formulas in $\H$ and $\G$ are translated in \mumall\ as follows:
\[  \begin{array}{rcl}
 \enc{P\wedge Q} &\stackrel{def}{=}& \enc{P}\lltens\enc{Q} \\
 \enc{P\vee Q}   &\stackrel{def}{=}& \enc{P}\llplus\enc{Q} \\
 \enc{s=t} &\stackrel{def}{=}& s=t \\
 \enc{\exists x.Px} &\stackrel{def}{=}& \exists x. \enc{Px} \\
 \enc{\mu{}B\t} &\stackrel{def}{=}& \mu\enc{B}\t \\
\end{array} \qquad
\begin{array}{rcl}
 \enc{\forall x.Px} &\stackrel{def}{=}& \forall x. \enc{Px} \\
 \enc{\nu{}B\t} &\stackrel{def}{=}& \nu\enc{B}\t \\
 \enc{P\supset Q} &\stackrel{def}{=}& \enc{P} \llimp \enc{Q} \\
 \enc{\lambda p\lambda\x. B p \x} &\stackrel{def}{=}&
         \lambda p\lambda\x.\enc{B p \x} \\
 \enc{p\t} &\eqdef& p\t
\end{array} \]
\end{definition}

% NOTE muLJL should be closed under cut elimination
%   it is unknown whether coinvariants in muLJ make it more expressive
%   (I guess yes)
% We could discuss this a little bit but it involves talking about
% cut elimination, saying that there is plenty of evidence but no result
% including equality (Clairambault is the closest).

For reference, the rules of \muLJL\ can be obtained simply from
the rules of the focused system presented in Figure~\ref{fig:muLJL},
by translating $\Gamma;\Gamma'\vdash P$ into $\Gamma,\Gamma'\vdash P$,
allowing both contexts to contain any $\H$ formula
and reading them as sets to allow contraction and weakening.

\begin{proposition} \label{prop:01}
Let $P$ be a $\G$ formula, and $\Gamma$ a context of $\H$ formulas.
Then $\Gamma\vdash P$ has a quasi-finite \muLJL\ derivation if and only if
$\;\vdash [\Gamma]^\perp, [P]$ has a quasi-finite \mumall\ derivation,
under the restrictions that (co)invariants
in \mumall\ are of the form $\lambda\x.~ [S\x]$ for $S\x\in\enc{\H}$.
\end{proposition}

\begin{proof}
The proof transformations are simple and compositional.
The induction rule corresponds to the $\nu$ rule for $(\mu{}[B]\t)^\bot$,
the proviso on invariants allowing the translations:
\[ \infer{\Gamma, \mu B \t \vdash G}{
          \Gamma,S\t\vdash G &
          B S \x \vdash S\x}
    \quad \longleftrightarrow \quad
   \infer{\vdash [\Gamma]^\perp, \nu \overline{[B]} \t, [G]}{
          \vdash [\Gamma]^\perp, [S]^\perp\t, [G] &
          \vdash \overline{[B]}[S]^\perp \x, [S]\x}
\]
Here, $[S]$ stands for $\lambda\x.~ [S\x]$,
and the validity of the translation relies on the fact that
$\overline{[B]}[S]^\perp\x$ is the same as $[BS\x]^\perp$.
Note that $BS$ belongs to $\H$ whenever both $S$ and $B$ are in $\H$,
meaning that for any $p$ and $\x$, $Bp\x \in \H$.
The coinduction rule is treated symmetrically,
except that in this case $B$ can be in $\G$:
\[ \infer{\Gamma \vdash \nu B \t}{\Gamma \vdash S \t &
          S\x\vdash BS\x}
   \quad\longleftrightarrow\quad
   \infer{\vdash \enc{\Gamma}^\perp, \nu\enc{B}\t}{
          \vdash \enc{\Gamma}^\perp, \enc{S}\t &
          \vdash \enc{S}^\perp\x, \enc{B}\enc{S}\x} \]
In order to restore the additive behavior of some intuitionistic
rules (\emph{e.g.}, $\wedge{}R$) and translate the structural rules,
we can contract and weaken the negative \mumall\ formulas corresponding
to encodings of $\H$ formulas.
\end{proof}

Linear logic provides an appealing proof theoretic setting because of
its emphasis on dualities and of its clear separation of concepts
(additive \vs{}multiplicative, asynchronous \vs{}synchronous).  Our experience
is that \mumall{} is a good place to study focusing in the presence of
least and greatest fixed point connectives.  To get similar results for
\muLJ, one can either work from scratch entirely
within the intuitionistic framework or use an encoding into linear logic.
Given a mapping from intuitionistic to linear logic, and a complete focused
proof system for linear logic, one can often build a complete
focused proof-system for intuitionistic logic.
\[
\xymatrix
{
 \vdash F \ar@{.>}[d]\ar[r] & \ar[d] \vdash [F] \\
 \vdash \Uparrow F  & \ar[l]\vdash \Uparrow [F] \\
}
\]
The usual encoding of intuitionistic logic into linear logic involves
exponentials, which can damage focusing structures by causing both
synchronous and asynchronous phases to end.
Hence, a careful study of the polarity of linear connectives must
be done (cf. \cite{danos93kgc,liang07csl}) in order to minimize the
role played by the exponentials in such encodings. Here, as a result of
Proposition~\ref{prop:01}, it is possible to get a complete focused
system for \muLJL\ that inherits exactly the strong structure of linear
$\mu$-focused derivations.

This system is presented in Figure~\ref{fig:muLJL}.
Its sequents have the form $\Gamma;\Gamma'\vdash P$ where
$\Gamma'$ is a multiset of synchronous formulas (fragment $\H$)
and the set $\Gamma$ contains frozen least fixed points
in $\H$.
First, notice that accordingly with the absence of exponentials
in the encoding into linear logic, there is no structural rule.
The asynchronous phase takes place on sequents where $\Gamma'$ is not empty.
The synchronous phase processes sequents of the form
$\Gamma ; \vdash P$, where the focus is without any ambiguity on $P$.
It is impossible to introduce any connective on the right when 
$\Gamma'$ is not empty.
As will be visible in the following proof of completeness,
the synchronous phase in \muLJL\ does not correspond exactly to
a synchronous phase in \mumall: it contains rules that are translated
into asynchronous \mumall\ rules, namely implication, universal
quantification and coinduction.
We introduced this simplification in order to simplify the presentation,
which is harmless since there is no choice in refocusing afterwards.

\begin{figure}[htpb]
\begin{center}
\[
\begin{array}{c}
\mbox{Asynchronous phase}
\\[6pt]
\infer{\Gamma;\Gamma', P\wedge Q \vdash R}{
       \Gamma;\Gamma', P, Q \vdash R}
\quad
\infer{\Gamma;\Gamma', P\vee Q \vdash R}{
       \Gamma;\Gamma', P\vdash R &
       \Gamma;\Gamma', Q\vdash R}
\\[6pt]
\infer{
  \Gamma;\Gamma', \exists{}x. P x \vdash Q}{
  \Gamma;\Gamma', P x \vdash Q}
\\[6pt]
\infer{\Gamma;\Gamma',  s = t \vdash P}{
    \{ (\Gamma;\Gamma' \vdash P)\theta :
      \theta\in csu(s\unif t) \} }
\\[6pt]
\infer{\Gamma;\Gamma',\mu{}B\t \vdash P}{
       \Gamma, \mu{}B\t;\Gamma'\vdash P}
\quad
\infer{\Gamma;\Gamma',\mu{}B\t \vdash P}{
  S \in \H &
  \Gamma;\Gamma', S\t \vdash P &
  ;BS\x \vdash S\x
}
\\ ~
\\
\mbox{Synchronous phase}
\\[6pt]
\infer{\Gamma;\vdash A\wedge B}{\Gamma;\vdash A & \Gamma;\vdash B}
\quad
\infer{\Gamma;\vdash A_0\vee A_1}{\Gamma;\vdash A_i}
\quad
\infer{\Gamma;\vdash A\supset B}{\Gamma; A \vdash B}
\\[6pt]
\infer{\Gamma;\vdash t=t}{}
\quad
\infer{\Gamma;\vdash \exists{}x. P x}{\Gamma;\vdash P t}
\quad
\infer{\Gamma;\vdash \forall{}x. P x}{\Gamma;\vdash P x}
\\[6pt]
\infer{\Gamma, \mu{}B\t;\vdash \mu{}B\t}{}
\quad
\infer{\Gamma;\vdash \mu{}B\t}{\Gamma;\vdash B(\mu{}B)\t}
\\[6pt]
\infer{\Gamma;\vdash \nu{}B\t}{
  S \in \H &
  \Gamma;\vdash S\t &
  ;S\x \vdash BS\x
}
\end{array}
\]
\end{center}

\caption{Focused proof system for \muLJL} \label{fig:muLJL}
\end{figure}

\begin{proposition}[Soundness and completeness]
The focused proof system for \muLJL\ is sound and complete
with respect to \muLJL:
any focused \muLJL\ derivation of $\Gamma';\Gamma\vdash P$
can be transformed into a \muLJL\ derivation of $\,\Gamma',\Gamma\vdash P$;
any quasi-finite \muLJL\ derivation of $\,\Gamma\vdash P$
can be transformed into a \muLJL\ derivation of $\cdot\;;\Gamma\vdash P$.
\end{proposition}

\begin{proof}
The soundness part is trivial: unfocused \muLJL\ derivations can be
obtained from focused derivations by removing focusing annotations.
Completeness is established using the translation to linear logic
as outlined above.
Given a \muLJL\ derivation of $\Gamma\vdash P$,
we obtain a \mumall\ derivation of $[\Gamma]\vdash [P]$ using
Proposition~\ref{prop:01}.
This derivation inherits quasi-finiteness,
so we can obtain a $\mu$-focused \mumall\ derivation
of $\vdash \Uparrow [\Gamma]^\perp, [P]$.
All sequents of this derivation correspond to encodings of \muLJL\ sequents,
always containing a formula that corresponds to the right-hand side
of \muLJL\ sequents.
By permutability of asynchronous rules,
we can require that asynchronous rules are applied on right-hand side
formulas only after any other asynchronous rule in our $\mu$-focused derivation.
Finally, we translate that focused derivation into a focused
\muLJL\ derivation.
Let $\Gamma$ be a multiset of least fixed points in $\H$,
$\Gamma'$ be a multiset of $\H$ formulas,
and $P$ be a formula in $\G$.
\begin{enumerate}
\item
  If there is a $\mu$-focused derivation of
  $\vdash ([\Gamma]^\perp)^* \Uparrow [\Gamma']^\perp, [P]$
  or $\vdash ([\Gamma]^\perp)^*,  [P] \Uparrow [\Gamma']^\perp$
  then there is a focused \muLJL\ derivation of
  $\Gamma;\Gamma'\vdash P$.
\item
  If there is a $\mu$-focused derivation of
  $\vdash ([\Gamma]^\perp)^* \Downarrow [P]$
  then there is a focused \muLJL\ derivation of $\Gamma;\vdash P$.
\end{enumerate}
We proceed by a simultaneous induction on the $\mu$-focused derivation.
\begin{enumerate}
\item
  Since $[P]$ is the only formula that may be synchronous,
  the $\mu$-focused derivation can only start with two switching rules:
  either $[P]$ is moved to the left of the arrow, in which case
  we conclude by induction hypothesis (1),
  or $\Gamma'$ is empty and $[P]$ is focused on, in which case
  we conclude by induction hypothesis (2).

  If the $\mu$-focused derivation starts with a logical rule,
  we translate it into a \muLJL\ focused rule before concluding by
  induction hypothesis.
  For instance, the $\llwith$ or $\neq$ rule, which can only be
  applied to a formula in $[\Gamma']^\perp$, respectively correspond
  to a left disjunction or equality rule.
  Other asynchronous \mumall\ rules translate differently depending
  on whether they are applied on $[\Gamma]^\perp$ or $[P]$:
  $\llpar$ can correspond to left conjunction or right implication;
  $\nu$ to left $\mu$ (induction) or right $\nu$ (coinduction);
  $\forall$ to left $\exists$ or right $\forall$.
  Note that in the case where $[P]$ is principal, the constraint on
  the order of asynchronous rules means that $\Gamma$ is empty,
  which is required by synchronous \muLJL\ rule.
  Finally, freezing is translated by the \muLJL\ rule
  moving a least fixed point from $\Gamma'$ to $\Gamma$.

\item
  If the $\mu$-focused derivation starts with the switching rule releasing
  focus from $[P]$ we conclude by induction hypothesis (1).
  Otherwise it is straightforward to translate the first rule and
  conclude by induction hypothesis (2):
  $\lltens$, $\llplus$, ${=}$, $\exists$ and $\mu$
  respectively map to the right rules for
  $\wedge$, $\vee$, ${=}$, $\exists$ and $\mu$.

  Note, however, that the tensor rule splits frozen formulas
  in $([\Gamma]^\perp)^*$, while the right conjunction rule of \muLJL{}
  does not. This is harmless because weakening is obviously admissible
  for the frozen context of \muLJL\ focused derivations.
  This slight mismatch means that we would still have a complete
  focused proof system for \muLJL\ if we enforced a linear use of
  the frozen context. We chose to relax this constraint as it does not
  make a better system for proof search.
\end{enumerate}
\vspace{-0.5cm}\end{proof}

Although \muLJL\ is only a small fragment of \muLJ,
it catches many interesting and useful problems.  For example, 
any Horn-clause specification can be expressed in $\H$ as a least
fixed point, and theorems that state properties such as totality or
functionality of predicates defined in this manner are in $\G$.
Theorems that state more model-checking properties, of the form
$\forall x.~ P~x\supset Q~x$, are in $\G$ provided
that $P$ and $Q$ are in $\H$.
Further, implications can be chained through a greatest fixed point
construction, which allows to specify various relations on
process behaviors.
For example, provided that one-step transitions $u\ra v$ are specified
in $\H$, simulation is naturally expressed in $\G$ as follows:
\[ \nu S \lambda x \lambda y.~
     \forall x'.~ x \ra x' \supset \exists y'.~ y \ra y' \wedge S~x'~y' \]
Finally, the theorems about natural numbers presented in
Section~\ref{sec:mumall_examples} are also in $\G$.
Although a formula in $\G$ can \emph{a priori} be a theorem in \muLJ{}
but not in \muLJL,
we have shown~\cite{baelde09tableaux} that \muLJL{}
is complete for inclusions of non-deterministic finite automata |
${\cal A}\subseteq {\cal B}$ being expressed naturally as
$\forall w.~ [{\cal A}]w \supset [{\cal B}]w$.
% NOTE in LL it uses a with, outside of the encoding of H,
%      but what I said is true if you only look at LJ

Interestingly, the \muLJL\ fragment has already been identified in
LINC~\cite{tiu05eshol} and the Bedwyr system~\cite{baelde07cade}
implements a proof-search strategy for it that is complete for finite
behaviors, \ie\ proofs without (co)induction nor axiom rules,
where a fixed point has to be treated in a finite number of unfoldings.
This strategy coincides with the focused
system for \muLJL, where the finite behavior restriction corresponds
to dropping the freezing rule,
obtaining a system where proof search consists in
eagerly eliminating any left-hand side (asynchronous) formula
before working on the goal (right-hand side),
without ever performing any contraction or weakening.
The logic \muLJ\ is closely related to LINC, the main difference being
the generic quantifier $\nabla$, which allows to specify and reason about
systems involving variable binding,
such as the $\pi$-calculus~\cite{tiu05concur}.
But we have shown~\cite{baelde08lfmtp} that $\nabla$ can be added
in an orthogonal fashion in \muLJ\ (or \mumall)
without affecting focusing results.

%% file: conclu.tex
We have defined \mumall, a minimal and well-structured proof system
featuring fixed points, and established the two main properties for
that logic.
The proof of cut elimination is the first contribution of this paper,
improving on earlier work and contributing to the understanding
of related works.
The second and main contribution is the study and design of focusing
for that logic. This challenging extension of focused proofs
forces us to reflect on the foundations of focusing,
and brought new proof search applications of focusing.
We have shown that \mumall\ is a good logic for the foundational
study of fixed points, but also a rich system that can directly
support interesting applications:
combining observations on admissible structural rules with our
$\mu$-focused system, we were able to derive a focused
proof system for an interesting fragment of \muLJ.

Although carried out in the simple logic \mumall,
this work on fixed points has proved meaningful in richer logics.
We have extended our focusing results to \muLL\ and \muLJ~\cite{baelde08phd},
naturally adapting the designs and proof techniques developed in this paper.
However, focused systems obtained by translating the target
logic into \mumall\ (or \muLL) are often not fully satisfying,
and better systems can be crafted and proved complete from scratch,
using the same techniques as for \mumall, with a stronger form of balancing
that imposes uniform asynchronous choices over all contractions of a formula.

Further work includes various projects relying on
\mumall\ and its extensions, from theory to implementation.
But we shall focus here on important open questions
that are of general interest concerning this formalism.
An obvious first goal would be to strengthen our weak normalization proof
into a strong normalization result.
The relationship between cut elimination and focusing also has to be
explored more; we conjectured that focusing preserves the identity
(cut elimination behavior) of proofs, and that the notion of quasi-finiteness
could be refined so as to be preserved by cut elimination.
Finally, it would be useful to be able to
characterize and control the complexity of normalization,
and consequently the expressiveness of the logic;
here, one could explore different classes of (co)invariants, or
other formulations of (co)induction.
% This is unclear and naive
%For instance, understand under which conditions a given class of invariants
%is sufficient (cf. the restrictions applying to \muLJL).
% It's actually not useful for recursive definitions
%Finally, the current formulation of our system using the functoriality
%transformation should make it
%possible to follow~\cite{matthes98csl} in relying not on a syntactic
%notion of monotonicity, but on a more semantic one, which should
%allow to work on more subtle well-founded definitions
%like that of reducibility.

%% file: main.bbl
\begin{thebibliography}{}

\bibitem[\protect\citeauthoryear{Alves, Fern{\'a}ndez, Florido, and
  Mackie}{Alves et~al\mbox{.}}{2006}]{alves06csl}
{\sc Alves, S.}, {\sc Fern{\'a}ndez, M.}, {\sc Florido, M.}, {\sc and} {\sc
  Mackie, I.} 2006.
\newblock The power of linear functions.
\newblock In {\em CSL 2006: Computer Science Logic}. 119--134.

\bibitem[\protect\citeauthoryear{Andreoli}{Andreoli}{1992}]{andreoli92jlc}
{\sc Andreoli, J.-M.} 1992.
\newblock Logic programming with focusing proofs in linear logic.
\newblock {\em J. of Logic and Computation\/}~{\em 2,\/}~3, 297--347.

\bibitem[\protect\citeauthoryear{Andreoli and Pareschi}{Andreoli and
  Pareschi}{1991}]{andreoli91ngc}
{\sc Andreoli, J.~M.} {\sc and} {\sc Pareschi, R.} 1991.
\newblock Linear objects: Logical processes with built-in inheritance.
\newblock {\em New Generation Computing\/}~{\em 9,\/}~3-4, 445--473.

\bibitem[\protect\citeauthoryear{Apt and van Emden}{Apt and van
  Emden}{1982}]{apt82jacm}
{\sc Apt, K.~R.} {\sc and} {\sc van Emden, M.~H.} 1982.
\newblock Contributions to the theory of logic programming.
\newblock {\em J. of the ACM\/}~{\em 29,\/}~3, 841--862.

\bibitem[\protect\citeauthoryear{Baelde}{Baelde}{2008a}]{baelde08phd}
{\sc Baelde, D.} 2008a.
\newblock A linear approach to the proof-theory of least and greatest fixed
  points.
\newblock Ph.D. thesis, Ecole Polytechnique.

\bibitem[\protect\citeauthoryear{Baelde}{Baelde}{2008b}]{baelde08lfmtp}
{\sc Baelde, D.} 2008b.
\newblock On the expressivity of minimal generic quantification.
\newblock In {\em International Workshop on Logical Frameworks and
  Meta-Languages: Theory and Practice (LFMTP 2008)}, {A.~Abel} {and}
  {C.~Urban}, Eds. Number 228 in ENTCS. 3--19.

\bibitem[\protect\citeauthoryear{Baelde}{Baelde}{2009}]{baelde09tableaux}
{\sc Baelde, D.} 2009.
\newblock On the proof theory of regular fixed points.
\newblock In {\em TABLEAUX 09: Automated Reasoning with Analytic Tableaux and
  Related Methods}, {M.~Giese} {and} {A.~Waller}, Eds. Number 5607 in LNAI.
  93--107.

\bibitem[\protect\citeauthoryear{Baelde, Gacek, Miller, Nadathur, and
  Tiu}{Baelde et~al\mbox{.}}{2007}]{baelde07cade}
{\sc Baelde, D.}, {\sc Gacek, A.}, {\sc Miller, D.}, {\sc Nadathur, G.}, {\sc
  and} {\sc Tiu, A.} 2007.
\newblock The {Bedwyr} system for model checking over syntactic expressions.
\newblock In {\em 21th Conf.\ on Automated Deduction (CADE)}, {F.~Pfenning},
  Ed. Number 4603 in LNAI. Springer, 391--397.

\bibitem[\protect\citeauthoryear{Baelde and Miller}{Baelde and
  Miller}{2007}]{baelde07lpar}
{\sc Baelde, D.} {\sc and} {\sc Miller, D.} 2007.
\newblock Least and greatest fixed points in linear logic.
\newblock In {\em International Conference on Logic for Programming and
  Automated Reasoning (LPAR)}, {N.~Dershowitz} {and} {A.~Voronkov}, Eds. LNCS,
  vol. 4790. 92--106.

\bibitem[\protect\citeauthoryear{Baelde, Miller, and Snow}{Baelde
  et~al\mbox{.}}{2010}]{baelde10ijcar}
{\sc Baelde, D.}, {\sc Miller, D.}, {\sc and} {\sc Snow, Z.} 2010.
\newblock Focused inductive theorem proving.
\newblock In {\em Fifth International Joint Conference on Automated Reasoning},
  {J.~Giesl} {and} {R.~H{\"a}hnle}, Eds. Number 6173 in LNCS. 278--292.

\bibitem[\protect\citeauthoryear{Barendregt}{Barendregt}{1992}]{barendregt97ha%
ndbook}
{\sc Barendregt, H.} 1992.
\newblock Lambda calculus with types.
\newblock In {\em Handbook of Logic in Computer Science}, {S.~Abramsky}, {D.~M.
  Gabbay}, {and} {T.~S.~E. Maibaum}, Eds. Vol.~2. Oxford University Press,
  117--309.

\bibitem[\protect\citeauthoryear{Brotherston}{Brotherston}{2005}]{brotherston0%
5tableaux}
{\sc Brotherston, J.} 2005.
\newblock Cyclic proofs for first-order logic with inductive definitions.
\newblock In {\em Automated Reasoning with Analytic Tableaux and Related
  Methods: Proceedings of TABLEAUX 2005}, {B.~Beckert}, Ed. LNAI, vol. 3702.
  Springer-Verlag, 78--92.

\bibitem[\protect\citeauthoryear{Burroni}{Burroni}{1986}]{burroni86}
{\sc Burroni, A.} 1986.
\newblock {R}{\'e}cursivit{\'e} graphique (1{\`e}re partie) : Cat{\'e}gorie des
  fonctions r{\'e}cursives primitives formelles.
\newblock {\em Cah. Topologie G{\'e}om. Diff{\'e}r. Cat{\'e}goriques\/}~{\em
  27,\/}~1, 49--79.

\bibitem[\protect\citeauthoryear{Chaudhuri and Pfenning}{Chaudhuri and
  Pfenning}{2005}]{chaudhuri05csl}
{\sc Chaudhuri, K.} {\sc and} {\sc Pfenning, F.} 2005.
\newblock Focusing the inverse method for linear logic.
\newblock In {\em CSL 2005: Computer Science Logic}, {C.-H.~L. Ong}, Ed. LNCS,
  vol. 3634. Springer, 200--215.

\bibitem[\protect\citeauthoryear{Chaudhuri, Pfenning, and Price}{Chaudhuri
  et~al\mbox{.}}{2008}]{chaudhuri08jar}
{\sc Chaudhuri, K.}, {\sc Pfenning, F.}, {\sc and} {\sc Price, G.} 2008.
\newblock A logical characterization of forward and backward chaining in the
  inverse method.
\newblock {\em J. of Automated Reasoning\/}~{\em 40,\/}~2-3 (Mar.), 133--177.

\bibitem[\protect\citeauthoryear{Clairambault}{Clairambault}{2009}]{clairambau%
lt09fossacs}
{\sc Clairambault, P.} 2009.
\newblock Least and greatest fixpoints in game semantics.
\newblock In {\em 12th International Conference on the Foundations of Software
  Science and Computational Structures (FOSSACS)}, {L.~de~Alfaro}, Ed. LNCS,
  vol. 5504. Springer, 16--31.

\bibitem[\protect\citeauthoryear{Danos, Joinet, and Schellinx}{Danos
  et~al\mbox{.}}{1993}]{danos93kgc}
{\sc Danos, V.}, {\sc Joinet, J.-B.}, {\sc and} {\sc Schellinx, H.} 1993.
\newblock The structure of exponentials: Uncovering the dynamics of linear
  logic proofs.
\newblock In {\em Kurt G{\"o}del Colloquium}, {G.~Gottlob}, {A.~Leitsch}, {and}
  {D.~Mundici}, Eds. LNCS, vol. 713. Springer, 159--171.

\bibitem[\protect\citeauthoryear{Danos, Joinet, and Schellinx}{Danos
  et~al\mbox{.}}{1995}]{danos93wll}
{\sc Danos, V.}, {\sc Joinet, J.-B.}, {\sc and} {\sc Schellinx, H.} 1995.
\newblock {LKT} and {LKQ}: sequent calculi for second order logic based upon
  dual linear decompositions of classical implication.
\newblock In {\em Advances in Linear Logic}, {J.-Y. Girard}, {Y.~Lafont}, {and}
  {L.~Regnier}, Eds. Number 222 in London Mathematical Society Lecture Note
  Series. Cambridge University Press, 211--224.

\bibitem[\protect\citeauthoryear{Delande and Miller}{Delande and
  Miller}{2008}]{delande08lics}
{\sc Delande, O.} {\sc and} {\sc Miller, D.} 2008.
\newblock A neutral approach to proof and refutation in {MALL}.
\newblock In {\em 23th Symp.\ on Logic in Computer Science}, {F.~Pfenning}, Ed.
  IEEE Computer Society Press, 498--508.

\bibitem[\protect\citeauthoryear{Delande, Miller, and Saurin}{Delande
  et~al\mbox{.}}{2010}]{delande10apal}
{\sc Delande, O.}, {\sc Miller, D.}, {\sc and} {\sc Saurin, A.} 2010.
\newblock Proof and refutation in {MALL} as a game.
\newblock {\em Annals of Pure and Applied Logic\/}~{\em 161,\/}~5 (Feb.),
  654--672.

\bibitem[\protect\citeauthoryear{Girard}{Girard}{1987}]{girard87tcs}
{\sc Girard, J.-Y.} 1987.
\newblock Linear logic.
\newblock {\em Theoretical Computer Science\/}~{\em 50}, 1--102.

\bibitem[\protect\citeauthoryear{Girard}{Girard}{1992}]{girard92mail}
{\sc Girard, J.-Y.} 1992.
\newblock A fixpoint theorem in linear logic.
\newblock An email posting to the mailing list linear@cs.stanford.edu.

\bibitem[\protect\citeauthoryear{Girard}{Girard}{2001}]{girard01mscs}
{\sc Girard, J.-Y.} 2001.
\newblock Locus solum: From the rules of logic to the logic of rules.
\newblock {\em Mathematical Structures in Computer Science\/}~{\em 11,\/}~3
  (June), 301--506.

\bibitem[\protect\citeauthoryear{Hodas and Miller}{Hodas and
  Miller}{1994}]{hodas94ic}
{\sc Hodas, J.} {\sc and} {\sc Miller, D.} 1994.
\newblock Logic programming in a fragment of intuitionistic linear logic.
\newblock {\em Information and Computation\/}~{\em 110,\/}~2, 327--365.

\bibitem[\protect\citeauthoryear{Huet}{Huet}{1975}]{huet75tcs}
{\sc Huet, G.} 1975.
\newblock A unification algorithm for typed $\lambda$-calculus.
\newblock {\em Theoretical Computer Science\/}~{\em 1}, 27--57.

\bibitem[\protect\citeauthoryear{Laurent}{Laurent}{2002}]{laurent02phd}
{\sc Laurent, O.} 2002.
\newblock Etude de la polarisation en logique.
\newblock Ph.D. thesis, {U}niversit\'e {A}ix-{M}arseille~{II}.

\bibitem[\protect\citeauthoryear{Laurent}{Laurent}{2004}]{laurent04unp}
{\sc Laurent, O.} 2004.
\newblock A proof of the focalization property of linear logic.
\newblock Unpublished note.

\bibitem[\protect\citeauthoryear{Laurent, Quatrini, and de~Falco}{Laurent
  et~al\mbox{.}}{2005}]{laurent05apal}
{\sc Laurent, O.}, {\sc Quatrini, M.}, {\sc and} {\sc de~Falco, L.~T.} 2005.
\newblock Polarized and focalized linear and classical proofs.
\newblock {\em Ann. Pure Appl. Logic\/}~{\em 134,\/}~2-3, 217--264.

\bibitem[\protect\citeauthoryear{Liang and Miller}{Liang and
  Miller}{2007}]{liang07csl}
{\sc Liang, C.} {\sc and} {\sc Miller, D.} 2007.
\newblock Focusing and polarization in intuitionistic logic.
\newblock In {\em CSL 2007: Computer Science Logic}, {J.~Duparc} {and} {T.~A.
  Henzinger}, Eds. LNCS, vol. 4646. Springer, 451--465.

\bibitem[\protect\citeauthoryear{Matthes}{Matthes}{1999}]{matthes98csl}
{\sc Matthes, R.} 1999.
\newblock Monotone fixed-point types and strong normalization.
\newblock In {\em CSL 1998: Computer Science Logic}, {G.~Gottlob},
  {E.~Grandjean}, {and} {K.~Seyr}, Eds. LNCS, vol. 1584. Berlin, 298--312.

\bibitem[\protect\citeauthoryear{McDowell and Miller}{McDowell and
  Miller}{2000}]{mcdowell00tcs}
{\sc McDowell, R.} {\sc and} {\sc Miller, D.} 2000.
\newblock Cut-elimination for a logic with definitions and induction.
\newblock {\em Theoretical Computer Science\/}~{\em 232}, 91--119.

\bibitem[\protect\citeauthoryear{Mendler}{Mendler}{1991}]{mendler91apal}
{\sc Mendler, N.~P.} 1991.
\newblock Inductive types and type constraints in the second order lambda
  calculus.
\newblock {\em Annals of Pure and Applied Logic\/}~{\em 51,\/}~1, 159--172.

\bibitem[\protect\citeauthoryear{Miller}{Miller}{1992}]{miller92jsc}
{\sc Miller, D.} 1992.
\newblock Unification under a mixed prefix.
\newblock {\em Journal of Symbolic Computation\/}~{\em 14,\/}~4, 321--358.

\bibitem[\protect\citeauthoryear{Miller}{Miller}{1996}]{miller96tcs}
{\sc Miller, D.} 1996.
\newblock Forum: {A} multiple-conclusion specification logic.
\newblock {\em Theoretical Computer Science\/}~{\em 165,\/}~1 (Sept.),
  201--232.

\bibitem[\protect\citeauthoryear{Miller, Nadathur, Pfenning, and
  Scedrov}{Miller et~al\mbox{.}}{1991}]{miller91apal}
{\sc Miller, D.}, {\sc Nadathur, G.}, {\sc Pfenning, F.}, {\sc and} {\sc
  Scedrov, A.} 1991.
\newblock Uniform proofs as a foundation for logic programming.
\newblock {\em Annals of Pure and Applied Logic\/}~{\em 51}, 125--157.

\bibitem[\protect\citeauthoryear{Miller and Nigam}{Miller and
  Nigam}{2007}]{miller07csla}
{\sc Miller, D.} {\sc and} {\sc Nigam, V.} 2007.
\newblock Incorporating tables into proofs.
\newblock In {\em CSL 2007: Computer Science Logic}, {J.~Duparc} {and} {T.~A.
  Henzinger}, Eds. LNCS, vol. 4646. Springer, 466--480.

\bibitem[\protect\citeauthoryear{Miller and Pimentel}{Miller and
  Pimentel}{2010}]{miller.ep}
{\sc Miller, D.} {\sc and} {\sc Pimentel, E.} 2010.
\newblock A formal framework for specifying sequent calculus proof systems.
\newblock Available from authors' websites.

\bibitem[\protect\citeauthoryear{Miller and Saurin}{Miller and
  Saurin}{2006}]{miller06mfps}
{\sc Miller, D.} {\sc and} {\sc Saurin, A.} 2006.
\newblock A game semantics for proof search: Preliminary results.
\newblock In {\em Proceedings of the Mathematical Foundations of Programming
  Semantics (MFPS05)}. Number 155 in ENTCS. 543--563.

\bibitem[\protect\citeauthoryear{Miller and Saurin}{Miller and
  Saurin}{2007}]{miller07cslb}
{\sc Miller, D.} {\sc and} {\sc Saurin, A.} 2007.
\newblock From proofs to focused proofs: a modular proof of focalization in
  linear logic.
\newblock In {\em CSL 2007: Computer Science Logic}, {J.~Duparc} {and} {T.~A.
  Henzinger}, Eds. LNCS, vol. 4646. Springer, 405--419.

\bibitem[\protect\citeauthoryear{Miller and Tiu}{Miller and
  Tiu}{2005}]{miller05tocl}
{\sc Miller, D.} {\sc and} {\sc Tiu, A.} 2005.
\newblock A proof theory for generic judgments.
\newblock {\em ACM Trans.\ on Computational Logic\/}~{\em 6,\/}~4 (Oct.),
  749--783.

\bibitem[\protect\citeauthoryear{Momigliano and Tiu}{Momigliano and
  Tiu}{2003}]{momigliano03types}
{\sc Momigliano, A.} {\sc and} {\sc Tiu, A.} 2003.
\newblock Induction and co-induction in sequent calculus.
\newblock In {\em Post-proceedings of TYPES 2003}, {M.~Coppo}, {S.~Berardi},
  {and} {F.~Damiani}, Eds. Number 3085 in LNCS. 293--308.

\bibitem[\protect\citeauthoryear{Nigam}{Nigam}{2009}]{nigam09phd}
{\sc Nigam, V.} 2009.
\newblock Exploiting non-canonicity in the sequent calculus.
\newblock Ph.D. thesis, Ecole Polytechnique.

\bibitem[\protect\citeauthoryear{Santocanale}{Santocanale}{2001}]{santocanale0%
1brics}
{\sc Santocanale, L.} 2001.
\newblock A calculus of circular proofs and its categorical semantics.
\newblock BRICS Report Series RS-01-15, BRICS, Dept.\ of Comp.\ Sci., Univ. of
  Aarhus. May.

\bibitem[\protect\citeauthoryear{Schroeder-Heister}{Schroeder-Heister}{1993}]{%
schroeder-Heister93lics}
{\sc Schroeder-Heister, P.} 1993.
\newblock Rules of definitional reflection.
\newblock In {\em Eighth {Annual Symposium on Logic in Computer Science}},
  {M.~Vardi}, Ed. IEEE Computer Society Press, IEEE, 222--232.

\bibitem[\protect\citeauthoryear{Tiu}{Tiu}{2004}]{tiu04phd}
{\sc Tiu, A.} 2004.
\newblock A logical framework for reasoning about logical specifications.
\newblock Ph.D. thesis, Pennsylvania State University.

\bibitem[\protect\citeauthoryear{Tiu}{Tiu}{2005}]{tiu05concur}
{\sc Tiu, A.} 2005.
\newblock Model checking for $\pi$-calculus using proof search.
\newblock In {\em Proceedings of CONCUR'05}, {M.~Abadi} {and} {L.~de~Alfaro},
  Eds. LNCS, vol. 3653. Springer, 36--50.

\bibitem[\protect\citeauthoryear{Tiu and Momigliano}{Tiu and
  Momigliano}{2010}]{tiu10unp}
{\sc Tiu, A.} {\sc and} {\sc Momigliano, A.} 2010.
\newblock Cut elimination for a logic with induction and co-induction.
\newblock {\em CoRR\/}~{\em abs/1009.6171}.

\bibitem[\protect\citeauthoryear{Tiu, Nadathur, and Miller}{Tiu
  et~al\mbox{.}}{2005}]{tiu05eshol}
{\sc Tiu, A.}, {\sc Nadathur, G.}, {\sc and} {\sc Miller, D.} 2005.
\newblock Mixing finite success and finite failure in an automated prover.
\newblock In {\em Empirically Successful Automated Reasoning in Higher-Order
  Logics (ESHOL'05)}. 79--98.

\end{thebibliography}
